\documentclass[times,sort&compress,3p]{elsarticle}

\usepackage[labelfont=bf]{caption}

\usepackage{amsmath,amsfonts,amssymb,amsthm,booktabs,color,epsfig,graphicx,hyperref,url}

\theoremstyle{plain}
\newtheorem{theorem}{Theorem}

\newtheorem{lemma}{Lemma}

\usepackage{subcaption}

\theoremstyle{definition}

\newtheorem{remark}{Remark}

\usepackage{multirow}
\usepackage{xcolor}
\usepackage{lscape}
\usepackage{enumitem}
\usepackage{algorithm,algpseudocode}

\usepackage{xr}
\makeatletter
\newcommand*{\addFileDependency}[1]{
  \typeout{(#1)}
  \@addtofilelist{#1}
  \IfFileExists{#1}{}{\typeout{No file #1.}}
}
\makeatother

\newcommand{\bA}{\mathbf{A}}
\newcommand{\bB}{\mathbf{B}}
\newcommand{\bC}{\mathbf{C}}
\newcommand{\bD}{\mathbf{D}}

\newcommand{\bI}{\mathbf{I}}
\newcommand{\bJ}{\mathbf{J}}

\newcommand{\bM}{\mathbf{M}}

\newcommand{\bP}{\mathbf{P}}
\newcommand{\bQ}{\mathbf{Q}}
\newcommand{\bR}{\mathbf{R}}

\newcommand{\bV}{\mathbf{V}}
\newcommand{\bW}{\mathbf{W}}
\newcommand{\bX}{\mathbf{X}}

\newcommand{\bd}{\mathbf{d}}
\newcommand{\be}{\mathbf{e}}

\newcommand{\bg}{\mathbf{g}}

\newcommand{\bx}{\mathbf{x}}
\newcommand{\by}{\mathbf{y}}

\newcommand{\sC}{\mathcal{C}}

\newcommand{\sF}{\mathcal{F}}
\newcommand{\sG}{\mathcal{G}}
\newcommand{\sH}{\mathcal{H}}

\newcommand{\sL}{\mathcal{L}}
\newcommand{\sM}{\mathcal{M}}

\newcommand{\sP}{\mathcal{P}}
\newcommand{\sQ}{\mathcal{Q}}

\newcommand{\sS}{\mathcal{S}}
\newcommand{\sT}{\mathcal{T}}

\newcommand{\E}{\mathbb{E}}
\newcommand{\Var}{\text{Var}}
\newcommand{\cov}{\text{cov}}
\newcommand{\tp}{\text{T}}

\newcommand{\AMSE}{\text{AMSE}}
\newcommand{\AB}{\text{AB}}

\newcommand{\MSE}{\text{MSE}}

\newcommand{\PM}{\text{PM}}
\newcommand{\SO}{\text{SO}}
\newcommand{\NO}{\text{NO}}
\newcommand{\CO}{\text{CO}}

\newcommand{\bbeta}{\boldsymbol{\beta}}

\newcommand{\bxi}{\boldsymbol{\xi}}
\newcommand{\bmu}{\boldsymbol{\mu}}

\newcommand{\bphi}{\boldsymbol{\phi}}

\newcommand{\bSigma}{\boldsymbol{\Sigma}}

\newcommand{\bPhi}{\boldsymbol{\Phi}}

\newcommand{\bLambda}{\boldsymbol{\Lambda}}

\newcommand*{\rttensor}[1]{\overline{\overline{#1}}}

\newcommand\numberthis{\addtocounter{equation}{1}\tag{\theequation}}

\begin{document}

\begin{frontmatter}

\title{Quadratic inference with dense functional responses}

\author[1]{Pratim Guha Niyogi\corref{mycorrespondingauthor}}
\author[2]{Ping-Shou Zhong}

\address[1]{Department of Biostatistics, Johns Hopkins Bloomberg School of Public Health, Maryland, USA}
\address[2]{Department of Mathematics, Statistics and Computer Science, University of Illinois at Chicago, Illinois, USA}

\cortext[mycorrespondingauthor]{Corresponding author. Email address: \url{pnyogi1@jhmi.edu}}

\begin{abstract}
We address the challenge of estimation in the context of constant linear effect models with
dense functional responses. In this framework, the conditional expectation of the response curve is represented by a linear combination of functional covariates with constant regression parameters. In this paper, we present an alternative solution by employing the quadratic inference approach, a well-established method for analyzing correlated data, to estimate the regression coefficients. Our approach leverages non-parametrically estimated basis functions, eliminating the need for choosing working correlation structures. Furthermore, we demonstrate that our method achieves a parametric $\sqrt{n}$-convergence rate, contingent on an appropriate choice of bandwidth. This convergence is observed when the number of repeated measurements per trajectory exceeds a certain threshold, specifically, when it surpasses $n^{a_{0}}$, with $n$ representing the number of trajectories. Additionally, we establish the asymptotic normality of the resulting estimator. The performance of the proposed method is compared with that of existing methods through extensive simulation studies, where our proposed method outperforms. Real data analysis is also conducted to demonstrate the proposed method.
\end{abstract}

\begin{keyword} 
Constant Linear-Effect Models\sep Functional Principal Component Analysis\sep Quadratic Inference\sep Semi-parametric Functional Regression.
 \MSC[2020] Primary 62G20 \sep
 Secondary 62H99
\end{keyword}

\end{frontmatter}

\section{Introduction}
\label{Chapter2-qif-Sec:introduction}
Longitudinal data analysis (LDA) involves tracking repeated measurements on the same individuals over time, allowing us to study changes over time and identify influencing factors. Unlike cross-sectional studies, which capture only ``between-individual'' responses at a single time point, LDA can capture ``within-individual" changes through repeated measurements. Longitudinal data is often observed in clusters, with each cluster representing measurements from one individual. As data complexity increases with advancing technology, functional data analysis (FDA) has become a vital tool, extending our understanding from finite to infinite dimensions.
In LDA, data are generally observed with noise for measurements at each time-point \citep{diggle2002analysis, hand2017practical}.
Moreover, a few repeated measurements are required in LDA, and the data are observed sparsely with noise. On the other hand, in FDA, data are densely observed as a continuous-time stochastic process without or with noise \citep{zhang2016sparse, li2022multivariate}. 
Often, the sampling plan can have an effect on the performance of the estimation procedures and inference \citep{hall2006properties}. In some situations, data are typically functions by nature and are observed densely over time. 
\citet{chiou2003functional} proposed a class of semi-parametric functional regression models to describe the influence of vector-valued covariates on a sample of the response curve. 
When data collection leads to experimental error, smoothing is performed at closely spaced time-points in order to reduce the effect of noise. 
The current developments of functional regression techniques have been rigorously studied in \citet{fan1999statistical, hall2007methodology,  chen2019functional}.
See a more recent and complete review in \citet{li2022multivariate}.
\par
The generalized estimating equation (GEE) technique proposed by \citet{liang1986longitudinal} has been extensively used in LDA
for estimation of parameters.
Although an efficient technique, the GEE may not be efficient when the correlation matrix is not correctly specified.
Hence, without requiring the estimation of the correlation parameters,
the quadratic inference function (QIF) approach proposed by \citet{qu2000improving} is useful for parameter estimation in longitudinal studies \citep{diggle2002analysis} and cluster randomized trials \citep{turner2017review}.
By representing the inverse of the working correlation matrix in terms of linear combinations of the basis matrices and involving multiple sets of score functions, the QIF approach has improved efficiency over GEE when the working correlation matrix is not specified correctly. 
Although it maintains the same efficiency as in the situation where the working correlation matrix is specified correctly, the QIF method is not independent of the choice of the working correlation matrix. 
A QIF method-based approach to varying-coefficient models for longitudinal data was proposed by \citet{qu2006quadratic}.
The related work of \citet{bai2008partial} is an extension of QIF for the partial linear model. 
An alternative method was presented in \citet{yu2020note}  where each set of score equations was solved separately and their solutions were combined afterward;
thereby providing results on inference for an optimally weighted estimator and extending those insights to the general setting with over-identified estimating equations.
\par
The fundamental limitations that all the above-mentioned powerful techniques suffer from are: (1) all the above methodologies require prior information on the working correlation structure or choose appropriate basis matrices, and (2) the performance of the classical QIF approach is unknown for dense functional data.
Our study is motivated by problems from multiple real data applications that involve dense functional data when information on the working correlation structure is lacking. 
Let us discuss two motivating examples that we will use to illustrate the proposed method in this paper (see Section \ref{Chapter2-qif-Sec:real-data} for details). 
In the \texttt{Beijing2017-data} example, particulate matter (PM) with a diameter of 
less than 2.5 micrometers is collected over different time-points in different 
locations in China. 
Scientists are interested in knowing the linear dependence of the pollution factor $\PM_{2.5}$ with other atmospheric chemicals \citep{liang2015assessing}. 
Figure \ref{Chapter2-qif-Fig:China} (left) in Section \ref{Chapter2-qif-Sec:real-data} pictorially demonstrates the readings of $\PM_{2.5}$ for the given locations over several hourly time-points; therefore, dense functional data analysis can be implemented. 
In another example from a neuroimaging study \texttt{Apnea-data}, scientists are interested in modeling the change of white matter structure among voxels in each region of interest (ROI) of the human brain. \citet{xiong2017brain} investigated white matter structural alterations using diffusion tensor imaging (DTI) in obstructive sleep apnea (OSA) patients. Here, the change of DTI parameters such as fractional anisotropy (FA) is investigated by using {\it constant linear effect models} { (see subsection 2.1 for a formal definition)} with the interaction of the count of lapses obtained from the Psychomotor Vigilance Task and voxel locations as predictors. We applied this model to each ROI and compared the results obtained from all the ROIs. 
\par
We propose a data-driven way to select the working covariance matrix and express the inverse of the covariance function in terms of the empirical eigen-functions of the covariance operator. 
The covariance operator can be estimated as in  \citet{hsing2015theoretical}
and other related methods based on functional principal component analysis (FPCA) as found in \citet{dauxois1982asymptotic,yao2005functional,hall2006properties,hall2007methodology, li2010uniform, zhang2022unified}.
Note that the estimation of the eigen-functions is nonparametric and it introduces some errors in the proposed estimation method. 
In this article, we answer the following question: \textit{while we estimate the eigen-functions nonparametrically from data, 
is the estimation of coefficient vectors in a semi-parametric problem $\sqrt{n}$- consistent in dense functional data, and can we achieve asymptotic normality?} The advantages of our proposed method are the following: First, our method preserves the good properties of the QIF method and is easier to implement since the eigen-functions can be estimated using the existing packages in statistical software such as $\textsf{R}$. Second, under some mild conditions, our proposed estimator can obtain the optimal convergence rate and is asymptotically normally distributed with less variance as compared to the classical QIF methods. Third, asymptotic results show the estimation accuracy of the coefficients in semi-parametric functional model, therefore, making the influence of the dimension reduction step using FPCA redundant. The error in the estimation of the eigen-functions contributes to the error in the estimation of the parameters. Under some mild bandwidth conditions, the above-mentioned error contribution is of the same order of magnitude as an error in parameter estimation if eigen-functions are known in advance.
\par
The rest of the paper is organized as follows. 
In Section \ref{Chapter2-qif-Sec:model}, we introduce the basic concept of QIF along with our proposed method. 
The asymptotic results for the proposed estimator are presented in Section \ref{Chapter2-qif-Sec:asymptotics}.
In Section \ref{Chapter2-qif-Sec:simulation}, we demonstrate the performance for finite samples. We also apply the proposed method to real data-sets in Section \ref{Chapter2-qif-Sec:real-data}. 
We conclude with some remarks in Section \ref{Chapter2-qif-Sec:discussion}. 
All technical proofs and additional tables from simulation results are presented in the Appendix of the article.

\section{Functional response model and estimation procedure}
\label{Chapter2-qif-Sec:model}
\subsection{Basic model}
To analyze longitudinal data, a straightforward application of a generalized linear model \citep{mccullagh1989generalized} for single response variables is not applicable due to the lack of independence between repeated measures. 
To account for the high correlation in the longitudinal data, some special techniques are required. A seminal work by \citet{liang1986longitudinal} proposed the use of generalized linear models for the analysis of longitudinal data. The model we consider in this article is commonly observed in spatial modeling, where associations among variables do not change over the functional domain (see \citet{zhang2021spatial} and references therein); which is termed as {\it constant linear effects model}.  In neuroimaging studies, constant linear effects model is a popular choice for the region of interest analysis (see \citet{ xiong2017brain, penny2011statistical, friston1994statistical, lindquist2008statistical}) due to its easy and practical interpretations of the constant coefficients. In this paper, the variable \textit{time} is used as a functional domain variable. 
\par
Let $y(t)$ be the response variable at time-point $t$ and 
$\bx(t)$ be $p$-dimensional covariates observed at time $t \in \sT$
where without loss of generality we assume $\sT = \left[ 0, 1\right]$, is the spectrum of the time-points. 
The stochastic process $y(t)$ is square-integrable with marginal mean $\E\{y(t)\ | \;\bx(t)\}$ and finite covariance function; 
the regression parameter $\bbeta$ is unknown and is to be efficiently estimated.
Thus, the constant linear effects model with longitudinal data have the following expression.
\begin{equation}
    y(t) = \bx(t)^{\tp}\bbeta + e(t),
    \label{secton21:linearmodel}
\end{equation}
where the stochastic process $y(t)$ is decomposed into two parts: one is the mean function $\mu(t) = \bx(t)^{\tp}\bbeta$ that depends on time-varying covariates and vector-valued coefficient vector $\bbeta$, and the other is the random error part $e(t)$ where $\E\{e(t)\} = 0$ and has finite second-order covariance $R(s,t)$. 
Let $y_{i}(t)$ be i.i.d. copies of the stochastic process and for each individual, the measurements are taken on $m_{i}$ discrete time-points $T_{ij}$
for $j = 1, \cdots, m_{i}; i = 1, \cdots, n$. 
Therefore, at time $T_{ij}$, we observe a $m_{i} \times 1$ response vector $y_{i}(T_{ij})$ and corresponding covariates $\bx_{i}(T_{ij})$ for the $i$-th subject. 
We assume that $m_{i}$'s are all of the same order as $m = n^{a}$ for some $a \geq 0$, thus $m_{i}/m$ is bounded below and above by some constants. 
Functional data are considered to be sparse depending on the choice of $a$ \citep{hall2006properties}.
Data with bounded $m$ or $a = 0$ are called sparse functional data and if $a \geq a_{0}$ where $a_{0}$ is a transition point, are called dense functional data. 
Moreover, the regions $(0, a_{0})$ are sometimes referred to as moderately dense.
Furthermore, we denote $y_{ij}$ and $\bx_{ij}$ as $y_{i}(T_{ij})$ and $\bx_{i}(T_{ij})$ respectively.
Therefore, $(y_{i1}, \cdots, y_{im_{i}})^{\tp}$ and $(\mu_{i1}, \cdots, \mu_{im_{i}})^{\tp}$ are $m_{i}$ component vectors, denoted as $\by_{i}$ and $\bmu_{i}$ respectively. 
In addition, define residuals $\be_{i} = \by_{i}-\bmu_{i}$, $m_i\times p$ design matrices
$\bX_i=(\bx_{ij})_{i=1,j=1}^{m_i,p}$,
and the derivative of $\bmu_i$ with respect to $\bbeta$, denoted as $\dot{\bmu}_i$, is a $m_{i} \times p$ matrix.
For instance, in case of the linear model discussed in (\ref{secton21:linearmodel}), $\dot{\bmu}_{i} = \bX_i$. We keep the notation $\dot{\bmu}_i$ because it could be generalized for other types of responses.
\par
In the classical problem of GEE, we estimate $\bbeta$ by solving the quasi-likelihood equations \citep{liang1986longitudinal}:
\begin{equation}
\label{Chapter2-qif-Eq:gee}
\sum_{i = 1}^{n}\dot{\bmu}_{i}^{\tp}\bV_{i}^{-1}(\by_{i} - \bmu_{i}) = 0.
\end{equation} 
We denote $\bV_{i} = \nu \bA_{i}^{1/2}\bR_{i}(\rho)\bA_{i}^{1/2}$ 
where $\bR_{i}(\rho)$ is the working correlation matrix, $\nu$ is an over-dispersion parameter and  $\bA_{i}$ is a diagonal matrix where entries are marginal variances $\Var(y_{i1})$, $\cdots$, $\Var(y_{im_{i}})$.
In this article, we simply set $\nu = 1$ while the extension to a general $\nu$ is straightforward. 
In practice, the prior knowledge of the working correlation matrix is unknown, 
and the estimation of the coefficient is influenced by its choice. Therefore, \citet{qu2000improving} suggested an expansion of the inverse of the working correlation matrix as $\bR(\rho)^{-1} = \sum_{k = 1}^{\kappa_{0}} a_{k}(\rho) \bM_{k}$ where $\bM_{k}$ are some basis matrices. 
\citet{zhou2012informative} modified linear representation by grouping the basis matrices into an identity matrix and some symmetric basis matrices. 
For example, if the working correlation matrix is exchangeable/compound symmetric, $\bR(\rho)^{-1} = c_{1}\bI_{m} + c_{2}\bJ_{m}$
where $\bI_{m}$ is the $m \times m$ identity matrix and $\bJ_{m}$ is the $m\times m$ matrix such that $0$ is in diagonal and $1$ is in off-diagonal positions. 
On the other hand, for first-order auto-regressive correlation matrix, $\bR(\rho)^{-1} = c_{1}\bI_{m} + c_{2}\bJ_{m}^{(1)} + c_{3}\bJ^{(2)}_{m}$ 
where $\bJ^{(1)}_{m}$ is a matrix with $1$ on the two main off-diagonal positions and $0$ otherwise, 
$\bJ^{(2)}_{m}$ is a matrix such that $1$ is in the corner positions, viz. $(1, 1)$ and $(m, m)$ and $0$ elsewhere. 
Here, $c_{k}$s are real constants that depend on the nuisance parameter $\rho$. 
Therefore, Equation (\ref{Chapter2-qif-Eq:gee}) reduces to the linear combination of the score vectors:
\begin{align}
\label{Chapter2-qif-Eq:g}
\bar{\bg}(\beta) &:= n^{-1}\sum_{i = 1}^{n}\bg_{i}(\beta):= 
\begin{bmatrix}
    \overline{\bg}^{(1)}(\beta)\\
    \vdots\\
    \overline{\bg}^{(\kappa_{0})}(\beta)\\
\end{bmatrix}
=
\begin{bmatrix}
n^{-1}\sum_{i = 1}^{n}\dot{\bmu}^{\tp}_{i} \bA_{i}^{-1/2}\bM_{1}\bA_{i}^{-1/2} (\by_{i} - \bmu_{i})\\
\vdots \\
n^{-1}\sum_{i = 1}^{n}\dot{\bmu}^{\tp}_{i} \bA_{i}^{-1/2}\bM_{\kappa_{0}}\bA_{i}^{-1/2} (\by_{i} - \bmu_{i})\\
\end{bmatrix}, 
\end{align}
where $\overline{\bg}^{(k)}(\bbeta):= n^{-1}\sum_{i=1}^{n}\bg^{(k)}(\bbeta) := n^{-1}\sum_{i = 1}^{n}\dot{\bmu}^{\tp}_{i} \bA_{i}^{-1/2}\bM_{k}\bA_{i}^{-1/2} (\by_{i} - \bmu_{i})$ for $k = 1, \cdots, \kappa_{0}$, each $\overline{\bg}^{(k)}(\bbeta)$ is a $p\times 1$ vector and $\overline{\bg}(\bbeta)$ is $p\kappa_{0} \times 1$ vector after stacking all $\overline{\bg}^{(k)}(\bbeta)$s. 
Due to the higher dimension of $\overline{\bg}$, \citet{qu2000improving} used the generalized method of moments (GMM) \citep{hansen1982large} for which the method of estimation boils down to minimization of the quadratic inference function $\sQ(\bbeta) = n\bar{\bg}(\bbeta)^{\tp}\widehat{\bC}(\bbeta)^{-1}\bar{\bg}(\bbeta)$ where $\widehat{\bC}(\bbeta) = n^{-1}\sum_{i = 1}^{n}\bar{\bg}_{i}(\bbeta)\bar{\bg}_{i}(\bbeta)^{\tp}$ is the sample covariance matrix of Equation (\ref{Chapter2-qif-Eq:g}). 
To obtain the solution of $\bbeta$, the Newton-Raphson method is used which iteratively updates the value of $\bbeta$. 

\subsection{Incorporating eigen-functions in QIF}
\label{Chapter2-qif-Sec:proposal}
Now, due to standard Karhunen-Lo\`eve expansions of $e_{i}(t) = y_{i}(t) - \mu_{i}(t)$ \citep{karhunen1946spektraltheorie,loeve1946functions}
\begin{equation}
\label{Chapter2-qif-Eq:KL-y}
 e_{i}(t) =  \sum_{r=1}^{\infty}\xi_{ir}\phi_{r}(t),
\end{equation}
with probability 1, where uncorrelated random variables $\xi_{ir} = <e_{i}, \phi_{r}>$ have zero mean and variance $\lambda_{r}$ for ordered eigen-values  $\lambda_{r}$ such that $\lambda_{1} \geq \lambda_{2} \geq \cdots \geq 0$ and $\phi_{r}$s are orthonormal eigen-functions such that $\int\phi_{r}(t)\phi_{l}(t)dt = \textbf{1}(r = l)$.
Due to Mercer's theorem \citep{j1909xvi}, a symmetric continuous non-negative definite kernel function $R$ has the representation $R(s,t)= \sum_{r=1}^{\infty}\lambda_{r}\phi_{r}(s)\phi_{r}(t)$ where the sum is absolutely and uniformly continuous. 
We extract the main directions of the variation of the response variables using FPCA. In this situation, we take the first $\kappa_{0}$ terms which provide a good approximation of the infinite sum in Equation (\ref{Chapter2-qif-Eq:KL-y}) by 
considering that the majority of the variations in the data are contained in the subspace spanned by few eigen-functions \citep{chen2019functional}. 
For finite $\kappa_{0} \geq 1$, we therefore consider the rank-$\kappa_{0}$ FPCA model, 
\begin{equation}
    \E\{y(t)|\bx(t)\} \approx \mu(t) + \sum_{r = 1}^{\kappa_{0}}\E\{\xi_{r}|\bx(t)\}\phi_{r}(t) .
\end{equation}
Unlike existing models in \citet{chiou2003functional} where the main interest of the work is to study how the vector covariates influence the whole response curve, 
we study point-wise effect of covariates on functional response
the influence of covariates vectors indirectly on the functional response through the regression of scores on covariates, via., $\E\{\xi_{r}|\bx(t)\}$. 
An analogue of the truncated empirical version of Equation (\ref{Chapter2-qif-Eq:KL-y}) and Mercer's representation can be provided easily, and we discuss the proposed method based on this truncated version. 
Moreover, we discuss how to choose $\kappa_{0}$ in our situation in Sections \ref{Chapter2-qif-Sec:asymptotics}
and \ref{Chapter2-qif-Sec:simulation} in detail.
\par
In this paper, we propose a data-driven way to compute the basis matrices to obtain the approximate inverse of $\bV$ as discussed earlier. 
In this approach, it is enough to find the eigen-functions to construct a GEE.
Let us define
\begin{equation}
\label{Chapter2-qif-Eq:g-propose}
\bar{\bg}(\bbeta) :=  
\begin{bmatrix}
    \overline{\bg}^{(1)}(\beta)\\
    \vdots\\
    \overline{\bg}^{(\kappa_{0})}(\beta)\\
\end{bmatrix}
:=
\begin{bmatrix}
n^{-1}\sum_{i = 1}^{n}\dot{\bmu}^{\tp}_{i} \widehat{\bPhi}_{i1} (\by_{i} - \bmu_{i})\\
\vdots \\
n^{-1}\sum_{i = 1}^{n} \dot{{\bmu}}^{\tp}_{i} \widehat{\bPhi}_{i\kappa_{0}} (\by_{i} - \bmu_{i})\\
\end{bmatrix},
\end{equation}
where, for $k =1, \cdots, \kappa_{0}$,  $\widehat{\bPhi}_{ik} = \left( m_{i}^{-2}\widehat{\phi}_{k}(T_{ij}) \widehat{\phi}_{k}(T_{ij'}) \right)_{j,j' = 1, \cdots, m_{i}}$ is an $m_i\times m_i$ symmetric matrix and $$\overline{\bg}^{(k)}(\bbeta) = n^{-1}\sum_{i=1}^{n}\bg_{i}^{(k)}(\bbeta) = n^{-1}\sum_{i = 1}^{n} \dot{\mu}^{\tp}_{i} \widehat{\bPhi}_{ik} (\by_{i} - \bmu_{i}).$$
Since the dimension of $\overline{\bg}$ in Equation (\ref{Chapter2-qif-Eq:g-propose}) is greater than the number of parameters to estimate, instead of setting $\overline{\bg}$ to zero, we minimize quadratic function by
$\widehat{\bbeta} = \arg\min_{\bbeta} \sQ(\bbeta) ~ \text{where} ~~ \sQ(\bbeta) =n\bar{\bg}(\bbeta)^{\tp}\widehat{\bC}(\bbeta)^{-1}\bar{\bg}(\bbeta)$
where, $\widehat{\bC}(\bbeta) = n^{-1} \sum_{i=1}^{n}\bg_{i}(\bbeta)\bg_{i}(\bbeta)^{\tp}$, where $\bg_{i}(\bbeta) = (\bg^{(1)\tp}_i(\bbeta), \cdots, \bg_i^{(\kappa_{0})\tp}(\bbeta))^{\tp}$
and $\bg^{(k)}_i(\bbeta)=\dot{\bmu}^{\tp}_{i} \widehat{\bPhi}_{ik} (\by_{i} - \bmu_{i})$.  
For the existence of $\widehat{\bC}^{-1}$ we need the additional restriction: $n \geq \kappa_0$ where $\kappa_{0}$ is the number of eigen-functions.
Under the given set-up, by Equation (8) in \citet{qu2000improving} the estimating equation for $\bbeta$ will be
$\dot{\sQ}(\bbeta) \approx 2\dot{\bar{\bg}}(\bbeta)^{\tp}\widehat{\bC}(\bbeta)^{-1}\bar{\bg}(\bbeta)$.
To obtain the solution of the above equation, we use a Newton-like method. 
In practice, the standard Newton method does not lead to a decrease in the objective function, that is, at each step of the iteration, there is no guarantee that $\sQ(\bbeta_{s+1}) < \sQ(\bbeta_{s})$. Therefore, we use the Algorithm \ref{Chapter2-qif-Algo:newton-type}  to estimate $\bbeta$ using the Quasi-Newton method with halving \citep{givens2012computational}.

\begin{algorithm}
\caption{Estimation of $\bbeta$ using the Quasi-Newton method with halving.}
\label{Chapter2-qif-Algo:newton-type}
\begin{algorithmic}
   \State Set $\widehat{\bbeta}_{{1}} {\gets \Tilde{\bbeta}}$ (initial estimates)
   \State Set $\epsilon_{0}$ (threshold, a small number)
   \State Set max.count (maximum number of repetition)
   \State {Set $l \gets 0$}
    \While{$\text{Error} > \epsilon_{0}$ {or $l = 0$}}
   \State Calculate $\dot{\sQ}(\widehat{\bbeta}_{1})$ and $\ddot{\sQ}(\widehat{\bbeta}_{1})$ based on $\widehat{\bbeta}_{1}$ {
    using proposed method}
    
    \State Initialise $r_{0} = 1$
    \State $\widehat{\bbeta}_{2} \gets \widehat{\bbeta}_{1} - r_{0}\ddot{\sQ}(\widehat{\bbeta}_{1})^{-1}\dot{\sQ}(\widehat{\bbeta}_{1})$
    \State Calculate ${\sQ}(\widehat{\bbeta}_{1})$ 
    and ${\sQ}(\widehat{\bbeta}_{2})$ 
    based on $\widehat{\bbeta}_{1}$ and $\widehat{\bbeta}_{2}$ respectively 
    using proposed method
    \While{$\sQ(\widehat{\bbeta}_{2}) > \sQ(\widehat{\bbeta}_{1})$}
    \State $r_{0} \gets r_{0}/2$
    \State $\widehat{\bbeta}_{2} \gets \widehat{\bbeta}_{1} - r_{0}\ddot{\sQ}(\widehat{\bbeta}_{1})^{-1}\dot{\sQ}(\widehat{\bbeta}_{1})$
    \State Calculate $\sQ(\widehat{\bbeta}_{1})$ 
            and $\sQ(\widehat{\bbeta}_{2})$ 
            based on $\widehat{\bbeta}_{1}$ and $\widehat{\bbeta}_{2}$ respectively 
            using proposed method
    \EndWhile
    \State Calculate $\text{Error} = \|\widehat{\bbeta}_{2} - \widehat{\bbeta}_{1}\|^{2}$
    \State $\widehat{\bbeta}_{1} \gets \widehat{\bbeta}_{2}$
    \State {$l \gets l + 1$}
    \EndWhile
    \State Output $\bbeta$
\end{algorithmic}
\end{algorithm}

\subsection{Estimation of eigen-functions}
\label{subChapter2-qif-Sec:estimation-eigen-comp}
Estimation of eigen-functions is an important step in our proposed quadratic inference technique. 
In general, FPCA plays an important role as a dimension-reduction technique in functional data analysis. 
Some important theories on FPCA have been developed in recent years. 
In particular, \citet{hall2006properties} proved various asymptotic expressions for FPCA for densely observed functional data. 
Later, \citet{hall2009theory} showed more common theoretical arguments, including the effect of the gap between eigen-value (a.k.a., spacing) on the property of eigen-value estimators.
In \citet{li2010uniform}, uniform rates of convergence of the mean and covariance functions are given, which are equipped for all possible choices/scenarios of $m_{i}$s. 
\par
Note that the error process $e(t)$ has mean zero, defined on compact set $\sT = [0,1]$ satisfying $\int_{\sT}\E\{e^{2}(t)\}dt < \infty$. 
The functional principal components can be constructed via the covariance function $R(s, t)$ (induces the kernel operator $\sF$) defined as 
$R(s, t) = \E\{e(s)e(t)\}$
which is assumed to be square-integrable. 
An empirical analogue of the spectral decomposition of $R$ can be obtained $\widehat{R}(s,t) = \sum_{r = 1}^{\infty} \widehat{\lambda}_{r}\widehat{\phi}_{r}(s)\widehat{\phi}_{r}(t)
$
where the random variables $\widehat{\lambda}_{1} \geq \widehat{\lambda}_{2} \geq \cdots \geq 0$ are the eigen-values of the estimated operator $\widehat{\sF}$ 
and the corresponding sequence of eigen-functions are $\widehat{\phi}_{1}, \widehat{\phi}_{2},\cdots$. 
Further, assume that $\int_{\sT}\phi_{r}\widehat{\phi}_{r} \geq 0$ to avoid the issue regarding change of sign \citep{hall2006properties} for practical comparison of eigen-functions,
otherwise there is no impact on the convergence rate of eigen-functions and hence the proposed estimators. 
\par
Suppose that $T_{ij}$ are observational points with a positive density function $f_{T}(\cdot)$. 
Assume $m_{i} \geq 2$ and define $N=\sum_{i=1}^{n}N_{i}$ where 
$N_{i} = m_{i}(m_{i}-1)$. 
This approach is based on local linear smoother which is popular in functional data analysis, including \citet{fan1996local, li2010uniform, yao2005functional} among many others. 
Let $K(\cdot)$ be a symmetric probability density function on $[-1,1]$,
which is used as kernel and $h > 0$ be bandwidth, 
thus re-scaled kernel function is defined as $K_{h}(\cdot) = h^{-1}K(\cdot)$.
Therefore, for given $s, t \in \sT$, 
choose $(\widehat{a}_{0}, \widehat{b}_{1}, \widehat{b}_{2})$ be the minimizer of the following equation.
\begin{align*}
    &n^{-1}\sum_{i=1}^{n}{N_{i}^{-1}}\underset{j_{1}\neq j_{2}}{\sum_{j_{1} =1}^{m_{i}}
    \sum_{j_{2} = 1}^{m_{i}}}
    \{e_{i}(T_{ij_{1}})e_{i}(T_{ij_{2}}) 
    - a_{0} 
    - b_{1}( T_{ij_{1}} - s) 
    - b_{2}( T_{ij_{2}} - t)\}^{2}
     \\
     &\qquad 
     \times 
    K_{h}\left({T_{ij_{1}}-s} \right)K_{h}\left({T_{ij_{2}}-t} \right).
    \numberthis
\end{align*}
Thus we estimate $R(s, t) = \E\{e(s)e(t)\}$
using the quantity $\widehat{a}_{0}$, viz.,
$\widehat{R}(s,t) = \widehat{a}_{0}$. The operator $\widehat{\sF}$ is in general positive semi-definite and the estimated eigen-values $\widehat{\lambda}_{r}$ are non-negative; indeed, $\widehat{R}$ is symmetric. In practice, $e_{i}(t)$'s are not observable and this is replaced by $\Tilde{e}_{i}(t) = y_{i}(t) - \bx_{i}^{\tp}(t)\Tilde{\bbeta}$ where $\Tilde{\bbeta}$ is an initial estimate that is consistent to $\bbeta$ but may not be efficient. For example, one can choose the initial estimator as an ordinary least squares of $\bbeta$.

\section{Asymptotic properties}
\label{Chapter2-qif-Sec:asymptotics}
In this section, we study the asymptotic properties of the proposed estimator. 
Let us introduce some notations. 
Assume that $m_{i}$s are all of the same order, viz, $m \equiv m(n)=n^a$ for some $a > 0$. 
Define, 
$d_{n1}(h) = h^{2} + h\overline{m}/m$ and 
$d_{n2}(h) = h^{4} + h^{3}\overline{m}/m + h^{2}\rttensor{m}/m^{2}$
where $\overline{m} = \lim{\sup}_{n\rightarrow \infty}n^{-1}\sum_{i=1}^{n}m/m_{i}$ 
and $\rttensor{m} = \lim{\sup}_{n\rightarrow \infty}n^{-1}\sum_{i=1}^{n}(m/m_{i})^{2}$. 
Denote 
$\delta_{n1}(h) = \left\{d_{n1}(h)\log n/(nh^{2})\right\}^{1/2}$ and
$\delta_{n2}(h) = \left\{d_{n2}(h)\log n/(nh^{4})\right\}^{1/2}$. 
Further, $\nu_{a, b} = \int t^{a}K^{b}(t)dt$. 
Define, $\bW = (\bphi(t_{1})^{\tp}, \cdots, \bphi(t_{m})^{\tp})^{\tp}$ is a matrix of order $m \times \kappa_{0}$ obtaining after stacking all $\bphi_{k}$s and random components $\xi_{i} = (\xi_{i1}, \cdots, \xi_{i\kappa_{0}})^{\tp}$. 
Further, $\bxi$ has mean zero and variance $\bLambda$ which is a diagonal matrix with components $\lambda_{1}, \cdots, \lambda_{\kappa_{0}}$.
The sign $``\lesssim"$ indicates that the left-hand side of the inequality is bounded by the right-hand up to a multiplicative positive constant, i.e. for two {sequence of} positive {real numbers} ${b}_{{n}1}$ and ${b}_{{n}2}$ we define {for large $n$,} ${b}_{{n}1}\lesssim {b}_{{n2}}$ as ${b}_{{n}1} \leq C {b}_{{n}2}$ where $C$ is a positive constant not involving $n$.
The following conditions are needed for further discussion of the asymptotic properties.

\begin{enumerate}[label=(C\arabic*)]
    \item\label{Chapter2-qif-Cond:kernel} Kernel function $K(\cdot)$ is a symmetric density function defined on bounded support $[-1, 1]$.
    \item\label{Chapter2-qif-Cond:density} Density function $f_{T}$ of $T$ is bounded above and away from infinity. Also, the density function is bounded below away from zero. Moreover, $f$ is differentiable and the derivative is continuous. 
    \item\label{Chapter2-qif-Cond:R} $R(s, t)$ is twice differentiable, and all second-order partial derivatives are bounded on $[0, 1]^{2}$.
    \item\label{Chapter2-qif-Cond:mean} $\E\{\sup_{t\in [0,1]} |e(t)|^{\gamma} \} < \infty$ and
    $\E\{\sup_{t \in [0,1]} |\bx_{i}(t)|^{{\gamma}}\} < \infty$ for some $\gamma \in (4, \infty)$.
    \item\label{Chapter2-qif-Cond:h} $h \rightarrow 0$ as $n \rightarrow \infty$ such that
    $d_{n1}^{-1}(\log n/ n)^{1-2/\gamma} \rightarrow 0$ 
    and 
    $d_{n2}^{-1}(\log n/ n)^{1-4/\gamma} \rightarrow 0$ 
    for $\gamma \in (4, \infty)$. 
    \item Condition for eigen-components. 
        \begin{enumerate}
            \item\label{Chapter2-qif-Cond:space} for each $1\leq k < r < \infty$,  
                    ${{\lambda_{k}}}/{|\lambda_{k} - \lambda_{r}|} 
                    \leq 
                    C_{0}{r}/{|k-r|}$ for non-zero finite generic constant $C_{0}$.
            \item\label{Chapter2-qif-Cond:V} For some $\alpha>0$, with the condition $V_{r}\lambda_{r}^{-2}r^{1+\alpha} \rightarrow 0$ as $r \rightarrow \infty$ where $V_{r} = \E\{\int\dot{\mu}(t)\phi_{r}(t)dt\}^{2}$.
            \item[] The above two conditions hold if 
            $\lambda_{r} = r^{-\tau_{1}}\Lambda(r)$ and $V_{r} = r^{-\tau_{2}}\Gamma(r)$ for slowly varying functions $\Lambda$ and $\Gamma$
            where $\tau_{2} > 1 + 2\tau_{1} > 3$ and $\tau = \alpha + \tau_{1}$.
            \item\label{Chapter2-qif-Cond:phi} $\int\phi_{k}^{4}(t)dt$ and $\int\dot{\mu}^{2}(t)\phi_{k}^{2}(t)dt$ are finite
            for all $k \geq 1$. 
        \end{enumerate}
    \item\label{Chapter2-qif-Cond:C} $\widehat{\bC}(\bbeta_{{0}})$ converges almost surely to an invertible matrix $\bC_{0} = \E\{\bg(\bbeta_{0})\bg(\bbeta_{0})^{\tp}\}$
    where $\bg(\bbeta_{0})=(\bg^{(1)\tp}(\bbeta_{0}),\cdots,\bg^{(\kappa_0)\tp}(\bbeta_{0}))^{\tp}$, { where $\bbeta_{0}$ is the true value of $\bbeta$}. 
    \item\label{Chapter2-qif-Cond:h-kappa} 
    Conditions for $h$ and $\kappa_{0}$.
    \begin{enumerate}
        \item If $a > 1/4$, $\kappa_{0} = O(n^{1/(3-\tau)})$ and $n^{-1/4} \lesssim h \lesssim n^{-(a+1)/5}$.
        \item If $0< a \leq 1/4$, $\kappa_{0} = O(n^{4(1+a)/5(3-\tau)})$ and $h \lesssim n^{-1/4}$.
    \end{enumerate}
\end{enumerate}
\begin{remark}
    Condition \ref{Chapter2-qif-Cond:kernel} is commonly used in non-parametric regression. 
    The bound condition for the density function of time-points has the standard Condition \ref{Chapter2-qif-Cond:density} for random design. 
    Similar results can be obtained for a fixed design where the grid-points are pre-fixed according to the design density $\int_{0}^{T_{j}}f(t)dt = j/m$ for $j = 1, \cdots, m$, for $m \geq 1$.
    Furthermore, it is important to note that this approach does not involve the requirement to obtain sample path differentiation when we invoke the estimation of eigenfunctions from \citet{li2010uniform}.
    Therefore the method could be applicable for Brownian motion which has a continuously non-differentiable sample path.  
    To expand in Taylor series, Condition \ref{Chapter2-qif-Cond:R} is required, and is also common in non-parametric regression. 
    Condition \ref{Chapter2-qif-Cond:mean} is required for a uniform bound for certain higher-order expectations to show uniform convergence. This is a similar condition adopted from \citet{li2010uniform}.
    Smoothness conditions in \ref{Chapter2-qif-Cond:h} and \ref{Chapter2-qif-Cond:h-kappa} is common in kernel smoothing and functional data to control bias and variance. 
    The first condition for tuning the parameters mentioned in \ref{Chapter2-qif-Cond:h} is similar to \citet{li2010uniform}.
    The required spacing assumptions for eigen-values in Conditions \ref{Chapter2-qif-Cond:space}  and \ref{Chapter2-qif-Cond:V} are similar as in \citet{hall2009theory}. Condition \ref{Chapter2-qif-Cond:phi} is a trivial assumption that frequently arises in functional data analysis literature.  
    In most of the situations, this condition automatically holds. Using the weak law of large numbers, Condition \ref{Chapter2-qif-Cond:C} holds for large $n$.
    Similar kind of conditions can be invoked such as convexity assumption, i.e., $\lambda_{r} - \lambda_{r+1} \leq \lambda_{r-1} - \lambda_{r}$ for all $r \geq 2$. 
    Condition \ref{Chapter2-qif-Cond:h-kappa} is determined to control the rate of the number of repeated measurements. 
\end{remark}
We present the following theorem to provide the asymptotic expansion and consistency of the proposed estimator for $\widehat{\bbeta}$. 
\begin{theorem}
\label{thm:main-AMSE}
Let $\bbeta_{0}$ be the true value of $\bbeta$. 
Under the Conditions (C1)-(C6), for $k = 1, \cdots, \kappa_{0}$,
we have the asymptotic mean square error for $\overline{\bg}^{(k)}(\bbeta_{0})$ (see Equation \eqref{Chapter2-qif-Eq:g-propose}) as 
\begin{equation}
    \AMSE\{\overline{\bg}^{(k)}(\bbeta_{0})\} = O\left\{ 
        n^{-1} + n^{-1}\kappa_{0}^{3-\tau}R_{n}(h)
    \right\} \qquad \text{almost surely,}
\end{equation}
where $R_{n}(h) = \left\{
            h^{4} + (1/n) + 
            (1/{nmh}) + 
            (1/{n^{2}m^{2}h^{2}}) + 
            (1/{n^{2}m^{4}h^{4}}) + 
            (1/{n^{2}mh}) + 
            (1/{n^{2}m^{3}h^{3}})
            \right\}$.
Moreover, under Condition \ref{Chapter2-qif-Cond:h-kappa}, $\AMSE\{\overline{\bg}^{(k)}(\bbeta_{0})\} = O(n^{-1})$.
Therefore, if in addition, Condition \ref{Chapter2-qif-Cond:C} holds,
as $n \rightarrow \infty$, $\|\widehat{\bbeta} - \bbeta_{0}\| = O(n^{-1/2})$ in probability. 
\end{theorem}
\noindent The following theorem states the results of asymptotic normality. 
\begin{theorem}
\label{thm:main-normality}
Define 
$\sC_{i} = \sum_{k_{1}=1}^{\kappa_{0}}
                  \sum_{k_{2}=1}^{\kappa_{0}}
                  \bPhi_{k_{1}}
                  \bX_{i}\bC_{i,k_{1}, k_{2}}^{-1}\bX_{i}^{\tp}
                  \bPhi_{k_{2}}$, 
where $\bC_{i,k_{1}, k_{2}}^{-1}$ is the $(k_{1}, k_{2})$ block of $\bC_{i, 0}^{-1}$ with 
$\bC_{i,0} = \E\left\{\bg_i(\bbeta_{0})\bg_{i}^{\tp}(\bbeta_{0})\right\}$. 
Assume that the conditions for Theorem \ref{thm:main-AMSE} holds.
Then $\sqrt{n}(\widehat{\bbeta} - \bbeta_{0}) \xrightarrow{d} N(0, \bSigma)$ 
where $\bSigma = \bB^{-1}\bA\bB^{-1}$, $\bA = \lim_{n\rightarrow\infty} n^{-1}\sum_{i=1}^{n}\E\{\bX_{i}^{\tp}\sC_{i}\be_{i}\be_{i}^{\tp}\sC_{i}\bX_{i}\}$
and $\bB = \lim_{n\rightarrow\infty}n^{-1}\sum_{i=1}^{n}\E\{\bX_{i}^{\tp}\sC_{i}\bX_{i}\}$.

\end{theorem}
\begin{remark}
    Here, the selection of the bandwidth only affects the second-order term of the MSE of $\widehat{\bbeta}$ 
    and has no effect on the asymptotic result of normality as long as $h$ satisfies the Conditions \ref{Chapter2-qif-Cond:h} and \ref{Chapter2-qif-Cond:h-kappa} along with some restrictions on $\kappa_{0}$. Moreover, it is important to observe that $\bPhi_{k}$ is normalized by the number of repeated measurements.
\end{remark}
{
\begin{remark}
The proposed method is applicable to both sparse and dense functional data. However, this article focuses on dense functional data as discussed in Section \ref{Chapter2-qif-Sec:introduction}. 
The second part of Theorem \ref{thm:main-AMSE} and Theorem \ref{thm:main-normality}
are derived based on the dense functional response and we assume $m_{i}$ is growing with sample size. This assumption enables us to simplify the asymptotic leading order term of the AMSE expression in Theorem \ref{thm:main-AMSE}. 
\end{remark}
}
\noindent Outline of the proofs of the proposed theorems is discussed in the Appendix with additional technical details. 

\section{Simulation studies}
\label{Chapter2-qif-Sec:simulation}
We conduct extensive numerical studies to compare the finite sample performance of our proposed method to that of \citet{qu2000improving} under different correlation structures. 
\subsection{Simulation set-up}
Consider the normal response model 
\begin{equation}
y_{i}(T_{ij}) = \bx_{i}(T_{ij})^{\tp}\bbeta + e_{i}(T_{ij}).    
\end{equation}
For $p=2$, we set coefficient vectors, $\bbeta = (\beta_{1}, \beta_{2})^{\tp}$ where $\beta_{1} = 1$ and $\beta_{2} = 0.5$. 
Consider the covariates $x_{ik}(t) = \chi_{i1}^{(k)} + \chi_{i2}^{(k)}\sqrt{2}\sin\left(\pi t \right) + \chi_{i3}^{(k)}\sqrt{2}\cos\left(\pi t\right)$ where
the coefficients $\chi_{i1}^{(k)} \sim N(0, (2^{-0.5(k-1)})^{2})$, 
$\chi_{i2}^{(k)} \sim N(0, (0.85 \times 2^{-0.5(k-1)})^{2})$,  $\chi_{i3}^{(k)} \sim N(0, (0.7 \times 2^{-0.5(k-1)})^{2})$ and $\chi_{ij}$s are mutually independent for each trajectories $i$ and each $j$.
In fixed design situations, associated observational times are fixed. 
Sample trajectories are observed at $m = 100$ equidistant time-points $\{ t_{1}, \cdots, t_{m}\}$ on $[0, 1]$. 
Set number of trajectories $n \in \{ 100, 300, 500\}$.
The residual process $e_{i}(t)$ is a smoothed function with mean zero and unknown covariance function, where 
each $e_{i}$ is distributed as $e_{i} = \sum_{k \geq 1}\xi_{i}\phi_{i}$ and $\xi_{k}$s are independent normal random variables with mean zero and respective variances $\lambda_{k}$.
For numerical computation, we truncate the finite series at $k = 3$ in Karhunen-Lo\`eve expansions for Situations (a), (b), and (c) as described below. 
In Situations (d) and (e), the error process is generated from given covariance functions. 
    
\begin{enumerate}[label=(\alph*)]
        \item\label{Chapter2-qif-Sim:bm} Brownian motion. The covariance function for the Brownian motion is $\min(s, t)$, $\lambda_{k} = 4/{\pi^{2}(2k - 1)^{2}}$ and $\phi_{k}(t) = \sqrt{2}\sin(t/\sqrt{\lambda_{k}})$.
        \item\label{Chapter2-qif-Sim:decay} Linear process. Consider the eigen-values be $\lambda_{k} = k^{-2l_{0}}$ and $\phi_{k}(t) = \sqrt{2}\cos(k\pi t)$. We fix $l_{0} \in \{1, 2, 3\}$.
        \item\label{Chapter2-qif-Sim:ou} Ornstein Uhlenbeck (OU) process. 
        For positive constants $\mu_{0}$ and $\rho_{0}$, we have a stochastic differential equation 
        for $e(t)$ as 
        $\partial e(t) =  -\mu_{0}e(t)\partial t + \rho_{0}\partial w(t)$ 
        for the Brownian motion $w(t)$.
        It can be shown that $\cov\{e(t), e(s)\} = c\exp\{-\mu_{0}|t-s|\}$ where $c = \rho_{0}^{2}/2\mu_{0}$. 
        Here we assume $c = 1$. 
        Thus, by solving the integral equation 
        we have $\phi_{k}(t) = A_{k}\cos(\omega_{k}t) + B_{k}\sin(\omega_{k}t)$ 
        and $\lambda_{k} = {2\mu_{0}}/{(\omega_{k}^{2} + \mu_{0}^{2})}$ 
        where $\omega$ is solution of $\cot(\omega) = {(\omega^{2} - \mu_{0}^{2})}/{2\mu_{0}\omega}$. 
        The constants $A_{k}$ and $B_{k}$ are defined as 
        $B_{k} = \mu_{0}A_{k}/\omega_{k}$ 
        where $A_{k} = \sqrt{{2\omega_{k}^{2}}/{(2\mu_{0} + \mu_{0}^{2} + \omega_{k}^{2})}}$. 
        Here $\mu_{0}$ is chosen to be $1$ or $3$.
        
        \item\label{Chapter2-qif-Sim:pe} Power exponential. $R(s, t) = \exp\{(-|s-t|/a_{0})^{b_{0}}\}$ where scale parameter
        $a_{0} = 1$ and shape parameter $b_{0} \in \{1, 2, 5\}$.
        \item\label{Chapter2-qif-Sim:quadratic} Rational quadratic. $R(s,t) = \left\{1+ {(s-t)^{2}}/{a_{0}^{2}}\right\}^{-b_{0}}$
        where scale parameter $a_{0} = 1$ and shape parameter $b_{0} \in \{1, 2, 5\}$.
\end{enumerate}

\subsection{Comparison and evaluation}
For each of the situations, we perform 500 simulation replicates. 
To execute \citet{qu2000improving}'s approach, we construct the scores using basis matrices as described in Example 1 (approximation of the compound symmetric correlation structure, denoted as \texttt{ldaCS} in the tables) 
and Example 2 (for the first-order autoregressive correlation structure, denoted as \texttt{ldaAR} in the tables) in their paper. 
Ordinary least squares estimate (\texttt{init}) is taken as the initial estimate of $\bbeta$ for both ours and \citet{qu2000improving}'s approach. We indicate \texttt{fda-k} as our proposed method with the number of basis functions $\texttt{k}$. {Additionally, we denote \texttt{fda-AIC} as our proposed method where the number of basis functions is determined by the AIC, as suggested by \citet{yao2005functional}.} 
The estimation procedure in the iterative algorithm converges when the square difference between the estimated values of two consecutive steps is bounded by a small number $10^{-10}$ or the maximum number of steps crosses 500, whichever happens earlier.
To make theoretical results and numerical examples consistent, 
we use \texttt{FPCA} function in R which is available in \texttt{fdapace} packages \citep{fdapace}. 
Bandwidths are selected using generalized cross-validation 
and the Epanechnikov kernel $K(x) = 0.75(1-x^{2})_{+}$ is used for estimation where $(a)_{+} = \max(a, 0)$. \par
The means and standard deviations (SD) of the regression coefficients based on 500 simulations are given as summary measures. 
We calculate the standard deviation mentioned in the tables based on 500 estimates from 500 replications that can be viewed as the true standard error. 
Moreover, we also compute absolute bias (defined as $\AB = \sum_{b=1}^{500} |\widehat{\bbeta}_{b} - \bbeta|/500$) and  mean square error (defined as $\MSE = \sum_{b = 1}^{500}(\widehat{\bbeta}_{b} -\bbeta)^{2}/500$)
to compare the performance of estimation, where for $b$-th replication $\widehat{\bbeta}_{b}$ be the estimated value for $\bbeta$,
MSEs are reported in the order of $10^{-2}$. {In the last column of all the tables, we report the average fraction of variance explained (FVE) under different values of $\kappa_{0}$.}
Since our objective is to see the performance of the proposed method, we compare the results for different choices of $\kappa_{0}$. However, in practice, the number of selected eigen-functions plays a critical role in the proposed method. We can choose $\kappa_{0}$ based on a scree plot where we plot the number of the component against the FVE. The elbow of the graph is found and the components to the left are considered as significant. Another possibility is to use AIC to select the number of eigen-functions as we have demonstrated in the numerical results.

\subsection{Simulation results}
Simulation results associated with the Brownian motion are shown in Table \ref{Chapter2-qif-Table:bm}. In this situation, we observe that our approach produces better results in terms of the dispersion measures. Tables \ref{Chapter2-qif-Table:decay1}, \ref{Chapter2-qif-Table:decay2} and \ref{Chapter2-qif-Table:decay3} of Appendix show results for linear processes, our proposed method performs better in situations with a working correlation matrix as AR but is comparable for exchangeable structure for $l_{0}=1,2,3$. 
Moreover, in our proposed method, as $l$ increases all dispersion measures such as MSE decrease. Results based on the OU process are documented in Tables \ref{Chapter2-qif-Table:ou1} and \ref{Chapter2-qif-Table:ou3} in Appendix. Our method outperforms the existing methods for both situations and as $\mu_0$ increases, the MSE decreases. For three different parameter choices of the power exponential and rational quadratic covariance structure, the numerical results are presented in Tables \ref{Chapter2-qif-Table:pe1}, \ref{Chapter2-qif-Table:pe2}, \ref{Chapter2-qif-Table:pe5} and \ref{Chapter2-qif-Table:quad1}, \ref{Chapter2-qif-Table:quad2}, \ref{Chapter2-qif-Table:quad5} in the Appendix, respectively. As before, we observe that our proposed method is finer than the existing ones in all sub-cases; but interestingly, when $b_0$ increases, the MSE decreases for the power exponential, whereas it increases for the rational quadratic covariance structure as expected due to the covariance structure. Overall, we observe that for all the above situations, as sample size increases, the dispersion measures, for example, SD and MSE decrease. It establishes that as sample size increases, the parameter estimates get closer and closer to the true parameters. In each of the above situations, the SDs for the proposed methods decrease as we increase $\kappa_{0}$ and stabilize after some value of $\kappa_{0}$ where the fraction of variance (FVE) is approximately 100\%. In most cases, the estimated $\kappa_{0}$ is close to 3 using AIC approach. However, in the specific scenario described in \ref{Chapter2-qif-Sim:pe} with $a_{0} = b_{0} = 1$, the median $\kappa_{0}$ selected by AIC is 1, 19 and 20 for $n=100, 300, 500$ respectively. We observe that the proposed method, which selects the number of eigenvalues using the AIC approach, performs better than the existing methods by \citet{qu2000improving}. Additionally, the results selected by AIC are generally consistent with those chosen by the scree plot criteria of FVE when FVE is sufficiently large.

\begin{table}[t!]
\centering
\caption{Performance of the estimation procedure where the residuals are generated from Brownian motion.  
Mean of the estimated coefficients, standard deviation, absolute bias, mean square error $(\times 100)$, and FVE in percentage are summarized.}
\begin{tabular}{rrrrrrrrrr}
\\
 Method & \multicolumn{4}{c}{$\beta_{1}$} & \multicolumn{4}{c}{$\beta_{2}$} & FVE \%-age\\ 
 & Mean & SD & AB  & MSE & Mean & SD & AB & MSE & \\
 \\
 &\multicolumn{9}{c}{$n = 100$}\\
 \\
  \texttt{init} & 0.9999 & 0.0373 & 0.0297 & 0.1391 & 0.4995 & 0.0486 & 0.0384 & 0.2354 &  \\ 
  \texttt{ldaAR} & 0.9991 & 0.0331 & 0.0265 & 0.1095 & 0.5004 & 0.0445 & 0.0353 & 0.1972 &  \\ 
  \texttt{ldaCS} & 0.9987 & 0.0316 & 0.0253 & 0.1000 & 0.4997 & 0.0411 & 0.0322 & 0.1685 &  \\ 
  $\texttt{fda-1}$& 0.9998 & 0.0564 & 0.0447 & 0.3180 & 0.5006 & 0.0743 & 0.0587 & 0.5516 & 86.2672 \\ 
  $\texttt{fda-2}$ & 1.0001 & 0.0269 & 0.0213 & 0.0723 & 0.4971 & 0.0362 & 0.0290 & 0.1314 & 96.3746 \\ 
  $\texttt{fda-3}$ & 0.9998 & 0.0231 & 0.0181 & 0.0532 & 0.4978 & 0.0317 & 0.0251 & 0.1010 & 99.9220 \\ 
  $\texttt{fda-4}$ & 1.0004 & 0.0052 & 0.0014 & 0.0028 & 0.4994 & 0.0092 & 0.0022 & 0.0085 & 99.9979 \\ 
  $\texttt{fda-5}$ & 0.9999 & 0.0021 & 0.0008 & 0.0004 & 0.4999 & 0.0051 & 0.0012 & 0.0026 & 100.0000 \\ 
  $\texttt{fda-AIC}$ & 0.9998 & 0.0231 & 0.0181 & 0.0532 & 0.4978 & 0.0317 & 0.0251 & 0.1010 & 99.9220 \\ 
  \\
  &\multicolumn{9}{c}{$n = 300$}\\
  \\
 \texttt{init} & 1.0002 & 0.0200 & 0.0162 & 0.0401 & 0.5003 & 0.0288 & 0.0226 & 0.0825 &  \\ 
  \texttt{ldaAR} & 1.0003 & 0.0184 & 0.0147 & 0.0336 & 0.5000 & 0.0259 & 0.0203 & 0.0670 &  \\ 
  \texttt{ldaCS} & 1.0007 & 0.0170 & 0.0134 & 0.0288 & 0.4995 & 0.0242 & 0.0190 & 0.0583 &  \\ 
  $\texttt{fda-1}$& 1.0002 & 0.0309 & 0.0251 & 0.0955 & 0.5008 & 0.0443 & 0.0350 & 0.1962 & 86.7578 \\ 
  $\texttt{fda-2}$ & 1.0002 & 0.0142 & 0.0114 & 0.0202 & 0.4998 & 0.0213 & 0.0169 & 0.0451 & 96.4747 \\ 
  $\texttt{fda-3}$ & 1.0002 & 0.0122 & 0.0098 & 0.0150 & 0.4992 & 0.0179 & 0.0144 & 0.0321 & 99.9745 \\ 
  $\texttt{fda-4}$ & 1.0002 & 0.0021 & 0.0003 & 0.0004 & 0.4998 & 0.0032 & 0.0005 & 0.0010 & 99.9993 \\ 
  $\texttt{fda-5}$ & 1.0000 & 0.0002 & 0.0001 & 0.0000 & 0.5000 & 0.0003 & 0.0002 & 0.0000 & 100.0000 \\ 
  $\texttt{fda-AIC}$ & 1.0002 & 0.0122 & 0.0098 & 0.0150 & 0.4992 & 0.0179 & 0.0144 & 0.0321 & 99.9745 \\ 
 \\
 &\multicolumn{9}{c}{$n = 500$}\\
 \\
 \texttt{init} & 1.0002 & 0.0148 & 0.0117 & 0.0219 & 0.5006 & 0.0223 & 0.0177 & 0.0497 &  \\ 
  \texttt{ldaAR} & 1.0005 & 0.0138 & 0.0111 & 0.0189 & 0.5000 & 0.0206 & 0.0162 & 0.0422 &  \\ 
  \texttt{ldaCS} & 1.0000 & 0.0128 & 0.0102 & 0.0164 & 0.4992 & 0.0184 & 0.0146 & 0.0340 &  \\ 
  $\texttt{fda-1}$& 1.0007 & 0.0234 & 0.0185 & 0.0545 & 0.5012 & 0.0348 & 0.0277 & 0.1213 & 86.7520 \\ 
  $\texttt{fda-2}$ & 0.9996 & 0.0105 & 0.0083 & 0.0110 & 0.5002 & 0.0157 & 0.0126 & 0.0247 & 96.5174 \\ 
  $\texttt{fda-3}$ & 0.9991 & 0.0091 & 0.0074 & 0.0084 & 0.4999 & 0.0133 & 0.0107 & 0.0177 & 99.9851 \\ 
  $\texttt{fda-4}$ & 1.0000 & 0.0002 & 0.0001 & 0.0000 & 0.5000 & 0.0003 & 0.0002 & 0.0000 & 99.9996 \\ 
  $\texttt{fda-5}$ & 1.0000 & 0.0002 & 0.0001 & 0.0000 & 0.5000 & 0.0003 & 0.0002 & 0.0000 & 100.0000 \\ 
  $\texttt{fda-AIC}$ & 0.9991 & 0.0091 & 0.0074 & 0.0084 & 0.4999 & 0.0133 & 0.0107 & 0.0177 & 99.9851 \\
\end{tabular}
\label{Chapter2-qif-Table:bm}
\end{table}

\section{Real data analysis}
\label{Chapter2-qif-Sec:real-data}
In this section, we apply our proposed method to motivating examples in two different data-sets. 
\subsection{Beijing's $\PM_{2.5}$ pollution study}
In the atmosphere, suspended microscopic particles of solid and liquid matter are commonly known as particulates or particulate matter (PM). Such particulates often have a strong noxious impact toward human health, climate, and visibility. 
One such common and fine type of atmospheric particle is $\PM_{2.5}$ with a diameter less than $2.5$ micrometers. 
Many developed and developing cities across the world are experiencing chronic air pollution, with major pollutants being $\PM_{2.5}$; Beijing and a substantial part of China are among such places. 
Some studies show that there are many non-ignorable sources of variability in the distribution and transmission pattern of $\PM_{2.5}$, which are confounded with secondary chemical generation. 
The atmospheric $\PM_{2.5}$ data used in our analysis were collected from the UCI machine learning repository \url{https://archive.ics.uci.edu/ml/datasets/Beijing+Multi-Site+Air-Quality+Data} \citep{liang2015assessing}. 
The data-set includes daily measurements of $\PM_{2.5}$ and associated covariates at twelve different locations in China, namely Aotizhongxin, 
Changping,
Dingling,
Dongsi,
Guanyuan,
Gucheng,
Huairou,
Nongzhanguan,
Shunyi,
Tiantan,
Wanliu, and
Wanshouxigong  
during January 2017. 
After excluding missing data, there were 608 
hourly data points in \texttt{Beijing2017-data}. 
We assume that the atmospheric measurements are independent since they are located quite apart. The objective of our analysis is to describe the trend of functional response $\PM_{2.5}$ (as shown in Figure \ref{Chapter2-qif-Fig:China} (left)) and evaluate the effect of covariates including the chemical compounds such as sulfur dioxide $(\SO_{2})$, nitrogen dioxide $(\NO_{2})$, carbon monoxide $(\CO)$ and ozone $(\text{O}_{3})$ over time. We smooth the covariates and responses to reduce variability and center them. Subsequently, we consider the following model
\begin{equation}
    Y_{i}(t) = \beta_{0} + \SO_{2}(t)\beta_{1} + \NO_{2}(t)\beta_{2} + \CO(t)\beta_{3} + \text{O}_{3}(t)\beta_{4} + e_{i}(t).
\end{equation}
\par
We use Algorithm \ref{Chapter2-qif-Algo:newton-type} to estimate the coefficients of the regression model mentioned above. Through the simulation results, we observe that if the values of $\kappa_0$ increase, the standard deviation of the coefficients decreases. For small FVE such as $50\%$, the corresponding $\kappa_0 = 1$ and the estimation procedure performs poorly; whereas for large FVE percentages, the estimation procedure has adequately improved in terms of standard error. 
The estimated values for $\beta_0,\beta_1,\beta_2,\beta_3,$ and $\beta_4 $ produce similar results across different choices of $\kappa_0$. From the estimated standard error, using scree-plot of standard error, we conclude that the suitable choice of $\kappa_0$ is approximately $10$. 
The estimated scaled eigen-values are provided in Figure \ref{Chapter2-qif-Fig:China} (right) which clearly shows their decay rate. The estimated coefficients with standard error are 
0.0009 (1.1644), 0.0829 (0.2584), 0.9503 (0.1586), 0.0196 (0.0037) and 1.1523 (0.1198) respectively.

\begin{figure}[htbp]
    \centering
    \begin{subfigure}[t]{0.55\textwidth}
        \centering
        \includegraphics[width=\textwidth, height=0.6\textwidth]{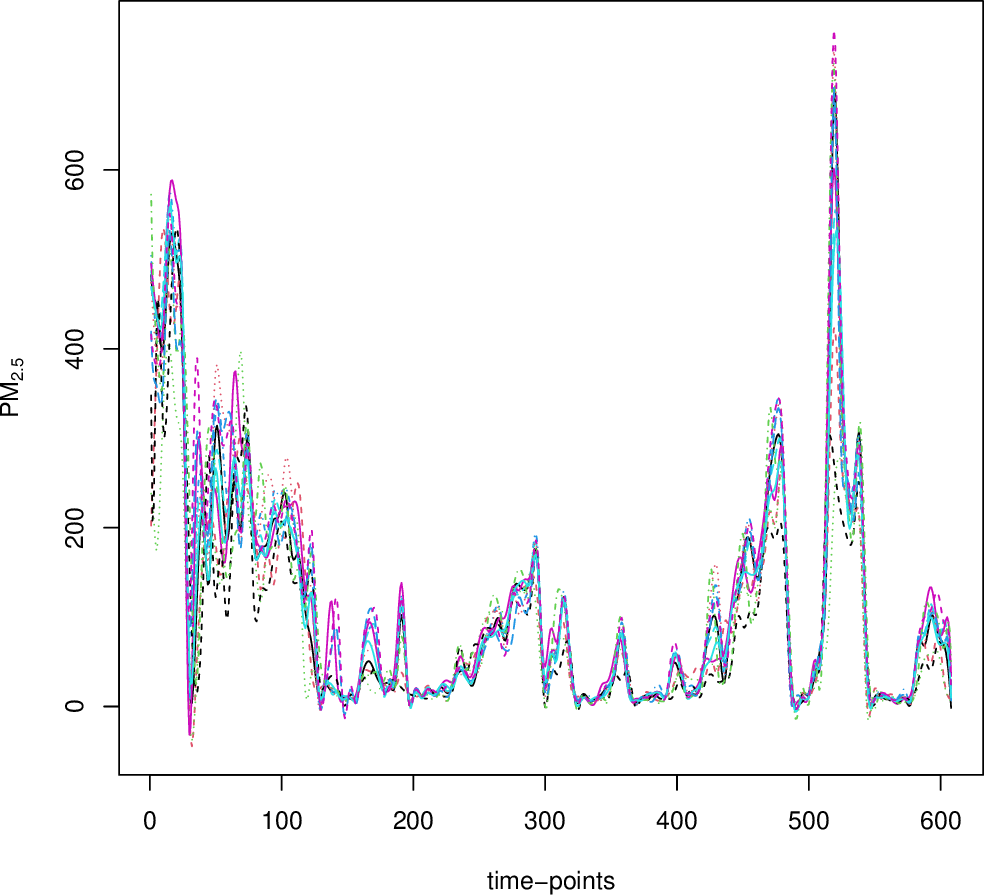}
        \caption{}
    \end{subfigure}%
    \qquad
    \begin{subfigure}[t]{0.33\textwidth}
        \centering
        \includegraphics[width=\textwidth, height=0.8\textwidth]{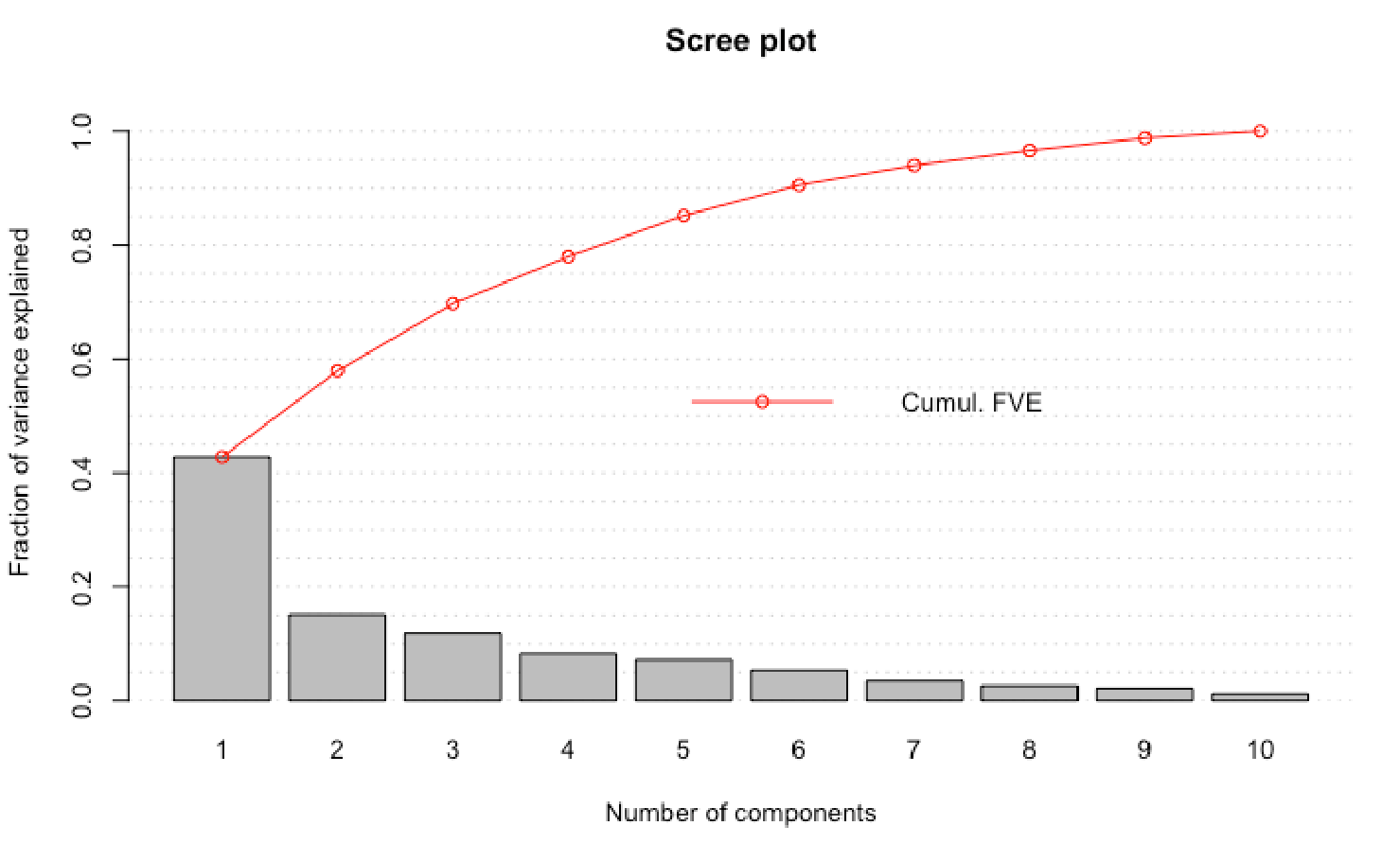}
        \caption{}
    \end{subfigure}
    \caption{\texttt{Beijing2017-data}: (left) Reading of hourly $\PM_{2.5}$ measures for twelve different locations over 608 hourly time-points during January 2017. (right) Scree plots of a fraction of variance explained (FVE).}
    \label{Chapter2-qif-Fig:China}
\end{figure}

\subsection{DTI study for sleep apnea patients}
\label{Chapter2-qif-Sec:real-data-dti}
The diffusion tensor imaging data (DTI) to understand the white matter structural alteration using diffusion tensor imaging (DTI)  are used to illustrate the application of the proposed method and its estimation procedure \citep{xiong2017brain}. MRI is a powerful technique for investigating the structural and functional changes in the brain during pathological and neuro-psychological processes. Due to the advancement in diffusion tensor imaging (DTI), several studies on white matter alterations associated with clinical variables can be found in recent literature.
For our analysis, we use \texttt{Apnea-data} obtained from one such study on obstructive sleep apnea (OSA) patients \citep{xiong2017brain}.
The data consists of 29 male patients between the ages 30-55 years who underwent a study for the diagnosis of continuous positive airway pressure (CPAP) therapy. 
Among those who have sleep disorders other than OSA, night-shift workers, patients with psychiatric disorders, hypertension, diabetes, and other neurological disorders were excluded. 
In this study, the psychomotor vigilance task (PVT) was performed in which a light was randomly switched-on on a screen for several seconds in a certain interval of time, and subjects were asked to press a button as soon as they saw the light appear on screen; such an experiment provides a numerical measure of sleepiness by counting the number of ``lapses'' for each individual.
Psychomotor vigilance task was performed after continuous positive airway pressure treatment which counts the number of lapses in the attention tracking test. In such a test, subjects were asked to press a button as soon as the light appeared on the screen which was turned on randomly for several seconds in a certain interval of time. The interesting measure is  sustained attention and provides a numerical measure of sleepiness by counting the number of ``lapses'' for each individual
\par
DTI was performed at 3T MRI scanner using a commercial 32-channel head coil, followed by the analysis using tract-based spatial statistics to investigate the difference in fractional anisotropy (FA) and other DTI parameters. The image acquisition is as follows. 
An axial T1-weighted image of the brain (3D-BRAVO) is collected with repetition time (TR) = 12ms, echo time (TE) = 5,2ms, flip angle = $13^{\circ}$, inversion time = 450 ms, matrix = $384\times 256$, voxel size = $1.2\times 0.57\times 0.69$mm and scan time = 2 min 54 sec. DTI are obtained in the axial plane using a spin-echo echo planner imaging sequence with TR = 4500ms, TE = 89.4ms, field of view = $20\times 20$cm$^{2}$, matrix size = $160\times 132$, slice thickness = 3mm, slice spacing = 1mm, b-values = 0, 1000 s/mm$^{2}$.
\par
Our objective is to investigate the structural alteration of white matter using DTI in patients with OSA over each voxel at various regions of the brain (called ROIs). 
Thus, our response variable is one of the DTI parameters, viz., fractional anisotropy (FA) and we are interested in studying the effect of the changes of FA over continuous domains such as voxels with the interaction of the lapses and the voxel locations in each ROIs. 
We consider the following model for each ROI.
\begin{equation}
\label{Chapter2-qif-Eq:apnea}
    \text{FA}_{i}(s) = \beta_{0} + \beta_{1}\text{lapses}_{i} \times s + e_{i}(s)
\end{equation}
where $s \in \sS$, a set of voxels in the considered ROIs. Using the Algorithm \ref{Chapter2-qif-Algo:newton-type}, we estimate the coefficients $\beta_{1}$ and $\beta_{2}$ as mentioned in the model \ref{Chapter2-qif-Eq:apnea} and the results are presented in Table \ref{Chapter2-qif-Table:real-data-apnea}. We find that the coefficient estimates are close enough to their initial estimates and the estimated standard error is smaller for the coefficients based on the proposed method. Here $\kappa_{0}$ (i.e., the number of eigen-functions) is determined for simplicity using FVE, which is fixed at 0.99, resulting in $\kappa_{0} = 7$.

\begin{table}[t!]
\caption{
Estimated values and associated standard errors for the regression coefficients are provided upto four decimal places based on the existing and proposed methods. First line corresponding to each ROI shows results based on initial estimates and the second line corresponds to that of proposed estimates.}
\centering
\begin{tabular}{rrrrrr}
  \\
 && \multicolumn{2}{c}{$\beta_{0}$} & \multicolumn{2}{c}{$\beta_{1}$}\\
 region & \# functional points & Estimate & Std. Error $(\times 100)$ & Estimate $(\times 100)$ & Std. Error $(\times 100)$ \\ 
  \\
  ROI.6 & 659  & 0.4512 & 0.1343 & -0.0606 & 0.0130 \\ 
         &  & 0.4512 & 0.0983 & -0.0605 & 0.0028 \\ \\
  ROI.7 & 1362  & 0.5048 & 0.0628 & 0.0309 & 0.0061 \\ 
         &  & 0.5050 & 0.0681 & 0.0342 & 0.0007 \\ \\
  ROI.8 & 1370  & 0.5256 & 0.0586 & -0.0667 & 0.0057 \\ 
         &  & 0.5271 & 0.0346 & -0.0733 & 0.0006 \\ \\
  ROI.9 & 690  & 0.4951 & 0.0910 & 0.2904 & 0.0088 \\ 
         &  & 0.5443 & 0.0874 & 0.1660 & 0.0014 \\ \\
  ROI.10 & 699  & 0.4951 & 0.0892 & 0.3314 & 0.0086 \\ 
         &  & 0.5262 & 0.1398 & 0.4231 & 0.0014 \\ \\
  ROI.11 & 968  & 0.4372 & 0.0979 & 0.1323 & 0.0095 \\ 
         &  & 0.4380 & 0.0637 & 0.1311 & 0.0009 \\ \\
  ROI.12 & 968  & 0.4529 & 0.0948 & 0.0965 & 0.0092 \\ 
         &  & 0.4664 & 0.0750 & 0.0504 & 0.0013 \\ \\
  ROI.13 & 992  & 0.5448 & 0.1060 & 0.3453 & 0.0103 \\ 
         &  & 0.5449 & 0.0856 & 0.3559 & 0.0011 \\ \\
  ROI.14 & 992  & 0.5435 & 0.1068 & 0.3432 & 0.0104 \\ 
         &  & 0.5436 & 0.0754 & 0.3437 & 0.0003 \\ \\
  ROI.37 & 1236  & 0.3695 & 0.0779 & -0.1126 & 0.0076 \\ 
         &  & 0.3713 & 0.0669 & -0.1175 & 0.0017 \\ \\
  ROI.38 & 1155  & 0.3564 & 0.0819 & -0.1356 & 0.0079 \\ 
         &  & 0.3578 & 0.0420 & -0.1356 & 0.0009 \\ \\
  ROI.39 & 1124  & 0.4618 & 0.0760 & 0.1972 & 0.0074 \\ 
         &  & 0.4621 & 0.0615 & 0.1996 & 0.0007 \\ \\
  ROI.40 & 1125  & 0.4786 & 0.0658 & 0.0953 & 0.0064 \\ 
         &  & 0.4780 & 0.0369 & 0.1016 & 0.0005 \\ \\
  ROI.45 & 380  & 0.4189 & 0.1071 & 0.1647 & 0.0104 \\ 
         &  & 0.4190 & 0.0175 & 0.1648 & 0.0001 \\ \\
  ROI.46 & 376  & 0.4074 & 0.1033 & 0.1988 & 0.0100 \\ 
         &  & 0.4074 & 0.0159 & 0.1994 & 0.0002 \\ \\
  ROI.47 & 596  & 0.4596 & 0.0932 & 0.1304 & 0.0090 \\ 
         &  & 0.4594 & 0.0191 & 0.1349 & 0.0001 \\ \\
  ROI.48 & 600  & 0.4045 & 0.0868 & 0.1100 & 0.0084 \\ 
         &  & 0.4036 & 0.0644 & 0.1067 & 0.0006 \\
  \\
\end{tabular}
\label{Chapter2-qif-Table:real-data-apnea}
\end{table}

\section{Discussion}
\label{Chapter2-qif-Sec:discussion}
In this article, we propose an estimation procedure for the constant linear effects model, which is commonly used in statistics \citep{zhang2021spatial} especially in spatial modeling. 
One of the key factors of this estimation procedure is the fact that it is based on the quadratic inference methodology that has served a huge role in the analysis of correlated data since it was discovered by \citet{qu2000improving}. 
In contrast with the existing method, our approach allows the number of repeated measurements to grow with sample size; therefore, the trajectories of individuals can be observed on a dense grid of a continuum domain. 
Instead of assuming a working correlation structure, we propose a data-driven way by estimating the eigen-functions that are obtained by functional principal component analysis. Here, we achieve $\sqrt{n}-$consistency of the parametric estimates in the regression model, even though the eigen-functions are estimated non-parametrically.
\par
Additionally, our method is easy to implement in a wide range of applications. The applicability of the proposed method is illustrated by extensive simulation studies. Moreover, two real-data applications in different scientific domains are provided which confirm the efficacy of the proposed method.

\section*{Acknowledgement}
We would like to thank Dr. Xiaohong Joe Zhou of the University of Illinois at Chicago for providing the \texttt{Apnea-data} used in Section \ref{Chapter2-qif-Sec:real-data-dti}. 

\appendix

\section{Some preliminary definitions and concepts of operators}
\label{Chapter2-qif-Subsec:preliminary}
Consider the standard $\sL^{2}[0, 1]$ space 
that defines the set of square-integrable functions defined on the closed set $[0,1]$ that takes values on the real line. 
The space $\sL^{2}[0,1]$ is equipped with an inner product and is defined as $\left< f, g \right> = \int_{0}^{1}f(t)g(t)dt$
for $f$ and $g$ in that space,
and forms a Hilbert space. 
Moreover, we denote the norm $\|\cdot\|_{2}$ in $\sL^{2}$ 
which is defined as $\|f\|_{2} 
= \left\{
    \int f^{2}(u)du
\right\}^{1/2}$. 
Define $\sF$ be an operator that assigns an element $f$ in $\sL^{2}[0,1]$ to a new element $\sF f$ in $\sL^{2}[0,1]$ moreover, $\sF$ is linear and bounded. 
A linear mapping $\sF f(\cdot) = \int R(\cdot, u)f(u)du$  for any function $f \in \sL^{2}[0,1]$ 
and for some integrable function $R(\cdot, \cdot)$ on 
$[0, 1]\times[0, 1]$. 
This function is preferably known as 
integral operator 
and the bivariate function $R$ is known as a kernel in statistics and functional analysis literature. 
Furthermore, 
under the assumption that $\int \int R^{2}(u,v)dudv < \infty$.
It is easy to see that $\sF f(\cdot)$ is uniformly continuous 
and compact for a non-negative definite symmetric kernel $R$. 
For some $\lambda$, in Fredholm integral equation
, $\sF\phi = \lambda\phi$
has non-zero solution $\phi$ then we call $\lambda$ as eigen-value of $\sF$ and the solution of the eigen-equation is called eigen-functions, altogether, the pair of eigen-values and eigen-function, viz., $(\lambda, \phi)$ are called eigen-elements. 
Due to non-negative definiteness of $\sF$, the eigen-values are ordered as $\lambda_{1} \geq \lambda_{2} \geq \cdots \geq 0$ .
\par
Suppose for self-adjoint compact operator on Hilbert space $\sH$
consider two operators $\sF$ and $\sG$, 
define perturbation operator  $\Delta = \sG - \sF$ such that 
$\sG = \sF + \Delta$ where $\sG$ is an approximation to $\sF$
where $\Delta$ amount of error is occurred. 
Let $\sF$ and $\sG$ have kernels $F$ and $G$ respectively with eigen-elements $(\theta_{r}, \psi_{r})$ and $(\lambda_{r}, \phi_{r})$. 
For simplicity, we assume that the eigen-values are distinct. Then the following Lemma provides perturbation of the eigen-functions. 

\begin{lemma}[Theorem 5.1.8 in \citet{hsing2015theoretical}]
\label{Chapter2-qif-Lemma:perturbation}
Let $(\lambda, \phi)$ be the eigen-components of $\sF$ and $(\theta, \psi)$ be that of $\sG$ with multiplicity of all eigen-values are restricted to be 1. 
Define $\eta_{k} = \min_{r\neq k}|\lambda_{r} - \lambda_{k}|$.
Assume $\left< \phi_{r}, \psi_{r}\right> \geq 0$ and $\eta_{k} > 0$. 
Then 
\begingroup
\allowdisplaybreaks
\begin{equation}
\label{Chapter2-qif-Eq:perturbationGeneral}
\psi_{k} - \phi_{k} 
    = \underset{r \neq k}{\sum_{r=1 }^{\infty}}
    (\theta_{k} - \lambda_{r})^{-1}
    \sP_{r}\Delta\psi_{k} + 
    \sP_{k}(\psi_{k} - \phi_{k}).
\end{equation}
\endgroup
The above equation follows 
\begingroup
\allowdisplaybreaks
\begin{equation}
\label{Chapter2-qif-Eq:perturbationBound}
    \psi_{k} - \phi_{k} = 
        \underset{r \neq k}{\sum_{r=1 }^{\infty}} 
        (\theta_{k} - \lambda_{r})^{-1}
        \sP_{r}\Delta\psi_{k} + O(\|\Delta\|^{2}).
\end{equation}
\endgroup
\end{lemma}
\begin{remark}
    Equation (\ref{Chapter2-qif-Eq:perturbationBound}) plays an important role in finding the bound of the proposed estimator introduced in Section 2.2. 
    Note that $\sup_{r\geq 1}|\theta_{r} - \lambda_{r}| \leq \|\Delta\| \leq \inf_{r \neq k}|\lambda_{k}-\lambda_{r}|$ (see Theorem 4.2.8 in \citet{hsing2015theoretical} for proof). 
    Thus, it is easy to see, $|\theta_{r} - \lambda_{r}| \leq |\lambda_{k} - \lambda_{r}|$ 
    which implies from Equation (\ref{Chapter2-qif-Eq:perturbationGeneral})
   \begingroup
   \allowdisplaybreaks
    \begin{align*}
    &\psi_{k} - \phi_{k} = \underset{r \neq k}{\sum_{r=1 }^{\infty}}
                (\lambda_{k} -\lambda_{r})^{-1}
                \sum_{s = 0}^{\infty}
                \left\{
                    {(\lambda_{k} - \theta_{r})}/{(\lambda_{k} - \lambda_{r})}
                \right\}^{s}\sP_{r}\Delta\{\phi_{k} + 
                (\psi_{k} - \phi_{k})\} + 
                \sP_{k}(\psi_{k} - \phi_{k})\\
                &= \underset{r \neq k}{\sum_{r=1 }^{\infty}}
                (\lambda_{k} - \lambda_{r})^{-1}
                \sP_{r}\Delta\phi_{k}
                + \underset{r \neq k}{\sum_{r=1 }^{\infty}}
                (\lambda_{k} - \lambda_{r})^{-1}
                    \sP_{r}\Delta(\psi_{k} - \phi_{k})\\
                & \qquad + \underset{r \neq k}{\sum_{r=1 }^{\infty}}\sum_{s = 1}^{\infty}
                    \left\{{(\lambda_{k} - \lambda_{s})^{s}}/{(\lambda_{k} - \lambda_{r})^{s+1}}\right\}
                    \sP_{r}\Delta\psi_{k}
                + \sP_{k}(\psi_{k} - \phi_{k}).
        \numberthis
    \end{align*}
    \endgroup
    Moreover, using Bessel's inequality, we can bound last three terms in the above equation by $\|\Delta\|^{2}$.
\end{remark}

\section{Some useful lemmas}
In this section, we present some useful lemmas. 
For convenience, let us recall the notation. 
Assume that $m_{i}$s are all of the same order, viz, $m \equiv m(n)$. 
Define, 
$d_{n1}(h) = h^{2} + h\overline{m}/m$ and 
$d_{n2}(h) = h^{4} + h^{3}\overline{m}/m + h^{2}\rttensor{m}/m^{2}$
where $\overline{m} = \lim{\sup}_{n\rightarrow \infty}n^{-1}\sum_{i=1}^{n}m/m_{i}$ 
and $\rttensor{m} = \lim{\sup}_{n\rightarrow \infty}n^{-1}\sum_{i=1}^{n}(m/m_{i})^{2}$. 
Denote 
$\delta_{n1}(h) = \left\{d_{n1}(h)\log n/(nh^{2})\right\}^{1/2}$, 
$\delta_{n2}(h) = \left\{d_{n2}(h)\log n/(nh^{4})\right\}^{1/2}$ and $\overline{\delta}_{n}(h) = h^{2} + \delta_{n1}(h) + \delta_{n2}^{2}(h)$. 
Further, $\nu_{a, b} = \int t^{a}K^{b}(t)dt$. 
Define, 
$\bW = (\bphi(t_{1})^{\tp}, \cdots, \bphi(t_{m})^{\tp})^{\tp}$ be matrix of order $m \times \kappa_{0}$ obtained after stacking all $\bphi_{k}$s and random components $\bxi_{i} = (\xi_{i1}, \cdots, \xi_{i\kappa_{0}})^{\tp}$. 
Further, $\bxi$ has mean zero and variance $\bLambda$ which is a diagonal matrix with components $\lambda_{1}, \cdots, \lambda_{k_{0}}$. 
The sign $'\lesssim'$ indicates that for two sequence of positive real numbers ${b}_{{n}1}$ and ${b}_{{n}2}$ we define for large n, ${b}_{{n}1}\lesssim {b}_{{n2}}$ as ${b}_{{n}1} \leq C {b}_{{n}2}$ where $C$ is a positive constant not involving $n$.

\begin{lemma}
\label{Chapter2-qif-Lemma:uniform-rate}
Consider $Z_{1}, \cdots, Z_{n}$ be independent and identically distributed random variables with mean zero and finite variance. 
Suppose that there exists an $M$ such that $P(|Z_{i}| \leq M) = 1$ for all $i = 1, \cdots, n$.
Let
$T_{n} = n^{-1}\sum_{i=1}^{n}Z_{i}$. 
then, 
$n^{-1}\sum_{i=1}^{n}Z_{i} = O((\log n/n)^{1/2}$ almost surely.
If $\sqrt{Var(T_{n})} = O\{(\log n/n)^{1/2}\}$ then $T_{n} = O(\log n/n)$ almost surely.
\end{lemma}
\begin{proof}
Bernstein’s inequality states that 
if $Z_{1}, \cdots, Z_{n}$ be centered independent bounded random variables with probability 1. 
Let $T_{n} = n^{-1}\sum_{i=1}^{n}Z_{i}$, then
let $\Var\{T_{n}\} = \sigma_{n}^{2}$. 
Then for any positive real number $u$, 
we have $P(|T_{n}| \geq u) \leq \exp\{-{nu^{2}}/(2\sigma^{2}_{n} + 2Mu/3)\}$ 
where $M$ is such that $P(|Z_{i}| \leq M) = 1$.
Moreover, if $T_{n}$ converges to its limit in probability fast enough, 
then it converge almost surely in the limit, i.e., if for any $u > 0$, $\sum_{n=1}^{\infty}P(|T_{n}| \geq u) < \infty$ them $T_{n}$ converges to zero almost surely.
Now, choose $u = \sqrt{{4\sigma^{2}_{n}\log n}/{n}} + {4M\log n}/{3n}$.
Thus, $\sum_{n=1}^{\infty}P(|T_{n}| \geq u) < \sum_{n = 1}^{\infty}1/n^{2}$ which is finite. 
Therefore, $T_{n} = O(u)$ almost surely. Now let $\sigma_{n} \leq \sqrt{4M^{2}\log n/9n}$, we have, $T_{n} = O(\log n/n)$ and if $\sigma_{n} = O(1)$ then $T_{n} = O((\log n/n)^{1/2})$ almost surely.
\end{proof}

\begin{lemma}
\label{Chapter2-qif-Lemma:phi-rate}
Suppose $T_{ij}$ are i.i.d. with density $f_{T}$. Then for 
fixed $i=1, \cdots, n$, any $k$ and $l \geq 1$, under assumptions (C2) and (C6)c, 
$m_{i}^{-1}\sum_{j=1}^{m_{i}}
    \phi_{k}(T_{ij})\phi_{l}(T_{ij}) 
    = \textbf{1}(k = l) + O(\left({\log m_{i}}/{m_{i}}\right)^{1/2})$
    {almost surely}
where $\textbf{1}$ is the indicator function.    

\end{lemma}
\begin{proof}
Observe,
$\E\left\{ m_{i}^{-1} \sum_{j=1}^{m_{i}}\phi_{k}(T_{ij})\phi_{l}(T_{ij}) \right\}
= \int\phi_{k}(t)\phi_{l}(t)dt = \textbf{1}(k = l)$
and, 
\begingroup
\allowdisplaybreaks
\begin{align*}
    &\Var\left\{m_{i}^{-1}\sum_{j = 1}^{m_{i}}\phi_{k}(T_{ij})\phi_{l}(T_{ij})\right\}
    = \E\left\{ 
            m_{i}^{-1}\sum_{j=1}^{m_{i}}\phi_{k}(T_{ij})\phi_{l}(T_{ij})
        \right\}^{2} - \textbf{1}(k=1)\\
    &= m_{i}^{-2}\sum_{j=1}^{m_{i}}
    \E\{
        \phi_{k}^{2}(T_{ij})\phi_{l}^{2}(T_{ij})
    \} 
    + m_{i}^{-2}
    \underset{j_{1}\neq j_{2}}{\sum_{j_{1} =1}^{m_{i}}\sum_{j_{2}=1}^{m_{i}}}
    \E\{
        \phi_{k}(T_{ij_{1}})\phi_{k}(T_{ij_{2}})\phi_{l}(T_{ij_{1}})\phi_{l}(T_{ij_{2}})
    \}\\
    & \qquad - \textbf{1}(k = l)\\
    &= m_{i}^{-2}\sum_{j=1}^{m_{i}}
    \E\{
        \phi_{k}^{2}(T_{ij})\phi_{l}^{2}(T_{ij})
    \} 
    + m_{i}^{-2}
    \underset{j_{1}\neq j_{2}}{\sum_{j_{1} =1}^{m_{i}}\sum_{j_{2}=1}^{m_{i}}}
    \E\{
        \phi_{k}(T_{ij_{1}})\phi_{l}(T_{ij_{1}})
    \}
    \E\{
        \phi_{k}(T_{ij_{2}})\phi_{l}(T_{ij_{2}})
    \}\\
    & \qquad - \textbf{1}(k = l)\\
    & = \begin{cases}
            m_{i}^{-1}\int\phi_{k}^{4}(t)dt + (m_{i}-1)/m_{i}(\int\phi_{k}^{2}(t)dt)^{2} -1 
            & \text{if } k = l\\
            m_{i}^{-1}\int\phi_{k}^{2}(t)\phi_{l}^{2}(t)dt +
            (m_{i}-1)/m_{i}(\int\int\phi_{k}(t)\phi_{l}(t)\phi_{k}(t')\phi_{l}(t')dtdt') 
            & \text{if } k\neq l
    \end{cases}\\ 
    &= O(1/m_{i}).\numberthis
\end{align*}
\endgroup
Therefore, applying the Lemma \ref{Chapter2-qif-Lemma:uniform-rate}, the result is immediate.
\end{proof}
\begin{lemma}
\label{Chapter2-qif-Lemma:mu-phi-rate}
Suppose $T_{ij}$ are i.i.d with density $f_{T}$. Then for fixed $i=1, \cdots, n$, 
for any $k \geq 1$, 
under assumptions (C2), (C6)c,  $m_{i}^{-1}\sum_{j=1}^{m_{i}}\dot{\mu}_{i}(T_{ij})\phi_{k}(T_{ij}) = \int\dot{\mu}_{i}(t)\phi_{k}(t)dt 
    + O\left((\log m_{i}/m_{i})^{1/2}\right)$ almost surely.
\end{lemma}
\begin{proof}
Observe, 
$\E\left\{m_{i}^{-1}\sum_{j=1}^{m_{i}}\dot{\mu}_{i}(T_{ij})\phi_{k}(T_{ij})\right\} = \int\dot{\mu}_{i}(t)\phi_{k}(t)dt$
and, 
\begin{align*}
            &\Var\left\{
            m_{i}^{-1}\sum_{j=1}^{m_{i}}\dot{\mu}_{i}(T_{ij})\phi_{k}(T_{ij})
        \right\} 
        \leq \E\left\{
            m_{i}^{-1}\sum_{j=1}^{m_{i}}\dot{\mu}_{i}(T_{ij})\phi_{k}(T_{ij})
        \right\}^{2}\\
        &= m_{i}^{-2}\sum_{j=1}^{m_{i}}\E\left\{ \dot{\mu}_{i}^{2}(T_{ij})\phi_{k}(T_{ij})^{2} \right\}
        + m_{i}^{-2}\underset{j_{1}\neq j_{2}}{\sum_{j_{1} = 1}^{m_{i}}\sum_{j_{2} = 1}^{m_{i}}}
        \E\{
            \dot{\mu}_{i}(T_{ij_{1}})\dot{\mu}_{i}(T_{ij_{2}})
            \phi_{k}(T_{ij_{1}})\phi_{k}(T_{ij_{2}})
        \}\\
        &= O(1/m_{i}), \qquad \text{since }\int\dot{\mu}^{2}(t)\phi_{k}^{2}(t)dt <\infty.
    \numberthis
\end{align*}

Therefore, applying the Lemma \ref{Chapter2-qif-Lemma:uniform-rate}, the result is immediate.
\end{proof}

\begin{lemma}
\label{Chapter2-qif-Lemma:part-U1}
Define, $\sM_{ir} = \int \dot{\mu}_{i}(t)\phi_{r}(t)dt$ and $V_{r} = \E\{\int\dot{\mu}(t)\phi_{r}(t)dt \}^{2}$
for $r \geq 1$.
Then under Conditions (C6)a and (C6)b, for some $\alpha > 0$ such that $V_{r}\lambda_{r}^{-2}r^{1+\alpha} \rightarrow 0 $ as $r \rightarrow \infty$ (due to Condition (C6)b, 
$\underset{r\neq k}{\sum_{r=1}^{\infty}}(\lambda_{k}-\lambda_{r})^{-1} n^{-1}\sum_{i=1}^{n}\sM_{ir}\xi_{ik} = O\left\{(\log n/n)^{1/2}\lambda_{k}^{1/2}k^{(1-\alpha)/2}\right\}$ {almost surely}.
\end{lemma}
\begin{proof}
It is easy to see that, 
$\E\Big\{ \underset{r\neq k}{\sum_{r=1}^{\infty}}(\lambda_{k}-\lambda_{r})^{-1} n^{-1}\sum_{i=1}^{n}\sM_{ir}\xi_{ik} \Big\} = 0$.
Using the spacing condition among the eigen-values in (C6)a, 
for each $1\leq k < r < \infty$ and for nonzero finite generic constant $C_{0}$, 
\begin{equation}
    \label{Chapter2-qif-Eq:space}
    {{\lambda_{k}}}/{|\lambda_{k} - \lambda_{r}|} 
    \leq 
    C_{0}{r}/{|k-r|}.
\end{equation}
Similar kind of conditions can be invoked such as the convexity assumption, i.e. $\lambda_{r} - \lambda_{r+1} \leq \lambda_{r-1} - \lambda_{r}$ for all $r \geq 2$. Thus, using Inequality (\ref{Chapter2-qif-Eq:space}), for some $\alpha>0$, with condition $V_{r}\lambda_{r}^{-2}r^{1+\alpha} \rightarrow 0$ as $r \rightarrow \infty$, 
we can write 
\begin{align*}
    \label{Chapter2-qif-Eq:infiniteSum}
        &\underset{r\neq k}{\sum_{r=1}^{\infty}}
        V_{r}(\lambda_{k}-\lambda_{r})^{-2}
        \lesssim\underset{r\neq k}{\sum_{r=1}^{\infty}}
        V_{r}\left\{{\max(k, r)}/{|k-r|\max(\lambda_{k}, \lambda_{r})} \right\}^{2}\\
        &= \sum_{r\leq k/2}V_{r}\lambda_{r}^{-2}{k^{2}}/{(k-r)^{2}} 
        + \sum_{r > 2k}V_{r}\lambda_{k}^{-2}{r^{2}}/{(k-r)^{2}}\\
        & \qquad + \sum_{k/2 < r < k}V_{r}\lambda_{r}^{-2}{k^{2}}/{(k-r)^{2}} 
        + \sum_{k < r < 2k}V_{r}\lambda_{k}^{-2}{r^{2}}/{(k-r)^{2}}\\
        &\lesssim\sum_{r \leq k/2, r > 2k}V_{r}\lambda_{r}^{-2} + k^{2}\sum_{k/2 < r < 2k}V_{r}\lambda_{r}^{-2}(k-r)^{-2}\\
        &\lesssim 1 + k^{1-\alpha}\sum_{k/2 < r< 2k}(k-r)^{-2} \lesssim k^{1-\alpha}.
        \numberthis
\end{align*}
This follows the line of proofs in \citet{hall2009theory} in different contexts. Thus, using the inequality (\ref{Chapter2-qif-Eq:infiniteSum}), it follows that 
\begin{align*}
\label{Chapter2-qif-Eq:varU1}    
        \Var\Big\{ \underset{r\neq k}{\sum_{r=1}^{\infty}}(\lambda_{k}-\lambda_{r})^{-1} n^{-1}\sum_{i=1}^{n}\sM_{ir}\xi_{ik} \Big\} 
        &= n^{-1}\lambda_{k} \underset{r\neq k}{\sum_{r=1}^{\infty}} V_{r}(\lambda_{k} -\lambda_{r})^{-2} = O\Big(n^{-1}\lambda_{k}k^{(1-\alpha)}\Big).
        \numberthis
\end{align*}
Therefore, applying Lemma \ref{Chapter2-qif-Lemma:uniform-rate}, the proof is immediate.
\end{proof}

\begin{lemma}
\label{Chapter2-qif-Lemma:part-U2}
For $\sM_{ir} = \int\dot{\mu}(t)\phi_{r}(t)dt$ and $\eta_{k} = \min_{r\neq k}|\lambda_{k} -\lambda_{r}| > 0$,
under Conditions (C6)a and (C6)b,
$\underset{r_{1}\neq k}{\sum_{r_{1}\neq 1}^{\infty}}
\underset{r_{2}\neq k}{\sum_{r_{2}\neq 1}^{\kappa_{0}}}
(\lambda_{k} - \lambda_{r_{1}})^{-1}
(\lambda_{k} - \lambda_{r_{2}})^{-1}
n^{-1}\sum_{i=1}^{n}\sM_{ir_{1}}\xi_{ir_{2}} = O\left(
    (\log n/n)^{1/2}\kappa_{0}^{(3-\alpha)/2}
    \lambda_{\kappa_{0}}^{-1}
    \left\{\sum_{r=1}^{\kappa_{0}}\lambda_{r}\right\}^{1/2}
\right)
$ almost surely.
\end{lemma}
\begin{proof}
It is not difficult to see that, 
$$\E\Big\{ \underset{r_{1}\neq k}{\sum_{r_{1}=1}^{\infty}}
    \underset{r_{2}\neq k}{\sum_{r_{2} = 1}^{\kappa_{0}}}
    (\lambda_{k} - \lambda_{r_{1}})^{-1}
    (\lambda_{k} - \lambda_{r_{2}})^{-1}
    n^{-1}\sum_{i=1}^{n}\sM_{ir_{1}}\xi_{ir_{2}} \Big\} = 0.$$
Moreover, using the spacing condition mentioned in (C6)b, one can derive the upper bound of $\eta_{k}^{-1}$ by 
\begingroup
\allowdisplaybreaks
\begin{align*}
\label{Chapter2-qif-Eq:eta-bound}
    \eta_{k}^{-1} &= \left\{\min_{r\neq k}|\lambda_{k} - \lambda_{r}| \right\}^{-1} = \max_{r \neq k}|\lambda_{k} - \lambda_{r}|^{-1} \\
        &\lesssim \max_{r \neq k}\left\{ {\max(k, r)}/{|k-r|\max(\lambda_{k}, \lambda_{r})}\right\} \leq \lambda_{k}^{-1}k.
    \numberthis
\end{align*}
\endgroup
Due to the monotonic decreasing property of eigen-values, for fixed $k = 1, \cdots, \kappa_{0}$, we have 
\begin{align*}
        \underset{r \neq k}{\sum_{r=1}^{\kappa_{0}}}\lambda_{r}(\lambda_{k} - \lambda_{r})^{-2} 
        &\lesssim \eta_{k}^{-2}\sum_{r = 1}^{\kappa_{0}}\lambda_{r} \lesssim \lambda_{k}^{-2}k^{2}\sum_{r=1}^{\kappa_{0}}\lambda_{r} \lesssim \lambda_{\kappa_{0}}^{-2}\kappa_{0}^{2}\sum_{r=1}^{\kappa_{0}}\lambda_{r}.
        \numberthis
\end{align*}
Therefore, the following holds under similar conditions to obtain the Inequality (\ref{Chapter2-qif-Eq:infiniteSum}),
\begingroup
\allowdisplaybreaks
\begin{align*}
    &\E\left\{ 
            \underset{r_{1}\neq k}{\sum_{r_{1}=1}^{\infty}}
            \underset{r_{2}\neq k}{\sum_{r_{2}=1}^{\kappa_{0}}}
            (\lambda_{k} - \lambda_{r_{1}})^{-2}
            (\lambda_{k} - \lambda_{r_{2}})^{-2}
            \left(n^{-1}\sum_{i=1}^{n}\sM_{ir_{1}}\xi_{ir_{2}}\right)^{2}
        \right\}\\
        & = n^{-1}
        \underset{r_{1}\neq k}{\sum_{r_{1}=1}^{\infty}}
        \underset{r_{1}\neq k}{\sum_{r_{1}=1}^{\kappa_{0}}}
          (\lambda_{k} - \lambda_{r_{1}})^{-2}
          (\lambda_{k} - \lambda_{r_{2}})^{-2}
          V_{r_{1}}\lambda_{r_{2}}\\
        & \lesssim n^{-1}k^{1-\alpha}
        \underset{r\neq k}{\sum_{r=1}^{\kappa_{0}}}
        \lambda_{r}(\lambda_{k} - \lambda_{r})^{-2}
        \lesssim n^{-1}\kappa_{0}^{3-\alpha}\lambda_{\kappa_{0}}^{-2}\sum_{r=1}^{\kappa_{0}}\lambda_{r}.
        \numberthis
\end{align*}
\endgroup
Therefore, applying the Lemma \ref{Chapter2-qif-Lemma:uniform-rate}, the result is immediate.
\end{proof}

\begin{remark}
\label{Chapter2-qif-Remark:per}
Define, $d_{ik}(T_{ij_{1}}, T_{ij_{2}}):= \widehat{\phi}_{k}(T_{ij_{1}})\widehat{\phi}_{k}(T_{ij_{2}}) - \phi_{k}(T_{ij_{1}})\phi_{k}(T_{ij_{2}})$.
Now replace $\sG$, $\theta$ and $\psi$ by $\widehat{\sF}$, $\widehat{\lambda}$ and $\widehat{\phi}$ respectively since $\widehat{\sF}$ be the approximation of $\sF$ and $\Delta$ is the corresponding perturbation operator in Equation (\ref{Chapter2-qif-Eq:perturbationGeneral}).  
Therefore, Lemma \ref{Chapter2-qif-Lemma:perturbation} immediately implies the following expansion, which is the key fact to represent the objective function in QIF.
\begin{equation}
    \widehat{\phi}_{k} - \phi_{k} = \underset{r \neq k}{\sum_{r = 1}^{\infty}} 
    (\lambda_{k} - \lambda_{r})^{-1}\left<\phi_{r}, \Delta\phi_{k} \right>\phi_{r} + O(\|\Delta\|^{2})\qquad\text{almost surely},
\end{equation}
where $\Delta$ be the integral operator with kernel $\widehat{R} - R$.
\end{remark}

\section{Proof of Theorem 1}
For the $k-$th element of $\overline{\bg}_{n}(\bbeta_{0})$ $(1\leq k \leq \kappa_{0})$,
\begingroup
\allowdisplaybreaks
\begin{align*}
    \overline{\bg}_{n}^{(k)}(\bbeta_{0}) & = n^{-1}\sum_{i=1}^{n}
        m_{i}^{-2}
        \dot{\bmu}_{i}^{\tp}\widehat{\bPhi}_{k}(\by_{i}-\bmu_{i})
        = n^{-1}\sum_{i=1}^{n}
            m_{i}^{-2}
            \dot{\bmu}_{i}^{\tp}\widehat{\bPhi}_{k}\bW\bxi_{i}\\
        & = n^{-1}\sum_{i=1}^{n}
            m_{i}^{-2}
            \dot{\bmu}_{i}^{\tp}\bPhi_{k}\bW\bxi_{i}
            + n^{-1}\sum_{i=1}^{n}
            m_{i}^{-2}
            \dot{\bmu}_{i}^{\tp}(\widehat{\bPhi}_{k} - \bPhi_{k})\bW\bxi_{i}
        := J_{k}^{n1} + J_{k}^{n2}.\numberthis
\end{align*}
\endgroup
Now, using Lemmas \ref{Chapter2-qif-Lemma:phi-rate} and \ref{Chapter2-qif-Lemma:mu-phi-rate},
the first part of the expression of $\overline{\bg}_{n,k}(\bbeta_{0})$ becomes
\begingroup
\allowdisplaybreaks
\begin{align*}
    J_{k}^{n1} &= n^{-1}\sum_{i = 1}^{n}
        m_{i}^{-2}\dot{\bmu}_{i}^{\tp}\bPhi_{k}\bW\bxi_{i}\\\
        & = n^{-1}\sum_{i=1}^{n}
            m_{i}^{-2}
            \sum_{j_{1}=1}^{m_{i}}
            \sum_{j_{2}=1}^{m_{i}}
            \sum_{l = 1}^{\kappa_{0}}
            \dot{\mu}_{i}(T_{ij_{1}})
            \phi_{k}(T_{ij_{1}})\phi_{k}(T_{ij_{2}})\phi_{l}(T_{ij_{2}})\xi_{il}\\
        & \lesssim n^{-1}\sum_{i = 1}^{n}
            \left\{\sM_{ik} + O((\log m/m)^{1/2})\right\}
            \left\{1+ O((\log m/m)^{1/2})\right\}
            \xi_{ik},\\
            & \qquad \text{where } \sM_{ik} = \int \dot{\mu}_{i}(t)\phi_{k}(t)dt\\
        &= n^{-1}\sum_{i=1}^{n}\sM_{ik}\xi_{ik}
        \left\{ 1 + O\left( (\log m/m)^{1/2}\right)\right\}
        = O\left((\log n/n)^{1/2}
        \left\{ 1 +  (\log m/m)^{1/2}\right\}\right)\\
        & \qquad \text{almost surely}.
        \numberthis
\end{align*}
\endgroup
On the other hand, the last part of $\overline{\bg}_{n}^{(k)}(\bbeta_0)$ can be expressed as 
\begin{align*}
\label{Chapter2-qif-Eq:Jk-n2}
        J_{k}^{n2} &= n^{-1}\sum_{i=1}^{n}
                {m_{i}^{-2}}\dot{\bmu}_{i}^{\tp}
                (\widehat{\bPhi}_{k}-\bPhi_{k})\bW\bxi_{i}\\
        &=n^{-1}\sum_{i=1}^{n}{m_{i}^{-2}}
                \sum_{j_{1}=1}^{m_{i}}\sum_{j_{2}=1}^{m_{i}}
                \sum_{l=1}^{\kappa_{0}}
                \dot{\mu}_{i}(T_{ij_{1}})d_{ik}(T_{ij_{1}}, T_{ij_{2}})
                \phi_{l}(T_{ij_{2}})\xi_{il},
        \numberthis
\end{align*}
where $d_{ik}(T_{ij_{1}}, T_{ij_{2}}):= \widehat{\phi}_{k}(T_{ij_{1}})\widehat{\phi}_{k}(T_{ij_{2}}) - \phi_{k}(T_{ij_{1}})\phi_{k}(T_{ij_{2}})$ as defined in Remark \ref{Chapter2-qif-Remark:per}.
Therefore, using the discussion in Remark \ref{Chapter2-qif-Remark:per}, we can obtain the following expression almost surely. 
\begingroup
\allowdisplaybreaks
\begin{align*}
    &d_{ik}(T_{ij_{1}}, T_{ij_{2}})
        := \widehat{\phi}_{k}(T_{ij_{1}})\widehat{\phi}_{k}(T_{ij_{2}}) - \phi_{k}(T_{ij_{1}})\phi_{k}(T_{ij_{2}})\\ 
        &=\left\{ \phi_{k}(T_{ij_{1}}) + \underset{r \neq k}{\sum_{r = 1}^{\infty}} 
        (\lambda_{k} - \lambda_{r})^{-1}\left<\phi_{r}, \Delta\phi_{k} \right>\phi_{r}(T_{ij_{1}}) + O(\|\Delta\|^{2})
        \right\}\\
        & \qquad \times\left\{
        \phi_{k}(T_{ij_{2}}) + 
        \underset{r \neq k}{\sum_{r = 1}^{\infty}} 
        (\lambda_{k} - \lambda_{r})^{-1}\left<\phi_{r},     \Delta\phi_{k} \right>\phi_{r}(T_{ij_{2}}) + O(\|\Delta\|^{2})
        \right\} - \phi_{k}(T_{ij_{1}})\phi_{k}(T_{ij_{2}}) 
        \\
        &= \underset{r\neq k}{\sum_{r=1}^{\infty}}(\lambda_{k} - \lambda_{r})^{-1}
        \left<
            \phi_{r}, \Delta\phi_{k}
        \right>
        \Big\{
        \phi_{r}(T_{ij_{1}})\phi_{k}(T_{ij_{2}})
        + \phi_{k}(T_{ij_{1}})\phi_{r}(T_{ij_{2}})\Big\}\\
        & \qquad + 
        \sum_{r_{1}\neq k}
        \sum_{r_{2}\neq k}
        (\lambda_{k}-\lambda_{r_{1}})^{-1}
        (\lambda_{k}-\lambda_{r_{2}})^{-1}
        \left< 
            \phi_{r_{1}}, \Delta\phi_{k}
        \right>
        \left< 
            \phi_{r_{2}}, \Delta\phi_{k}
        \right>
        \phi_{r_{1}}(T_{ij_{1}})
        \phi_{r_{2}}(T_{ij_{2}})\\
        & \qquad + O(\|\Delta\|^{2})\\
        &:= I_{ik}^{n1}(T_{ij_{1}}, T_{ij_{2}})
        + I_{ik}^{n2}(T_{ij_{1}}, T_{ij_{2}})
        + O(\|\Delta\|^{2}).
        \numberthis
\end{align*}
\endgroup
Thus, almost surely, we can write 
\begingroup\allowdisplaybreaks
\begin{align*}
     &n^{-1}\sum_{i=1}^{n}
                m_{i}^{-2}\dot{\bmu}_{i}^{\tp}
                (\widehat{\bPhi}_{k}-\bPhi_{k})\bW\bxi_{i}\\
        &= n^{-1}\sum_{i=1}^{n}m_{i}^{-2}
                \sum_{j_{1}=1}^{m_{i}}\sum_{j_{2}=1}^{m_{i}}
                \sum_{l=1}^{\kappa_{0}}
                \dot{\mu}_{i}(T_{ij_{1}})d_{ik}(T_{ij_{1}}, T_{ij_{2}})
                \phi_{l}(T_{ij_{2}})\xi_{il}\\
        &=n^{-1}\sum_{i=1}^{n}m_{i}^{-2}
                \sum_{j_{1}=1}^{m_{i}}\sum_{j_{2}=1}^{m_{i}}
                \sum_{l=1}^{\kappa_{0}}
                \dot{\mu}_{i}(T_{ij_{1}})d_{ik}(T_{ij_{1}}, T_{ij_{2}})
                \phi_{l}(T_{ij_{2}})\xi_{il}\\
        &=n^{-1}\sum_{i=1}^{n}m_{i}^{-2}
                \sum_{j_{1}=1}^{m_{i}}\sum_{j_{2}=1}^{m_{i}}
                \sum_{l = 1}^{\kappa_{0}}\dot{\mu}_{i}(T_{ij_{1}})
                \Big\{ 
                    I_{ik}^{n1}(T_{ij_{1}}, T_{ij_{2}}) + I_{ik}^{n2}(T_{ij_{1}}, T_{ij_{2}})
                \Big\}
                \phi_{l}(T_{ij_{2}})\xi_{il}
                + O(\|\Delta\|^{2})\\
        &:= J_{k1}^{n2} + J_{k2}^{n2} + O(\|\Delta\|^{2}).
    \numberthis
\end{align*}
\endgroup
Under assumptions (C1)-(C5), 
by using Theorem 3.3 in \citet{li2010uniform},
$\|\Delta\|^{2} = O(h^{4} + \delta_{n2}^{2}(h))$ almost surely.
Now observe that 
\begingroup
\allowdisplaybreaks
\begin{align*}
\label{Chapter2-qif-Eq:JK1_n2}
        &J_{k1}^{n2} 
        =n^{-1}\sum_{i=1}^{n}m_{i}^{-2}
                \sum_{j_{1}=1}^{m_{i}}\sum_{j_{2}=1}^{m_{i}}
                \sum_{l=1}^{\kappa_{0}}
                \underset{r\neq k}{\sum_{r=1}^{\infty}}
                (\lambda_{k} - \lambda_{r})^{-1}
                \dot{\mu}_{i}(T_{ij_{1}})\Big\{ 
                    \phi_{r}(T_{ij_{1}})\phi_{k}(T_{ij_{2}})
                    + \phi_{k}(T_{ij_{1}})\phi_{r}(T_{ij_{2}})
                \Big\}\\
        & \qquad \times\phi_{l}(T_{ij_{2}})
                \left< \phi_{r}, \Delta\phi_{k}\right>
                \xi_{il}\\
        & \lesssim
                n^{-1}\sum_{i=1}^{n}m_{i}^{-1}
                \sum_{j_{1} = 1}^{m_{i}}
                \underset{r\neq k}{\sum_{r=1}^{\infty}}\sum_{l = 1}^{\kappa_{0}}
                (\lambda_{k} - \lambda_{r})^{-1}
                \dot{\mu}_{i}(T_{ij_{1}})\phi_{r}(T_{ij_{1}})
                \left\{\textbf{1}(l = k) + O((\log m/m)^{1/2})\right\}
                \left<\phi_{r}, \Delta\phi_{k} \right>
                \xi_{il}\\
        & \qquad +  n^{-1}\sum_{i=1}^{n}m_{i}^{-1}
                \sum_{j_{1} = 1}^{m_{i}}
                \underset{r\neq k}{\sum_{r=1}^{\infty}}
                \sum_{l = 1}^{\kappa_{0}}
                (\lambda_{k} - \lambda_{r})^{-1}
                \dot{\mu}_{i}(T_{ij_{1}})
                \phi_{k}(T_{ij_{1}})
                \left\{\textbf{1}(r = l) + O((\log m/m)^{1/2})\right\}
                \left<\phi_{r}, \Delta\phi_{k} \right>
                \xi_{il}\\
        & \lesssim 
                n^{-1}\sum_{i=1}^{n}m_{i}^{-1}\sum_{j_{1} = 1}^{m_{i}}
                \underset{r\neq k}{\sum_{r=1}^{\infty}}
                (\lambda_{k} - \lambda_{r})^{-1}
                \dot{\mu}_{i}(T_{ij_{1}})\phi_{r}(T_{ij_{1}})\left<\phi_{r}, \Delta\phi_{k}\right>\xi_{ik}\left\{ 1+ O((\log m/m)^{1/2}) \right\} \\
        & \qquad + 
                n^{-1}\sum_{i=1}^{n}m_{i}^{-1}\sum_{j_{1}=1}^{m_{i}}
                \underset{r \neq k}{\sum_{r = 1}^{\kappa_{0}}}
                (\lambda_{k} - \lambda_{r})^{-1}
                \dot{\mu}_{i}(T_{ij_{1}})\phi_{k}(T_{ij_{1}})\left< \phi_{r}, \Delta\phi_{k}\right>\xi_{ir}
                \left\{ 1+ O((\log m/m)^{1/2}) \right\}\\
        &:= (U_{k1}^{n} + U_{k2}^{n})\left\{ 1+ O((\log m/m)^{1/2})\right\}\qquad\text{almost surely}.
        \numberthis
\end{align*}
\endgroup
Then applying the triangle inequality, we have 
\begin{align*}
        U_{k1}^{n} &= n^{-1}\sum_{i=1}^{n}m_{i}^{-1}\sum_{j=1}^{m_{i}}\underset{r\neq k}{\sum_{r=1}^{\infty}}(\lambda_{k}-\lambda_{r})^{-1}
        \dot{\mu}_{i}(T_{ij})\phi_{r}(T_{ij}) \left<\Delta\phi_{k}, \phi_{r} \right>\xi_{ik}\\
        &\lesssim \|\Delta\phi_{k}\| 
        \underset{r\neq k}{\sum_{r=1}^{\infty}}(\lambda_{k}-\lambda_{r})^{-1} n^{-1}\sum_{i=1}^{n}
        \left\{\sM_{ir} + O((\log m/m)^{1/2})\right\}
        \xi_{ik}\\
        &= \|\Delta\phi_{k}\|
            \underset{r\neq k}{\sum_{r=1}^{\infty}}
            (\lambda_{k}-\lambda_{r})^{-1}
            n^{-1}\sum_{i=1}^{n}\sM_{ir}\xi_{ik}
            \left\{ 
                1+O((\log m/m)^{1/2})
            \right\},
    \numberthis
\end{align*}
where $\sM_{ir} = \int\dot{\mu}_{i}(t)\phi_{r}(t)dt$. 
By Lemma 6 of \citet{li2010uniform}, under conditions (C1)-(C5), for any measurable bounded function $e$ on $[0,1]$, $\|\Delta\phi_{k}\| = O(h^{2} + \delta_{n1}(h) + \delta^{2}_{n2}(h)) \equiv O(\overline{\delta}_{n}(h))$ {almost surely},
where $\overline{\delta}_{n}(h) = h^{2} + \delta_{n1}(h) + \delta^{2}_{n2}(h)$.
Thus, in addition with Inequalities (\ref{Chapter2-qif-Eq:varU1}) in Lemma \ref{Chapter2-qif-Lemma:part-U1},
we obtain 
$$U_{k1}^{n} = O\left(\overline{\delta}_{n}(h)(\log n/n)^{1/2}\lambda_{k}^{1/2}k^{(1-\alpha)/2}\{1+ (\log m/m)^{1/2}\}\right)$$ almost surely. 
Next, under the spacing condition mentioned earlier and in assumption (C6)a, using the Inequality (\ref{Chapter2-qif-Eq:eta-bound}), recall $\eta_{k}^{-1} \lesssim \lambda_{k}^{-1}k$.
Thus, observe that 
\begin{align*}
    U_{k2}^{n} &= n^{-1}\sum_{i=1}^{n}m_{i}^{-1}
            \sum_{j = 1}^{m_{i}}
            \underset{r\neq k}{\sum_{r = 1}^{\kappa_{0}}}
            (\lambda_{k}-\lambda_{r})^{-1}
            \dot{\mu}_{i}(T_{ij})\phi_{k}(T_{i_{j}})
            \left<\Delta\phi_{k}, \phi_{r} \right>\xi_{ir}\\
        & \lesssim \|\Delta\phi_{k}\|\underset{r \neq k}{\sum_{r = 1}^{\kappa_{0}}}
            (\lambda_{k} - \lambda_{r})^{-1}
            \left\{n^{-1}\sum_{i=1}^{n}\sM_{ik}\xi_{ir}\right\}
            \left\{1+O((\log m/m)^{1/2})\right\}\\
        & \lesssim\|\Delta\phi_{k}\|\eta_{k}^{-1}
            \underset{r \neq k}{\sum_{r=1}^{\kappa_{0}}}
            \left\{n^{-1}\sum_{i=1}^{n}\sM_{ik}\xi_{ir}\right\}
            \left\{1+O((\log m/m)^{1/2})\right\}.
        \numberthis
\end{align*}
Using condition (C6)b, we also have $V_{k}^{1/2}\eta_{k}^{-1} \lesssim V_{k}^{1/2}\lambda_{k}^{-1}k = O(k^{(1-\alpha)/2})$. 
Finally, combining with the bounds for $U_{k1}^{n}, U_{k2}^{n}$, we have, almost surely, 
\begin{align*}
    J_{k1}^{n2} &= O\Big\{ 
            (\log n/n)^{1/2}\overline{\delta}_{n}(h) 
            \Big(
                \lambda_{k}^{1/2} k^{(1-\alpha)/2} + 
                \eta_{k}^{-1}V_{k}^{1/2}
                \underset{r \neq k}{\sum_{r = 1}^{\kappa_{0}}} \lambda_{r}^{1/2}
            \Big)
            \left(1+ (\log m/m)^{1/2}\right)
        \Big\}
        \\
        &:= O(\omega_{k1}(n, h)),
    \numberthis
\end{align*}
where 
$\omega_{k1}(n, h) = 
    (\log n/n)^{1/2}
    \overline{\delta}_{n}(h)
    k^{(1-\alpha)/2} 
    \sum_{k = 1}^{\kappa_{0}}\lambda_{r}^{1/2}
    \left\{1+ (\log m/m)^{1/2}\right\}$.
\par
It is easy to see that $\sum_{r=1}^{\kappa_{0}}\lambda_{r}^{1/2} \sim \kappa_{0}^{-{\tau_{1}}/{2} + 1}$. 
Therefore, for $\tau = \alpha + \tau_{1}$,
$\omega_{k1}(n,h) \sim 
    (\log n/n)^{1/2}
    \overline{\delta}_{n}(h)
    \kappa_{0}^{(3-\tau)/2}
    \left\{1+ (\log m/m)^{1/2}\right\}$.
Similarly, to the derivation of the bound for $J_{k1}^{n2}$, we can write 
\begingroup\allowdisplaybreaks
\begin{align*}
    \label{Chapter2-qif-Eq:JK2_n2}
    J_{k2}^{n2} &= {n}^{-1}\sum_{i=1}^{n}{m_{i}^{-2}}
            \sum_{j_{1} = 1}^{m_{i}}
            \sum_{j_{2} = 1}^{m_{i}}
            \sum_{l=1}^{\kappa_{0}}
        \dot{\mu}_{i}(T_{ij_{1}})I_{ik}^{n2}(T_{ij_{1}}, T_{ij_{2}})
            \phi_{l}(T_{ij_{2}})\xi_{il}\\
        &={n}^{-1}\sum_{i=1}^{n}{m_{i}^{-2}}
            \sum_{j_{1} = 1}^{m_{i}}
            \sum_{j_{2} = 1}^{m_{i}}
            \sum_{l=1}^{\kappa_{0}}
            \underset{r_{1} \neq k}{\sum_{r_{1}=1}^{\infty}}
            \underset{r_{2} \neq k}{\sum_{r_{1}=1}^{\infty}}
            (\lambda_{k} - \lambda_{r_{1}})^{-1}
            (\lambda_{k} - \lambda_{r_{2}})^{-1}
            \left<\phi_{r_{1}}, \Delta\phi_{k} \right>
            \left<\phi_{r_{2}}, \Delta\phi_{k} \right>\\
        &\qquad\times\dot{\mu}_{i}(T_{ij_{1}})
            \phi_{r_{1}}(T_{ij_{1}})
            \phi_{r_{2}}(T_{ij_{2}})
            \phi_{l}(T_{ij_{2}})\xi_{il}\\
        &\lesssim {n}^{-1}\sum_{i=1}^{n}{m_{i}}^{-1}
            \sum_{j_{1} = 1}^{m_{i}}\sum_{l=1}^{\kappa_{0}}
            \underset{r_{1} \neq k}{\sum_{r_{1}=1}^{\infty}}
            \underset{r_{2} \neq k}{\sum_{r_{2}=1}^{\infty}}
            (\lambda_{k} - \lambda_{r_{1}})^{-1}
            (\lambda_{k} - \lambda_{r_{2}})^{-1}
            \left<\phi_{r_{1}}, \Delta\phi_{k} \right>
            \left<\phi_{r_{2}}, \Delta\phi_{k} \right>\\
        &\qquad\times\dot{\mu}_{i}(T_{ij_{1}})
            \phi_{r_{1}}(T_{ij_{1}})
            \left\{\textbf{1}(r_{2} = l) + 
            O((\log m/m)^{1/2})\right\}\xi_{il}\\
        &=n^{-1}\sum_{i=1}^{n}{m_{i}}^{-1}
            \sum_{j_{1} = 1}^{m_{i}}
            \underset{r_{1} \neq k}{\sum_{r_{1}=1}^{\infty}}
            \underset{r_{2} \neq k}{\sum_{r_{2}=1}^{\kappa_{0}}}
            (\lambda_{k} - \lambda_{r_{1}})^{-1}
            (\lambda_{k} - \lambda_{r_{2}})^{-1}
            \left<\phi_{r_{1}}, \Delta\phi_{k} \right>
            \left<\phi_{r_{2}}, \Delta\phi_{k} \right>\\
        &\qquad\times\dot{\mu}_{i}(T_{ij_{1}})
            \phi_{r_{1}}(T_{ij_{1}})\xi_{ir_{2}}
            \left\{1+ O((\log m/m)^{1/2})\right\}\\
        & \lesssim \|\Delta\phi_{k}\|^{2}
            \underset{r_{1} \neq k}{\sum_{r_{1}=1}^{\infty}}
            \underset{r_{2} \neq k}{\sum_{r_{2}=1}^{\kappa_{0}}}
            (\lambda_{k} - \lambda_{r_{1}})^{-1}
            (\lambda_{k} - \lambda_{r_{2}})^{-1}
            \left({n}^{-1}\sum_{i=1}^{n}
            \left\{\sM_{ir_{1}} + O((\log m/m)^{1/2})\right\}\xi_{ir_{2}}\right)\\
        &\qquad\times\left\{1+ O((\log m/m)^{1/2})\right\}\\
        &= \|\Delta\phi_{k}\|^{2}
            \underset{r_{1} \neq k}{\sum_{r_{1}=1}^{\infty}}
            \underset{r_{2} \neq k}{\sum_{r_{2}=1}^{\kappa_{0}}}
            (\lambda_{k} - \lambda_{r_{1}})^{-1}
            (\lambda_{k} - \lambda_{r_{2}})^{-1}
        \times\left\{
            {n}^{-1}\sum_{i=1}^{n}\sM_{ir_{1}}\xi_{ir_{2}}
        \right\}
        \left\{
            1 + O((\log m/m)^{1/2})
        \right\}.
    \numberthis
\end{align*}
\endgroup
Therefore, Inequality (\ref{Chapter2-qif-Eq:JK2_n2}) immediately follows using Lemma \ref{Chapter2-qif-Lemma:part-U2},
\begin{align*}
    J_{k2}^{n2} &= O\left( 
        (\log n/n)^{1/2}
        \overline{\delta}^{2}_{n}(h)
        \kappa_{0}^{(3-\alpha)/2}
        \lambda_{\kappa_{0}}^{-1}
        \{\sum_{r=1}^{\kappa_{0}}  \lambda_{r}\}^{1/2}
        \left\{1+ (\log m/m)^{1/2}\right\}
    \right)\\ 
    &:= O(\omega_{k2}(n, h)), 
    \qquad \text{almost surely}.
\end{align*}
where $$\omega_{k2}(n, h) = 
    (\log n/n)^{1/2}
    \overline{\delta}^{2}_{n}(h)
    \kappa_{0}^{(3-\alpha)/2}
    \lambda_{\kappa_{0}}^{-1}
    \sum_{r=1}^{\kappa_{0}}  
    \lambda_{r}\left\{1+ (\log m/m)^{1/2}\right\}.$$ We observe that $\sum_{r=1}^{\kappa_{0}}\lambda_{r} \sim \kappa_{0}^{-\tau_{1} + 1}$ under assumption (C6)b. 
Thus, 
$\omega_{k2}(n,h) \sim (\log n/n)
\overline{\delta}^{2}(h)\kappa_{0}^{(4-\tau)/2}
\{1+ (\log m/m)^{1/2} \}
$.
Since $\omega_{n2} = O(\omega_{n1})$ and $\overline{\delta}_{n}^{2}(h) = O(\omega_{n1})$, in summary, for each $k = 1, \cdots, \kappa_{0}$, 
$$\overline{\bg}^{(k)}(\bbeta) = 
    O\left( 
        (\log n/n)^{1/2}
        \left\{ 1 + (\log m/m)^{1/2}\right\}
        + \omega_{k1}(n,h)
    \right)$$
almost surely. Since $1/(nm) = O(1/n)$, 
$$\AMSE\{ \overline{\bg}^{(k)}(\bbeta_{0}) \} = O\left( n^{-1} + n^{-1}\kappa_{0}^{3-\tau}R_{n}(h)\right),$$ 
where $R_{n}(h) = \left\{
            h^{4} + \frac{1}{n} + 
            \frac{1}{nmh} + 
            \frac{1}{n^{2}m^{2}h^{2}} + 
            \frac{1}{n^{2}m^{4}h^{4}} + 
            \frac{1}{n^{2}mh} + 
            \frac{1}{n^{2}m^{3}h^{3}}
            \right\}$.
Combining the above conditions, we find that 
if $a > 1/4$, $\kappa_{0} = O(n^{1/(3-\tau)})$ 
and $n^{-1/4} \lesssim h \lesssim n^{-(a+1)/5}$ 
then $\AMSE\{\overline{\bg^{(k)}}(\bbeta_{0})\} = O(1/n)$. 
On the other hand, if $a  \leq 1/4$, $\kappa_{0} = O(n^{4(1+a)/5(3-\tau)})$ and $h \lesssim n^{-1/4}$ $\AMSE\{\overline{\bg^{(k)}}(\bbeta_{0})\} = O(1/n)$. 
\par
Note that for a three-dimensional array 
$(\partial C/\partial\beta_{1}, \cdots, \partial C/\partial\beta_{p})$ 
such that the following is a $p\times 1$ vector.
\begin{equation}
 \overline{\bg}(\bbeta_{0})^{\tp}\bC^{-1}(\bbeta_{0})\dot{\bC}(\bbeta_{0})\bC^{-1}(\bbeta_{0})\overline{\bg}(\bbeta_{0})   
\end{equation}
Therefore, 
\begin{equation}
\label{Chapter2-qif-Eq:Qdot}
    n^{-1}\dot{\sQ}(\bbeta_{0}) = 2\dot{\overline{\bg}}(\bbeta_{0})^{\tp}\bC^{-1}(\bbeta_{0})\overline{\bg}(\bbeta_{0}) - 
    \overline{\bg}(\bbeta_{0})^{\tp}\bC^{-1}(\bbeta_{0})\dot{\bC}(\bbeta_{0})\bC^{-1}(\bbeta_{0})\overline{\bg}(\bbeta_{0})
\end{equation}
and 
\begin{equation}
\label{Chapter2-qif-Eq:Qdotdot}
        n^{-1}\ddot{\sQ}(\bbeta_{0}) = 2\dot{\overline{\bg}}(\bbeta_{0})^{\tp}\bC^{-1}(\bbeta_{0})\dot{\overline{\bg}}(\bbeta_{0}) 
        + r_{n1} + r_{n2} + r_{n3} + r_{n4}
\end{equation}
where 
\begin{equation}
    \begin{split}
        r_{n1} &= 2\ddot{\overline{\bg}}(\bbeta_{0})\bC^{-1}(\bbeta_{0})\overline{\bg}(\bbeta_{0})\\
        r_{n2} &= - 4\dot{\overline{\bg}}(\bbeta_{0})^{\tp}\bC^{-1}(\bbeta_{0})\dot{\bC}(\bbeta_{0})\bC^{-1}(\bbeta_{0})\overline{\bg}(\bbeta_{0})\\
        r_{n3} &= 2\dot{\overline{\bg}}(\bbeta_{0})\bC^{-1}(\bbeta_{0})\bC^{-1}(\bbeta_{0})
        \dot{\bC}(\bbeta_{0})\bC^{-1}(\bbeta_{0})\overline{\bg}(\bbeta_{0})\\
        r_{n4} &= -\overline{\bg}(\bbeta_{0})^{\tp}\bC^{-1}(\bbeta_{0})\ddot{\bC}(\bbeta_{0})\bC^{-1}(\bbeta_{0})\overline{\bg}(\bbeta_{0}).
    \end{split}
\end{equation}
Since $\overline{\bg}(\bbeta_{0}) = O_{P}(n^{-1/2})$ and the weight matrix converges almost surely to an invertible matrix, 
$$\overline{\bg}(\bbeta_{0})^{\tp}\bC^{-1}(\bbeta_{0})\dot{\bC}(\bbeta_{0})\bC^{-1}(\bbeta_{0})\overline{\bg}(\bbeta_{0}) = o(n^{-1})$$ almost surely. 
Furthermore, $r_{n1} = O(n^{-1/2})$, $r_{n2} = o(n^{-1/2})$, $r_{n3} = O(n^{-1/2})$, and $r_{n4} = O(n^{-1})$ almost surely. Combining these bounds, we have $r_{n} = o(1)$ almost surely.
Therefore, $\| n^{-1}\dot{\sQ} - 2\dot{\overline{\bg}}(\bbeta_{0})^{\tp}\bC^{-1}(\bbeta_{0})\overline{\bg}(\bbeta_{0}) \| = o_{P}(n^{-1})$ 
and 
$\|n^{-1}\ddot{\sQ} - 2\dot{\overline{\bg}}(\bbeta_{0})^{\tp}\bC^{-1}(\bbeta_{0})\dot{\overline{\bg}}(\bbeta_{0}) \| = o_{P}(1)$.\par
The following lines are based on common steps in the GEE literature that includes \citet{mccullagh1989generalized, balan2005asymptotic, tian2014penalized} among many others. 
Let $\bbeta_{n} = \bbeta_{0} + \delta\bd$ where set $\delta = n^{-1/2}$. We have to show that for any $\epsilon > 0$ there exists a large constant $c$ such that 
\begin{equation}
\label{Chapter2-qif-Eq:prob-cond}
    P\{\inf_{\|d\| = c} \sQ(\bbeta_{n}) \geq \sQ(\bbeta_{0}) \} > 1-\epsilon.
\end{equation}
Note that the above statement is always true if $\epsilon \geq 1$. Thus, we assume that $\epsilon \in (0, 1)$. Due to Taylor series expansion, 
\begin{equation}
    \sQ(\bbeta_{n}) = \sQ(\bbeta_{0} + \delta\bd) 
    = \sQ(\bbeta_{0}) + \delta\bd^{\tp}\dot{\sQ}(\bbeta_{0}) + 0.5 \delta\bd^{\tp}\ddot{\sQ}(\bbeta_{0})\bd + \|\bd\|^{2}o_{P}(1).
\end{equation}
Now, observe that, using Equation (\ref{Chapter2-qif-Eq:Qdot}),
\begin{equation}
    \delta\bd^{\tp}\dot{\sQ}(\bbeta_{0}) = \|\bd\|O_{P}(\sqrt{n}\delta) + \|\bd\|O_{P}(\delta),
\end{equation}
and
\begin{equation}
    0.5\delta^{2}\bd^{\tp}\ddot{\sQ}(\bbeta^{*})\bd  = n\delta^{2}\bd^{\tp}\dot{\overline{\bg}}(\bbeta_{0})\bC^{-1}(\bbeta_{0})\dot{\bg}(\bbeta_{0})\bd + n\delta^{2}\|\bd\|^{2}o_{P}(1).
\end{equation}
Therefore, for given $\epsilon > 0$, there exists a large enough $c$ such that the above equation (\ref{Chapter2-qif-Eq:prob-cond}) holds. This implies that there exists a $\hat{\bbeta}$ that satisfies $\|\widehat{\bbeta} - \bbeta_{0}\| = O_{P}(\delta)$.
Thus, for large $n$, with probability 1, $\sQ(\bbeta)$ attains the minimal value at $\hat{\bbeta}$ and therefore, $\dot{\sQ} = 0$.

\section{Proof of Theorem 2}
Recall, $\sC_{i} = \sum_{k_{1}=1}^{\kappa_{0}}
                  \sum_{k_{2}=1}^{\kappa_{0}}
                  \bPhi_{k_{1}}
                  \bX_{i}\bC_{k_{1}, k_{2}}^{-1}\bX_{i}^{\tp}
                  \bPhi_{k_{2}}$, 
where $\bC_{k_{1}, k_{2}}^{-1}$ is the $(k_{1}, k_{2})$ block of $\bC_{0}^{-1}$. 
Similarly, we can define $\widehat{\sC}_{i} =             \sum_{k_{1}=1}^{\kappa_{0}}
                  \sum_{k_{2}=1}^{\kappa_{0}}
                  \widehat{\bPhi}_{k_{1}}
                  \bX_{i}\bC_{k_{1}, k_{2}}^{-1}\bX_{i}^{\tp}
                  \widehat{\bPhi}_{k_{2}}$.
It is easy to observe that $\widehat{\sC}_{i} =
        \sC_{i} + \sum_{k_{1}=1}^{\kappa_{0}}\sum_{k_{2}=1}^{\kappa_{0}}
        (\widehat{\bPhi}_{k_{1}} - \bPhi_{k_{1}})
        \bX_{i}\bC_{k_{1}, k_{2}}^{-1}\bX_{i}^{\tp}
        (\widehat{\bPhi}_{k_{2}} - \bPhi_{k_{2}}) 
        + 2\sum_{k_{1}=1}^{\kappa_{0}}\sum_{k_{2}=1}^{\kappa_{0}}
        \bPhi_{k_{1}}
        \bX_{i}\bC_{k_{1}, k_{2}}^{-1}\bX_{i}^{\tp}
        (\widehat{\bPhi}_{k_{2}} - \bPhi_{k_{2}})$.                   
Therefore, $n^{-1}\sum_{i=1}^{n}\dot{\bmu}_{i}^{\tp}(\widehat{\sC}_{i} - \sC)\bX_{i} = n^{-1}\sum_{i=1}^{n}\sum_{k_{1}=1}^{\kappa_{0}}\sum_{k_{2}=1}^{\kappa_{0}}\bP_{ik_{1}}\bC_{k_{1}, k_{2}}^{-1}\bP_{ik_{2}}$ 
and
$n^{-1}\sum_{i=1}^{n}\dot{\bmu}_{i}^{\tp}(\widehat{\sC}_{i} - \sC)\by_{i} = n^{-1}\sum_{i=1}^{n}\sum_{k_{1}=1}^{\kappa_{0}}\sum_{k_{2}=1}^{\kappa_{0}}\bP_{ik_{1}}\bC_{k_{1}, k_{2}}^{-1}\bQ_{ik_{2}}$ 
where 
$\bP_{i,k} = \dot{\bmu}_{i}^{\tp}\bD_{ik}\bX_{i}$ and $\bQ_{ik}= \dot{\bmu}_{i}^{\tp}\bD_{ik}\by_{i}$ 
with $\bD_{ik}$ be the difference matrix with $(j_{1}, j_{2})$-th element is $d_{i}(T_{ij_{1}}, T_{ij_{2}})$. 
Thus, note that, almost surely we have the following relation,  
\begingroup
\allowdisplaybreaks
\begin{align*}
    \bP_{ik} &= m_{i}^{-2}\sum_{j_{1}=1}^{m_{i}}\sum_{j_{2}=1}^{m_{i}}
        \dot{\mu}_{i}(T_{ij_{1}})d_{i}(T_{ij_{1}}, T_{ij_{2}})x_{i}(T_{ij_{2}})\\
        &\lesssim m_{i}^{-2}\sum_{j_{1}=1}^{m_{i}}\sum_{j_{2}=1}^{m_{i}}
        \dot{\mu}_{i}(T_{ij_{1}})
        \underset{r \neq k}{\sum_{r=1}^{\infty}}(\lambda_{k}-\lambda_{r})^{-1}
        \left<\phi_{r}, \Delta\phi_{k} \right>
        \left\{
            \phi_{r}(T_{ij_{1}})\phi_{k}(T_{ij_{2}}) + \phi_{k}(T_{ij_{1}})\phi_{r}(T_{ij_{2}})
        \right\}\\
        & \qquad + O(\|\Delta\|^{2})\\
        &\lesssim \underset{r\neq k}{\sum_{r=1}^{\infty}}(\lambda_{k} - \lambda_{r})^{-1}\|\Delta\phi_{k}\| + O(\|\Delta\|^{2}) =O(\varpi),\\
        & \qquad \text{ since $m_{i}^{-1}\sum_{j=1}^{m_{i}}\dot{\mu}(T_{ij})\phi_{k}(T_{ij})$ and 
        $m_{i}^{-1}\sum_{j=1}^{m_{i}}x_{i}(T_{ij})\phi_{k}(T_{ij})$ are finite,}
    \numberthis
\end{align*}
\endgroup
where $\varpi = \underset{r\neq k}{\sum_{r = 1}^{\infty}}(\lambda_{k}-\lambda_{r})^{-1} \overline{\delta}_{n}(h) + h^{2} + \delta_{n2}^{2}(h)$. 
A similar result can be obtained for $\bQ_{ik}$. Combining such results, in summary, we have 
$-{2}{n^{-1}}\sum_{i=1}^{n}\bX_{i}^{\tp}\widehat{\sC}_{i}(\by_{i} - \bX_{i}^{\tp}\widehat{\bbeta}) = 
-{2}{n^{-1}}\sum_{i=1}^{n}\bX_{i}^{\tp}\sC_{i}(\by_{i} - \bX_{i}^{\tp}\widehat{\bbeta}) + O(\varpi_{n})$.
Since, for $n \rightarrow \infty$, $\sQ(\bbeta)$ attains a minimal value at $\bbeta = \widehat{\bbeta}$, we therefore have $\dot{\sQ}(\widehat{\bbeta}) = 0$. Thus, 

\begin{align*}
&\dot{\sQ}(\widehat{\bbeta}) = 
        -{2}{n^{-1}}\sum_{i = 1}^{n}
        \bX_{i}^{\tp}\widehat{\sC}_{i}(\by_{i} - \bX_{i}^{\tp}\widehat{\bbeta})\\ 
        &= -{2}{n^{-1}}\sum_{i = 1}^{n}
        \bX_{i}^{\tp}(\widehat{\sC}_{i} - \sC_{i})(\by_{i} - \bX_{i}^{\tp}\widehat{\bbeta}) 
        -{2}{n^{-1}}\sum_{i = 1}^{n}
        \bX_{i}^{\tp}{\sC}_{i}(\by_{i} - \bX_{i}^{\tp}\widehat{\bbeta}) 
        = 0.
    \numberthis
\end{align*}
Therefore, almost surely, we have, 
\begin{align*}
    -{2}{n^{-1}}\sum_{i=1}^{n}\bX_{i}^{\tp}\sC_{i}(\by_{i} - \bX_{i}^{\tp}\widehat{\bbeta}) + O(\varpi_{n}) &= 0\\
    -{2}{n^{-1}}\sum_{i=1}^{n}\bX_{i}^{\tp}\sC_{i}(\bX_{i}^{\tp}\bbeta_{0} + \be_{i} - \bX_{i}^{\tp}\widehat{\bbeta}) + O(\varpi_{n}) &= 0\\
    \sqrt{n}(\widehat{\bbeta} - \bbeta_{0})
    \Big\{
         n^{-1}\sum_{i = 1}^{n}\bX_{i}^{\tp}\sC_{i}\bX_{i}
    \Big\} 
    &= {n^{-1}}
        \sum_{i=1}^{n}\bX_{i}^{\tp}\sC_{i}\be_{i}.
    \numberthis
\end{align*}
Now, using the central limit theorem, we can obtain the following. 
\begin{equation}
   \sum_{i=1}^{n}\bX_{i}^{\tp}\sC_{i}\be_{i}/{\sqrt{n}} \xrightarrow[]{d} N(0, \bA).
\end{equation}
In addition, by the law of large numbers $n^{-1}\sum_{i=1}^{n}\bX_{i}^{\tp}\sC_{i}\bX_{i} \rightarrow\bB$ in probability. Therefore, using the Slutsky theorem, we complete the proof of Theorem 2.

\begin{table}[t!]
\caption{Performance of the estimation procedure where the residuals are generated from linear process 
with $l_{0} = 1$. Mean of the estimated coefficients, standard deviation, absolute bias, mean square error $(\times 100)$ and FVE in percentage are summarized.}
\centering
\begin{tabular}{rrrrrrrrrr}
  \\
 Method & \multicolumn{4}{c}{$\beta_{1}$} & \multicolumn{4}{c}{$\beta_{2}$} & FVE \%-age\\ 
 \\
 & Mean & SD & AB  & MSE & Mean & SD & AB & MSE & \\
 \\
 &\multicolumn{9}{c}{$n = 100$}\\
 \\
\texttt{init} & 1.0019 & 0.0340 & 0.0267 & 0.1155 & 0.5001 & 0.0456 & 0.0356 & 0.2078 &  \\ 
  \texttt{ldaAR} & 1.0007 & 0.0230 & 0.0187 & 0.0528 & 0.5010 & 0.0361 & 0.0285 & 0.1303 &  \\ 
  \texttt{ldaCS} & 1.0000 & 0.0006 & 0.0005 & 0.0000 & 0.4999 & 0.0008 & 0.0006 & 0.0001 &  \\ 
  $\texttt{fda-1}$& 1.0093 & 0.1436 & 0.1150 & 2.0675 & 0.5020 & 0.2010 & 0.1570 & 4.0311 & 73.0607 \\ 
  $\texttt{fda-2}$ & 1.0036 & 0.1055 & 0.0828 & 1.1123 & 0.5052 & 0.1337 & 0.1070 & 1.7869 & 91.6726 \\ 
  $\texttt{fda-3}$ & 1.0038 & 0.1024 & 0.0804 & 1.0473 & 0.5056 & 0.1303 & 0.1020 & 1.6969 & 99.7657 \\ 
  $\texttt{fda-4}$ & 1.0000 & 0.0092 & 0.0021 & 0.0084 & 0.5006 & 0.0096 & 0.0024 & 0.0093 & 99.9993 \\ 
  $\texttt{fda-5}$ & 1.0000 & 0.0011 & 0.0007 & 0.0001 & 0.5000 & 0.0017 & 0.0010 & 0.0003 & 100.0000 \\ 
  $\texttt{fda-AIC}$ & 1.0026 & 0.0997 & 0.0781 & 0.9936 & 0.5052 & 0.1284 & 0.0998 & 1.6477 & 99.7824 \\ 
 \\
 &\multicolumn{9}{c}{$n = 300$}\\
 \\
 \texttt{init} & 0.9991 & 0.0181 & 0.0144 & 0.0326 & 0.5006 & 0.0268 & 0.0212 & 0.0715 &  \\ 
  \texttt{ldaAR} & 0.9995 & 0.0133 & 0.0105 & 0.0178 & 0.5000 & 0.0212 & 0.0169 & 0.0447 &  \\ 
  \texttt{ldaCS} & 1.0000 & 0.0003 & 0.0002 & 0.0000 & 0.5000 & 0.0005 & 0.0004 & 0.0000 &  \\ 
  $\texttt{fda-1}$& 0.9958 & 0.0767 & 0.0616 & 0.5888 & 0.5037 & 0.1163 & 0.0918 & 1.3523 & 73.4907 \\ 
  $\texttt{fda-2}$ & 0.9974 & 0.0567 & 0.0458 & 0.3220 & 0.5011 & 0.0804 & 0.0648 & 0.6460 & 91.7757 \\ 
  $\texttt{fda-3}$ & 0.9974 & 0.0564 & 0.0455 & 0.3182 & 0.5007 & 0.0800 & 0.0643 & 0.6391 & 99.9225 \\ 
  $\texttt{fda-4}$ & 1.0000 & 0.0005 & 0.0001 & 0.0000 & 0.5000 & 0.0009 & 0.0002 & 0.0001 & 99.9998 \\ 
  $\texttt{fda-5}$ & 1.0000 & 0.0002 & 0.0001 & 0.0000 & 0.5000 & 0.0003 & 0.0002 & 0.0000 & 100.0000 \\ 
  $\texttt{fda-AIC}$ & 0.9974 & 0.0564 & 0.0455 & 0.3182 & 0.5007 & 0.0800 & 0.0643 & 0.6391 & 99.9225 \\ 
 \\
 &\multicolumn{9}{c}{$n = 500$}\\
 \\
  \texttt{init} & 0.9999 & 0.0152 & 0.0121 & 0.0230 & 0.5027 & 0.0207 & 0.0163 & 0.0436 &  \\ 
  \texttt{ldaAR} & 1.0008 & 0.0100 & 0.0079 & 0.0100 & 0.5019 & 0.0161 & 0.0129 & 0.0263 &  \\ 
  \texttt{ldaCS} & 1.0000 & 0.0003 & 0.0002 & 0.0000 & 0.5000 & 0.0004 & 0.0003 & 0.0000 &  \\ 
  $\texttt{fda-1}$& 0.9990 & 0.0657 & 0.0523 & 0.4303 & 0.5113 & 0.0883 & 0.0698 & 0.7913 & 73.4098 \\ 
  $\texttt{fda-2}$ & 1.0013 & 0.0468 & 0.0371 & 0.2185 & 0.5078 & 0.0651 & 0.0520 & 0.4292 & 91.8501 \\ 
  $\texttt{fda-3}$ & 1.0014 & 0.0459 & 0.0364 & 0.2107 & 0.5074 & 0.0650 & 0.0516 & 0.4274 & 99.9490 \\ 
  $\texttt{fda-4}$ & 1.0000 & 0.0001 & 0.0001 & 0.0000 & 0.5000 & 0.0001 & 0.0001 & 0.0000 & 99.9999 \\ 
  $\texttt{fda-5}$ & 1.0000 & 0.0001 & 0.0001 & 0.0000 & 0.5000 & 0.0001 & 0.0001 & 0.0000 & 100.0000 \\ 
  $\texttt{fda-AIC}$ & 1.0014 & 0.0459 & 0.0364 & 0.2107 & 0.5074 & 0.0650 & 0.0516 & 0.4274 & 99.9490 \\
\end{tabular}
\label{Chapter2-qif-Table:decay1}
\end{table}

\begin{table}[t!]
\centering
\caption{Performance of the estimation procedure where the residuals are generated from linear process 
with $l_{0} = 2$. Mean of the estimated coefficients, standard deviation, absolute bias, mean square error $(\times 100)$ and FVE in percentage are summarized.}{
\begin{tabular}{rrrrrrrrrr}
  \\
 Method & \multicolumn{4}{c}{$\beta_{1}$} & \multicolumn{4}{c}{$\beta_{2}$} & FVE \%-age\\ 
 \\
 & Mean & SD & AB  & MSE & Mean & SD & AB & MSE & \\
 \\
 &\multicolumn{9}{c}{$n = 100$}\\
 \\
  \texttt{init} & 1.0020 & 0.0331 & 0.0261 & 0.1094 & 0.4999 & 0.0451 & 0.0349 & 0.2028 &  \\ 
  \texttt{ldaAR} & 1.0002 & 0.0167 & 0.0133 & 0.0278 & 0.5014 & 0.0232 & 0.0183 & 0.0538 &  \\ 
  \texttt{ldaCS} & 1.0000 & 0.0003 & 0.0002 & 0.0000 & 0.5000 & 0.0004 & 0.0003 & 0.0000 &  \\ 
  $\texttt{fda-1}$& 1.0096 & 0.1431 & 0.1144 & 2.0520 & 0.5022 & 0.2003 & 0.1561 & 4.0029 & 92.5933 \\ 
  $\texttt{fda-2}$ & 1.0014 & 0.0648 & 0.0508 & 0.4188 & 0.5037 & 0.0842 & 0.0667 & 0.7094 & 98.5532 \\ 
  $\texttt{fda-3}$ & 1.0009 & 0.0535 & 0.0406 & 0.2860 & 0.5018 & 0.0744 & 0.0570 & 0.5524 & 99.7251 \\ 
  $\texttt{fda-4}$ & 1.0000 & 0.0074 & 0.0017 & 0.0055 & 0.4999 & 0.0103 & 0.0024 & 0.0106 & 99.9991 \\ 
  $\texttt{fda-5}$ & 1.0000 & 0.0045 & 0.0009 & 0.0021 & 0.5000 & 0.0021 & 0.0009 & 0.0004 & 100.0000 \\ 
  $\texttt{fda-AIC}$ & 1.0006 & 0.0538 & 0.0409 & 0.2891 & 0.5021 & 0.0748 & 0.0573 & 0.5581 & 99.6771 \\ 
 \\
 &\multicolumn{9}{c}{$n = 300$}\\
 \\
  \texttt{init} & 0.9991 & 0.0175 & 0.0140 & 0.0308 & 0.5006 & 0.0263 & 0.0208 & 0.0689 &  \\ 
  \texttt{ldaAR} & 0.9999 & 0.0091 & 0.0071 & 0.0083 & 0.4997 & 0.0127 & 0.0102 & 0.0161 &  \\ 
  \texttt{ldaCS} & 1.0000 & 0.0002 & 0.0001 & 0.0000 & 0.5000 & 0.0002 & 0.0002 & 0.0000 &  \\ 
  $\texttt{fda-1}$& 0.9957 & 0.0767 & 0.0616 & 0.5893 & 0.5038 & 0.1164 & 0.0919 & 1.3541 & 92.9525 \\ 
  $\texttt{fda-2}$ & 0.9989 & 0.0365 & 0.0295 & 0.1331 & 0.4998 & 0.0511 & 0.0410 & 0.2608 & 98.7539 \\ 
  $\texttt{fda-3}$ & 0.9992 & 0.0349 & 0.0282 & 0.1218 & 0.4991 & 0.0483 & 0.0385 & 0.2332 & 99.9061 \\ 
  $\texttt{fda-4}$ & 0.9999 & 0.0017 & 0.0002 & 0.0003 & 0.4999 & 0.0033 & 0.0004 & 0.0011 & 99.9997 \\ 
  $\texttt{fda-5}$ & 1.0000 & 0.0002 & 0.0001 & 0.0000 & 0.5000 & 0.0002 & 0.0001 & 0.0000 & 100.0000 \\ 
  $\texttt{fda-AIC}$ & 0.9991 & 0.0349 & 0.0281 & 0.1218 & 0.4991 & 0.0485 & 0.0386 & 0.2352 & 99.8515 \\
  \\
 &\multicolumn{9}{c}{$n = 500$}\\
 \\
  \texttt{init} & 0.9998 & 0.0148 & 0.0118 & 0.0219 & 0.5026 & 0.0201 & 0.0158 & 0.0409 &  \\ 
  \texttt{ldaAR} & 1.0006 & 0.0073 & 0.0059 & 0.0054 & 0.5008 & 0.0098 & 0.0079 & 0.0097 &  \\ 
  \texttt{ldaCS} & 1.0000 & 0.0001 & 0.0001 & 0.0000 & 0.5000 & 0.0002 & 0.0001 & 0.0000 &  \\ 
  $\texttt{fda-1}$& 0.9990 & 0.0656 & 0.0523 & 0.4295 & 0.5114 & 0.0882 & 0.0698 & 0.7897 & 92.9459 \\ 
  $\texttt{fda-2}$ & 1.0013 & 0.0296 & 0.0233 & 0.0874 & 0.5039 & 0.0415 & 0.0329 & 0.1735 & 98.7957 \\ 
  $\texttt{fda-3}$ & 1.0012 & 0.0285 & 0.0224 & 0.0812 & 0.5036 & 0.0407 & 0.0322 & 0.1670 & 99.9389 \\ 
  $\texttt{fda-4}$ & 1.0000 & 0.0001 & 0.0001 & 0.0000 & 0.5000 & 0.0001 & 0.0001 & 0.0000 & 99.9998 \\ 
  $\texttt{fda-5}$ & 1.0000 & 0.0001 & 0.0001 & 0.0000 & 0.5000 & 0.0001 & 0.0001 & 0.0000 & 100.0000 \\ 
  $\texttt{fda-AIC}$ & 1.0012 & 0.0285 & 0.0224 & 0.0814 & 0.5035 & 0.0408 & 0.0323 & 0.1671 & 99.9101 \\
\end{tabular}}
\label{Chapter2-qif-Table:decay2}
\end{table}

\begin{table}[t!]
\centering
\caption{Performance of the estimation procedure where the residuals are generated from linear process 
with $l_{0} = 3$. Mean of the estimated coefficients, standard deviation, absolute bias, mean square error $(\times 100)$ and FVE in percentage are summarized.}{
\begin{tabular}{rrrrrrrrrr}
  \\
 Method & \multicolumn{4}{c}{$\beta_{1}$} & \multicolumn{4}{c}{$\beta_{2}$} & FVE \%-age\\ 
 \\
 & Mean & SD & AB  & MSE & Mean & SD & AB & MSE & \\
 \\
 &\multicolumn{9}{c}{$n = 100$}\\
 \\
\texttt{init} & 1.0020 & 0.0328 & 0.0260 & 0.1077 & 0.4998 & 0.0451 & 0.0349 & 0.2027 &  \\ 
  \texttt{ldaAR} & 1.0000 & 0.0087 & 0.0069 & 0.0076 & 0.5009 & 0.0122 & 0.0096 & 0.0149 &  \\ 
  \texttt{ldaCS} & 1.0000 & 0.0001 & 0.0001 & 0.0000 & 0.5000 & 0.0002 & 0.0002 & 0.0000 &  \\ 
  $\texttt{fda-1}$& 1.0096 & 0.1430 & 0.1143 & 2.0496 & 0.5023 & 0.2001 & 0.1559 & 3.9978 & 97.9586 \\ 
  $\texttt{fda-2}$ & 1.0000 & 0.0326 & 0.0250 & 0.1058 & 0.5023 & 0.0440 & 0.0344 & 0.1941 & 99.5699 \\ 
  $\texttt{fda-3}$ & 0.9998 & 0.0144 & 0.0072 & 0.0208 & 0.4986 & 0.0209 & 0.0103 & 0.0436 & 99.8941 \\ 
  $\texttt{fda-4}$ & 1.0002 & 0.0053 & 0.0013 & 0.0028 & 0.5002 & 0.0067 & 0.0018 & 0.0044 & 99.9991 \\ 
  $\texttt{fda-5}$ & 1.0003 & 0.0037 & 0.0008 & 0.0014 & 0.4998 & 0.0042 & 0.0011 & 0.0018 & 100.0000 \\
  $\texttt{fda-AIC}$ & 1.0003 & 0.0323 & 0.0243 & 0.1040 & 0.5024 & 0.0435 & 0.0334 & 0.1898 & 99.6253 \\ 
  \\
  &\multicolumn{9}{c}{$n = 300$}\\
  \\
  \texttt{init} & 0.9991 & 0.0174 & 0.0139 & 0.0303 & 0.5006 & 0.0262 & 0.0207 & 0.0684 &  \\ 
  \texttt{ldaAR} & 1.0000 & 0.0047 & 0.0037 & 0.0022 & 0.4997 & 0.0066 & 0.0052 & 0.0044 &  \\ 
  \texttt{ldaCS} & 1.0000 & 0.0001 & 0.0001 & 0.0000 & 0.5000 & 0.0001 & 0.0001 & 0.0000 &  \\ 
  $\texttt{fda-1}$& 0.9957 & 0.0767 & 0.0616 & 0.5896 & 0.5038 & 0.1164 & 0.0919 & 1.3543 & 98.2262 \\ 
  $\texttt{fda-2}$ & 0.9996 & 0.0197 & 0.0158 & 0.0386 & 0.4996 & 0.0274 & 0.0219 & 0.0750 & 99.7663 \\ 
  $\texttt{fda-3}$ & 1.0002 & 0.0151 & 0.0101 & 0.0227 & 0.4990 & 0.0186 & 0.0125 & 0.0347 & 99.9255 \\ 
  $\texttt{fda-4}$ & 1.0001 & 0.0022 & 0.0003 & 0.0005 & 0.5001 & 0.0017 & 0.0003 & 0.0003 & 99.9997 \\ 
  $\texttt{fda-5}$ & 1.0000 & 0.0002 & 0.0001 & 0.0000 & 0.5000 & 0.0002 & 0.0001 & 0.0000 & 100.0000 \\ 
  $\texttt{fda-AIC}$ & 0.9996 & 0.0197 & 0.0158 & 0.0386 & 0.4996 & 0.0274 & 0.0219 & 0.0750 & 99.7663 \\ 
  \\
  &\multicolumn{9}{c}{$n = 500$}\\
  \\
  \texttt{init} & 0.9998 & 0.0147 & 0.0117 & 0.0216 & 0.5025 & 0.0199 & 0.0157 & 0.0400 &  \\ 
  \texttt{ldaAR} & 1.0003 & 0.0038 & 0.0031 & 0.0015 & 0.5003 & 0.0052 & 0.0041 & 0.0027 &  \\ 
  \texttt{ldaCS} & 1.0000 & 0.0001 & 0.0001 & 0.0000 & 0.5000 & 0.0001 & 0.0001 & 0.0000 &  \\ 
  $\texttt{fda-1}$& 0.9990 & 0.0656 & 0.0523 & 0.4293 & 0.5114 & 0.0882 & 0.0698 & 0.7895 & 98.2518 \\ 
  $\texttt{fda-2}$ & 1.0008 & 0.0159 & 0.0126 & 0.0253 & 0.5015 & 0.0222 & 0.0175 & 0.0495 & 99.8021 \\ 
  $\texttt{fda-3}$ & 1.0001 & 0.0126 & 0.0091 & 0.0158 & 0.5002 & 0.0179 & 0.0130 & 0.0320 & 99.9449 \\ 
  $\texttt{fda-4}$ & 1.0000 & 0.0001 & 0.0001 & 0.0000 & 0.5000 & 0.0001 & 0.0001 & 0.0000 & 99.9998 \\ 
  $\texttt{fda-5}$ & 1.0000 & 0.0001 & 0.0001 & 0.0000 & 0.5000 & 0.0001 & 0.0001 & 0.0000 & 100.0000 \\ 
  $\texttt{fda-AIC}$ & 1.0008 & 0.0159 & 0.0126 & 0.0253 & 0.5015 & 0.0222 & 0.0175 & 0.0495 & 99.8021 \\
\end{tabular}}
\label{Chapter2-qif-Table:decay3}
\end{table}

\begin{table}[t!]
\caption{Performance of the estimation procedure where the residuals are generated from Ornstein-Uhlenbeck process 
with $\mu_{0} = 1$. Mean of the estimated coefficients, standard deviation, absolute bias, mean square error $(\times 100)$ and FVE in percentage are summarized.}
\centering
\begin{tabular}{rrrrrrrrrr}
  \\
 Method & \multicolumn{4}{c}{$\beta_{1}$} & \multicolumn{4}{c}{$\beta_{2}$} & FVE \%-age\\ 
 \\
 & Mean & SD & AB  & MSE & Mean & SD & AB & MSE & \\
 \\
 &\multicolumn{9}{c}{$n = 100$}\\
 \\
  \texttt{init} & 1.0003 & 0.0541 & 0.0434 & 0.2922 & 0.4994 & 0.0711 & 0.0563 & 0.5048 &  \\ 
  \texttt{ldaAR} & 1.0001 & 0.0476 & 0.0383 & 0.2261 & 0.4984 & 0.0650 & 0.0513 & 0.4214 &  \\ 
  \texttt{ldaCS} & 0.9994 & 0.0398 & 0.0316 & 0.1581 & 0.4978 & 0.0534 & 0.0421 & 0.2849 &  \\ 
  $\texttt{fda-1}$& 1.0009 & 0.0705 & 0.0563 & 0.4964 & 0.5006 & 0.0947 & 0.0747 & 0.8954 & 79.4156 \\ 
  $\texttt{fda-2}$ & 1.0001 & 0.0453 & 0.0358 & 0.2044 & 0.4982 & 0.0608 & 0.0482 & 0.3690 & 94.9669 \\ 
  $\texttt{fda-3}$ & 0.9993 & 0.0386 & 0.0307 & 0.1491 & 0.4978 & 0.0511 & 0.0405 & 0.2613 & 99.9949 \\ 
  $\texttt{fda-4}$ & 1.0003 & 0.0127 & 0.0081 & 0.0162 & 0.4991 & 0.0242 & 0.0146 & 0.0583 & 99.9992 \\ 
  $\texttt{fda-5}$ & 0.9997 & 0.0084 & 0.0071 & 0.0071 & 0.4991 & 0.0159 & 0.0124 & 0.0254 & 100.0000 \\ 
  $\texttt{fda-AIC}$ & 0.9993 & 0.0386 & 0.0307 & 0.1491 & 0.4978 & 0.0511 & 0.0405 & 0.2613 & 99.9949 \\ 
  \\
  &\multicolumn{9}{c}{$n = 300$}\\
 \\
 \texttt{init} & 1.0000 & 0.0288 & 0.0233 & 0.0829 & 0.5003 & 0.0416 & 0.0329 & 0.1728 &  \\ 
  \texttt{ldaAR} & 1.0000 & 0.0258 & 0.0206 & 0.0662 & 0.4996 & 0.0368 & 0.0294 & 0.1350 &  \\ 
  \texttt{ldaCS} & 1.0002 & 0.0212 & 0.0170 & 0.0449 & 0.4987 & 0.0314 & 0.0254 & 0.0983 &  \\ 
  $\texttt{fda-1}$& 0.9997 & 0.0388 & 0.0316 & 0.1503 & 0.5014 & 0.0560 & 0.0440 & 0.3130 & 79.9233 \\ 
  $\texttt{fda-2}$ & 1.0005 & 0.0240 & 0.0193 & 0.0576 & 0.4983 & 0.0352 & 0.0284 & 0.1242 & 95.0283 \\ 
  $\texttt{fda-3}$ & 0.9999 & 0.0202 & 0.0161 & 0.0409 & 0.4994 & 0.0298 & 0.0240 & 0.0884 & 99.9984 \\ 
  $\texttt{fda-4}$ & 0.9999 & 0.0072 & 0.0064 & 0.0051 & 0.4997 & 0.0129 & 0.0111 & 0.0166 & 99.9998 \\ 
  $\texttt{fda-5}$ & 1.0000 & 0.0072 & 0.0064 & 0.0051 & 0.4997 & 0.0128 & 0.0111 & 0.0165 & 100.0000 \\ 
  $\texttt{fda-AIC}$ & 0.9999 & 0.0202 & 0.0161 & 0.0409 & 0.4994 & 0.0298 & 0.0240 & 0.0884 & 99.9984 \\ 
  \\
  &\multicolumn{9}{c}{$n = 500$}\\
 \\ 
 \texttt{init} & 1.0000 & 0.0288 & 0.0233 & 0.0829 & 0.5003 & 0.0416 & 0.0329 & 0.1728 &  \\ 
  \texttt{ldaAR} & 1.0000 & 0.0258 & 0.0206 & 0.0662 & 0.4996 & 0.0368 & 0.0294 & 0.1350 &  \\ 
  \texttt{ldaCS} & 1.0002 & 0.0212 & 0.0170 & 0.0449 & 0.4987 & 0.0314 & 0.0254 & 0.0983 &  \\ 
  $\texttt{fda-1}$& 0.9997 & 0.0388 & 0.0316 & 0.1503 & 0.5014 & 0.0560 & 0.0440 & 0.3130 & 79.9233 \\ 
  $\texttt{fda-2}$ & 1.0005 & 0.0240 & 0.0193 & 0.0576 & 0.4983 & 0.0352 & 0.0284 & 0.1242 & 95.0283 \\ 
  $\texttt{fda-3}$ & 0.9999 & 0.0202 & 0.0161 & 0.0409 & 0.4994 & 0.0298 & 0.0240 & 0.0884 & 99.9984 \\ 
  $\texttt{fda-4}$ & 0.9999 & 0.0072 & 0.0064 & 0.0051 & 0.4997 & 0.0129 & 0.0111 & 0.0166 & 99.9998 \\ 
  $\texttt{fda-5}$ & 1.0000 & 0.0072 & 0.0064 & 0.0051 & 0.4997 & 0.0128 & 0.0111 & 0.0165 & 100.0000 \\ 
  $\texttt{fda-AIC}$ & 0.9999 & 0.0161 & 0.0127 & 0.0259 & 0.5001 & 0.0226 & 0.0178 & 0.0509 & 99.9991 \\ 
\end{tabular}
\label{Chapter2-qif-Table:ou1}
\end{table}

\begin{table}[t!]
\caption{Performance of the estimation procedure where the residuals are generated from Ornstein-Uhlenbeck process 
with $\mu_{0} = 3$. Mean of the estimated coefficients, standard deviation, absolute bias, mean square error $(\times 100)$ and FVE in percentage are summarized.}
\centering
\begin{tabular}{rrrrrrrrrr}
  \\
 Method & \multicolumn{4}{c}{$\beta_{1}$} & \multicolumn{4}{c}{$\beta_{2}$} & FVE \%-age\\ 
 \\
 & Mean & SD & AB  & MSE & Mean & SD & AB & MSE & \\
 \\
 &\multicolumn{9}{c}{$n = 100$}\\
 \\
 \texttt{init} & 1.0001 & 0.0454 & 0.0363 & 0.2056 & 0.4990 & 0.0591 & 0.0469 & 0.3487 &  \\ 
  \texttt{ldaAR} & 0.9996 & 0.0390 & 0.0312 & 0.1521 & 0.4979 & 0.0521 & 0.0414 & 0.2717 &  \\ 
  \texttt{ldaCS} & 0.9999 & 0.0429 & 0.0341 & 0.1841 & 0.4981 & 0.0564 & 0.0448 & 0.3176 &  \\ 
  $\texttt{fda-1}$& 1.0005 & 0.0557 & 0.0445 & 0.3100 & 0.5003 & 0.0743 & 0.0588 & 0.5511 & 59.1692 \\ 
  $\texttt{fda-2}$ & 1.0005 & 0.0459 & 0.0362 & 0.2107 & 0.4981 & 0.0604 & 0.0480 & 0.3639 & 87.0120 \\ 
  $\texttt{fda-3}$ & 1.0000 & 0.0436 & 0.0348 & 0.1894 & 0.4975 & 0.0568 & 0.0455 & 0.3221 & 99.9908 \\ 
  $\texttt{fda-4}$ & 1.0004 & 0.0136 & 0.0058 & 0.0185 & 0.4989 & 0.0237 & 0.0106 & 0.0562 & 99.9960 \\ 
  $\texttt{fda-5}$ & 1.0001 & 0.0049 & 0.0033 & 0.0024 & 0.4997 & 0.0099 & 0.0061 & 0.0098 & 100.0000 \\ 
  $\texttt{fda-AIC}$ & 1.0000 & 0.0436 & 0.0348 & 0.1894 & 0.4975 & 0.0568 & 0.0455 & 0.3221 & 99.9908 \\ 
 \\
 &\multicolumn{9}{c}{$n = 300$}\\
 \\
 \texttt{init} & 1.0002 & 0.0239 & 0.0191 & 0.0568 & 0.4998 & 0.0345 & 0.0274 & 0.1190 &  \\ 
  \texttt{ldaAR} & 1.0003 & 0.0209 & 0.0167 & 0.0435 & 0.4990 & 0.0303 & 0.0244 & 0.0916 &  \\ 
  \texttt{ldaCS} & 1.0003 & 0.0224 & 0.0178 & 0.0500 & 0.4992 & 0.0328 & 0.0265 & 0.1075 &  \\ 
  $\texttt{fda-1}$& 0.9998 & 0.0305 & 0.0247 & 0.0928 & 0.5011 & 0.0443 & 0.0349 & 0.1957 & 59.4418 \\ 
  $\texttt{fda-2}$ & 1.0004 & 0.0240 & 0.0192 & 0.0574 & 0.4991 & 0.0347 & 0.0279 & 0.1201 & 86.9290 \\ 
  $\texttt{fda-3}$ & 1.0003 & 0.0222 & 0.0177 & 0.0492 & 0.4992 & 0.0327 & 0.0264 & 0.1070 & 99.9972 \\ 
  $\texttt{fda-4}$ & 1.0002 & 0.0043 & 0.0039 & 0.0018 & 0.4998 & 0.0075 & 0.0065 & 0.0057 & 99.9987 \\ 
  $\texttt{fda-5}$ & 1.0001 & 0.0038 & 0.0035 & 0.0014 & 0.4998 & 0.0068 & 0.0059 & 0.0046 & 100.0000 \\ 
  $\texttt{fda-AIC}$ & 1.0003 & 0.0222 & 0.0177 & 0.0492 & 0.4992 & 0.0327 & 0.0264 & 0.1070 & 99.9972 \\ 
  \\
  &\multicolumn{9}{c}{$n = 500$}\\
 \\
 \texttt{init} & 1.0003 & 0.0180 & 0.0141 & 0.0323 & 0.5017 & 0.0268 & 0.0213 & 0.0720 &  \\ 
  \texttt{ldaAR} & 1.0001 & 0.0155 & 0.0123 & 0.0241 & 0.5016 & 0.0232 & 0.0186 & 0.0541 &  \\ 
  \texttt{ldaCS} & 1.0002 & 0.0171 & 0.0134 & 0.0291 & 0.5013 & 0.0251 & 0.0202 & 0.0629 &  \\ 
  $\texttt{fda-1}$& 1.0005 & 0.0229 & 0.0180 & 0.0523 & 0.5021 & 0.0346 & 0.0276 & 0.1201 & 59.3104 \\ 
  $\texttt{fda-2}$ & 1.0004 & 0.0177 & 0.0138 & 0.0314 & 0.5022 & 0.0268 & 0.0217 & 0.0724 & 87.0183 \\ 
  $\texttt{fda-3}$ & 1.0001 & 0.0173 & 0.0136 & 0.0300 & 0.5010 & 0.0249 & 0.0197 & 0.0620 & 99.9983 \\ 
  $\texttt{fda-4}$ & 1.0000 & 0.0042 & 0.0038 & 0.0017 & 0.5005 & 0.0075 & 0.0065 & 0.0056 & 99.9992 \\ 
  $\texttt{fda-5}$ & 1.0000 & 0.0038 & 0.0035 & 0.0015 & 0.5004 & 0.0066 & 0.0058 & 0.0044 & 100.0000 \\ 
  $\texttt{fda-AIC}$ & 1.0001 & 0.0173 & 0.0136 & 0.0300 & 0.5010 & 0.0249 & 0.0197 & 0.0620 & 99.9983 \\
\end{tabular}
\label{Chapter2-qif-Table:ou3}
\end{table}

\begin{table}[t!]
\caption{Performance of the estimation procedure where the residuals are generated with power exponential covariance function 
where scale parameter $a_{0} = 1$ and shape parameter $b_{0} = 1$. Mean of the estimated coefficients, standard deviation, absolute bias, mean square error $(\times 100)$ and FVE in percentage are summarized.}
\centering
\begin{tabular}{rrrrrrrrrr}
  \\
 Method & \multicolumn{4}{c}{$\beta_{1}$} & \multicolumn{4}{c}{$\beta_{2}$} & FVE \%-age\\ 
 \\
 & Mean & SD & AB  & MSE & Mean & SD & AB & MSE & \\
 \\
 &\multicolumn{9}{c}{$n = 100$}\\
 \\
 \texttt{init} & 0.9968 & 0.0525 & 0.0423 & 0.2758 & 0.4961 & 0.0755 & 0.0603 & 0.5705 &  \\ 
  \texttt{ldaAR} & 0.9985 & 0.0486 & 0.0387 & 0.2361 & 0.4962 & 0.0702 & 0.0562 & 0.4938 &  \\ 
  \texttt{ldaCS} & 0.9978 & 0.0389 & 0.0309 & 0.1514 & 0.4986 & 0.0549 & 0.0438 & 0.3013 &  \\ 
  $\texttt{fda-1}$& 0.9960 & 0.0708 & 0.0564 & 0.5018 & 0.4951 & 0.1010 & 0.0813 & 1.0195 & 73.1399 \\ 
  $\texttt{fda-2}$ & 0.9975 & 0.0439 & 0.0347 & 0.1929 & 0.4980 & 0.0629 & 0.0496 & 0.3948 & 87.5389 \\ 
  $\texttt{fda-3}$ & 0.9975 & 0.0381 & 0.0305 & 0.1453 & 0.4987 & 0.0531 & 0.0425 & 0.2811 & 92.2253 \\ 
  $\texttt{fda-4}$ & 0.9978 & 0.0392 & 0.0311 & 0.1540 & 0.4977 & 0.0540 & 0.0434 & 0.2914 & 94.4211 \\ 
  $\texttt{fda-5}$ & 0.9978 & 0.0388 & 0.0302 & 0.1508 & 0.4975 & 0.0527 & 0.0420 & 0.2773 & 95.6881 \\ 
  $\texttt{fda-6}$ & 0.9979 & 0.0393 & 0.0306 & 0.1546 & 0.4977 & 0.0534 & 0.0424 & 0.2852 & 96.5184 \\ 
  $\texttt{fda-7}$ & 0.9990 & 0.0396 & 0.0308 & 0.1570 & 0.4986 & 0.0535 & 0.0423 & 0.2857 & 97.0882 \\ 
  $\texttt{fda-8}$ & 0.9985 & 0.0405 & 0.0317 & 0.1637 & 0.4981 & 0.0540 & 0.0429 & 0.2915 & 97.5117 \\ 
  $\texttt{fda-AIC}$ & 0.9977 & 0.0454 & 0.0356 & 0.2064 & 0.4969 & 0.0608 & 0.0488 & 0.3697 & 99.0414 \\ 
  \\
 &\multicolumn{9}{c}{$n = 300$}\\
 \\
  \texttt{init} & 0.9975 & 0.0297 & 0.0236 & 0.0885 & 0.4989 & 0.0428 & 0.0352 & 0.1832 &  \\ 
  \texttt{ldaAR} & 0.9972 & 0.0281 & 0.0224 & 0.0793 & 0.4992 & 0.0389 & 0.0316 & 0.1509 &  \\ 
  \texttt{ldaCS} & 0.9988 & 0.0226 & 0.0180 & 0.0513 & 0.4995 & 0.0300 & 0.0238 & 0.0899 &  \\ 
  $\texttt{fda-1}$& 0.9967 & 0.0391 & 0.0310 & 0.1539 & 0.4986 & 0.0588 & 0.0482 & 0.3456 & 73.4971 \\ 
  $\texttt{fda-2}$ & 0.9986 & 0.0254 & 0.0203 & 0.0648 & 0.4990 & 0.0335 & 0.0266 & 0.1122 & 87.5190 \\ 
  $\texttt{fda-3}$ & 0.9987 & 0.0221 & 0.0173 & 0.0490 & 0.4996 & 0.0293 & 0.0233 & 0.0859 & 92.1425 \\ 
  $\texttt{fda-4}$ & 0.9986 & 0.0219 & 0.0171 & 0.0479 & 0.4997 & 0.0294 & 0.0233 & 0.0862 & 94.3152 \\ 
  $\texttt{fda-5}$ & 0.9988 & 0.0213 & 0.0167 & 0.0456 & 0.5000 & 0.0287 & 0.0232 & 0.0823 & 95.5616 \\ 
  $\texttt{fda-6}$ & 0.9989 & 0.0213 & 0.0166 & 0.0454 & 0.5001 & 0.0289 & 0.0232 & 0.0832 & 96.3760 \\ 
  $\texttt{fda-7}$ & 0.9988 & 0.0212 & 0.0166 & 0.0449 & 0.5003 & 0.0291 & 0.0234 & 0.0843 & 96.9441 \\ 
  $\texttt{fda-8}$ & 0.9988 & 0.0212 & 0.0166 & 0.0449 & 0.5001 & 0.0292 & 0.0236 & 0.0852 & 97.3641 \\ 
  $\texttt{fda-AIC}$ & 0.9986 & 0.0218 & 0.0172 & 0.0478 & 0.5005 & 0.0301 & 0.0241 & 0.0904 & 99.0306 \\ 
  \\
  &\multicolumn{9}{c}{$n = 500$}\\
 \\
  \texttt{init} & 0.9996 & 0.0212 & 0.0169 & 0.0450 & 0.4989 & 0.0316 & 0.0252 & 0.0999 &  \\ 
  \texttt{ldaAR} & 0.9993 & 0.0189 & 0.0151 & 0.0356 & 0.4983 & 0.0294 & 0.0237 & 0.0868 &  \\ 
  \texttt{ldaCS} & 0.9996 & 0.0170 & 0.0135 & 0.0288 & 0.4998 & 0.0235 & 0.0193 & 0.0553 &  \\ 
  $\texttt{fda-1}$& 0.9996 & 0.0278 & 0.0222 & 0.0773 & 0.4984 & 0.0422 & 0.0333 & 0.1777 & 73.6172 \\ 
  $\texttt{fda-2}$ & 0.9992 & 0.0186 & 0.0150 & 0.0347 & 0.4995 & 0.0264 & 0.0216 & 0.0694 & 87.5609 \\ 
  $\texttt{fda-3}$ & 1.0002 & 0.0164 & 0.0130 & 0.0268 & 0.5001 & 0.0225 & 0.0183 & 0.0505 & 92.1497 \\ 
  $\texttt{fda-4}$ & 1.0001 & 0.0163 & 0.0129 & 0.0264 & 0.4999 & 0.0226 & 0.0184 & 0.0508 & 94.3084 \\ 
  $\texttt{fda-5}$ & 1.0000 & 0.0160 & 0.0127 & 0.0255 & 0.4996 & 0.0223 & 0.0182 & 0.0496 & 95.5492 \\ 
  $\texttt{fda-6}$ & 1.0000 & 0.0161 & 0.0128 & 0.0260 & 0.4995 & 0.0222 & 0.0182 & 0.0493 & 96.3576 \\ 
  $\texttt{fda-7}$ & 0.9999 & 0.0161 & 0.0128 & 0.0258 & 0.4994 & 0.0219 & 0.0179 & 0.0481 & 96.9249 \\ 
  $\texttt{fda-8}$ & 0.9999 & 0.0161 & 0.0128 & 0.0258 & 0.4994 & 0.0219 & 0.0177 & 0.0478 & 97.3423 \\ 
  $\texttt{fda-AIC}$ & 1.0001 & 0.0164 & 0.0129 & 0.0267 & 0.4991 & 0.0220 & 0.0180 & 0.0485 & 99.0283 \\ 
\end{tabular}
\label{Chapter2-qif-Table:pe1}
\end{table}

\begin{table}[t!]
\caption{Performance of the estimation procedure where the residuals are generated with power exponential covariance function 
where scale parameter $a_{0} = 1$ and shape parameter $b_{0} = 2$. Mean of the estimated coefficients, standard deviation, absolute bias, mean square error $(\times 100)$ and FVE in percentage are summarized.}
\centering
\begin{tabular}{rrrrrrrrrr}
  \\
 Method & \multicolumn{4}{c}{$\beta_{1}$} & \multicolumn{4}{c}{$\beta_{2}$} & FVE \%-age\\ 
 \\
 & Mean & SD & AB  & MSE & Mean & SD & AB & MSE & \\
 \\
 &\multicolumn{9}{c}{$n = 100$}\\
 \\
  \texttt{init} & 0.9967 & 0.0563 & 0.0454 & 0.3170 & 0.4957 & 0.0811 & 0.0649 & 0.6591 &  \\ 
  \texttt{ldaAR} & 0.9980 & 0.0506 & 0.0399 & 0.2562 & 0.4970 & 0.0724 & 0.0578 & 0.5234 &  \\ 
  \texttt{ldaCS} & 0.9982 & 0.0373 & 0.0294 & 0.1389 & 0.4983 & 0.0531 & 0.0425 & 0.2814 &  \\ 
  $\texttt{fda-1}$& 0.9957 & 0.0766 & 0.0611 & 0.5872 & 0.4947 & 0.1094 & 0.0881 & 1.1964 & 85.6976 \\ 
  $\texttt{fda-2}$ & 0.9979 & 0.0440 & 0.0348 & 0.1934 & 0.4975 & 0.0633 & 0.0503 & 0.4010 & 98.9964 \\ 
  $\texttt{fda-3}$ & 0.9993 & 0.0239 & 0.0187 & 0.0569 & 0.4988 & 0.0343 & 0.0273 & 0.1177 & 99.9585 \\ 
  $\texttt{fda-4}$ & 0.9984 & 0.0224 & 0.0176 & 0.0505 & 0.5004 & 0.0305 & 0.0242 & 0.0931 & 99.9987 \\ 
  $\texttt{fda-5}$ & 0.9984 & 0.0224 & 0.0176 & 0.0505 & 0.5004 & 0.0305 & 0.0242 & 0.0931 & 100.0000 \\ 
 $\texttt{fda-AIC}$ & 0.9981 & 0.0347 & 0.0263 & 0.1207 & 0.4999 & 0.0480 & 0.0373 & 0.2303 & 99.5245 \\ 
  \\
  &\multicolumn{9}{c}{$n = 300$}\\
 \\
 \texttt{init} & 0.9973 & 0.0318 & 0.0253 & 0.1014 & 0.4987 & 0.0461 & 0.0378 & 0.2118 &  \\ 
  \texttt{ldaAR} & 0.9972 & 0.0297 & 0.0239 & 0.0888 & 0.4992 & 0.0400 & 0.0323 & 0.1597 &  \\ 
  \texttt{ldaCS} & 0.9990 & 0.0216 & 0.0172 & 0.0468 & 0.4995 & 0.0289 & 0.0229 & 0.0835 &  \\ 
  $\texttt{fda-1}$& 0.9964 & 0.0424 & 0.0336 & 0.1804 & 0.4984 & 0.0637 & 0.0521 & 0.4047 & 86.0991 \\ 
  $\texttt{fda-2}$ & 0.9987 & 0.0253 & 0.0201 & 0.0642 & 0.4990 & 0.0336 & 0.0266 & 0.1129 & 99.0443 \\ 
  $\texttt{fda-3}$ & 0.9992 & 0.0146 & 0.0114 & 0.0213 & 0.5003 & 0.0197 & 0.0159 & 0.0388 & 99.9593 \\ 
  $\texttt{fda-4}$ & 0.9993 & 0.0128 & 0.0100 & 0.0165 & 0.5000 & 0.0176 & 0.0139 & 0.0309 & 99.9987 \\ 
  $\texttt{fda-5}$ & 0.9993 & 0.0128 & 0.0100 & 0.0165 & 0.5000 & 0.0176 & 0.0139 & 0.0309 & 100.0000 \\ 
  $\texttt{fda-AIC}$ & 0.9992 & 0.0219 & 0.0168 & 0.0481 & 0.4995 & 0.0289 & 0.0226 & 0.0836 & 99.3822 \\ 

  \\
  &\multicolumn{9}{c}{$n = 500$}\\
  \\
  \texttt{init} & 0.9994 & 0.0226 & 0.0180 & 0.0512 & 0.4987 & 0.0340 & 0.0270 & 0.1153 &  \\ 
  \texttt{ldaAR} & 0.9993 & 0.0201 & 0.0161 & 0.0403 & 0.4983 & 0.0309 & 0.0250 & 0.0954 &  \\ 
  \texttt{ldaCS} & 0.9993 & 0.0163 & 0.0130 & 0.0266 & 0.4995 & 0.0227 & 0.0186 & 0.0517 &  \\ 
  $\texttt{fda-1}$& 0.9995 & 0.0301 & 0.0240 & 0.0906 & 0.4983 & 0.0456 & 0.0361 & 0.2081 & 86.1915 \\ 
  $\texttt{fda-2}$ & 0.9990 & 0.0186 & 0.0150 & 0.0346 & 0.4992 & 0.0265 & 0.0217 & 0.0699 & 99.0585 \\ 
  $\texttt{fda-3}$ & 1.0004 & 0.0111 & 0.0087 & 0.0122 & 0.5002 & 0.0148 & 0.0121 & 0.0219 & 99.9596 \\ 
  $\texttt{fda-4}$ & 1.0007 & 0.0100 & 0.0080 & 0.0100 & 0.5007 & 0.0129 & 0.0106 & 0.0168 & 99.9987 \\ 
  $\texttt{fda-5}$ & 1.0007 & 0.0100 & 0.0080 & 0.0100 & 0.5007 & 0.0129 & 0.0106 & 0.0168 & 100.0000 \\ 
  $\texttt{fda-AIC}$ & 0.9995 & 0.0170 & 0.0135 & 0.0287 & 0.5000 & 0.0234 & 0.0189 & 0.0546 & 99.2948 \\ 
\end{tabular}
\label{Chapter2-qif-Table:pe2}
\end{table}

\begin{table}[t!]
\caption{Performance of the estimation procedure where the residuals are generated with power exponential covariance function 
where scale parameter $a_{0} = 1$ and shape parameter $b_{0} = 5$. Mean of the estimated coefficients, standard deviation, absolute bias, mean square error $(\times 100)$ and FVE in percentage are summarized.}
\centering
\begin{tabular}{rrrrrrrrrr}
  \\
 Method & \multicolumn{4}{c}{$\beta_{1}$} & \multicolumn{4}{c}{$\beta_{2}$} & FVE \%-age\\ 
 \\
 & Mean & SD & AB  & MSE & Mean & SD & AB & MSE & \\
 \\
 &\multicolumn{9}{c}{$n = 100$}\\
 \\
 \texttt{init} & 0.9970 & 0.0582 & 0.0468 & 0.3390 & 0.4952 & 0.0841 & 0.0676 & 0.7089 &  \\ 
  \texttt{ldaAR} & 0.9977 & 0.0528 & 0.0415 & 0.2789 & 0.4961 & 0.0772 & 0.0621 & 0.5956 &  \\ 
  \texttt{ldaCS} & 0.9996 & 0.0274 & 0.0218 & 0.0747 & 0.4979 & 0.0396 & 0.0321 & 0.1572 &  \\ 
  $\texttt{fda-1}$& 0.9956 & 0.0810 & 0.0646 & 0.6564 & 0.4943 & 0.1157 & 0.0933 & 1.3399 & 92.2949 \\ 
  $\texttt{fda-2}$ & 0.9995 & 0.0375 & 0.0296 & 0.1400 & 0.4966 & 0.0548 & 0.0441 & 0.3013 & 99.7252 \\ 
  $\texttt{fda-3}$ & 0.9993 & 0.0283 & 0.0203 & 0.0801 & 0.5004 & 0.0383 & 0.0274 & 0.1462 & 99.8787 \\ 
  $\texttt{fda-4}$ & 0.9991 & 0.0121 & 0.0053 & 0.0147 & 0.4995 & 0.0165 & 0.0073 & 0.0273 & 99.9461 \\ 
  $\texttt{fda-5}$ & 1.0001 & 0.0116 & 0.0022 & 0.0134 & 0.5005 & 0.0146 & 0.0030 & 0.0212 & 99.9842 \\ 
  $\texttt{fda-AIC}$ & 0.9997 & 0.0371 & 0.0294 & 0.1375 & 0.4966 & 0.0548 & 0.0441 & 0.3012 & 99.7266 \\ 
  \\
 &\multicolumn{9}{c}{$n = 300$}\\
 \\
 \texttt{init} & 0.9972 & 0.0328 & 0.0260 & 0.1082 & 0.4986 & 0.0482 & 0.0395 & 0.2317 &  \\ 
  \texttt{ldaAR} & 0.9971 & 0.0310 & 0.0251 & 0.0968 & 0.4994 & 0.0426 & 0.0344 & 0.1810 &  \\ 
  \texttt{ldaCS} & 0.9995 & 0.0156 & 0.0123 & 0.0244 & 0.4996 & 0.0216 & 0.0171 & 0.0467 &  \\ 
  $\texttt{fda-1}$& 0.9962 & 0.0449 & 0.0356 & 0.2027 & 0.4982 & 0.0673 & 0.0552 & 0.4524 & 92.6741 \\ 
  $\texttt{fda-2}$ & 0.9990 & 0.0213 & 0.0168 & 0.0454 & 0.4994 & 0.0291 & 0.0229 & 0.0846 & 99.8076 \\ 
  $\texttt{fda-3}$ & 0.9995 & 0.0196 & 0.0153 & 0.0384 & 0.4992 & 0.0271 & 0.0213 & 0.0736 & 99.9391 \\ 
  $\texttt{fda-4}$ & 1.0005 & 0.0075 & 0.0044 & 0.0056 & 0.4994 & 0.0105 & 0.0064 & 0.0111 & 99.9716 \\ 
  $\texttt{fda-5}$ & 0.9999 & 0.0023 & 0.0006 & 0.0005 & 0.5000 & 0.0037 & 0.0010 & 0.0014 & 99.9889 \\ 
  $\texttt{fda-AIC}$ & 0.9990 & 0.0213 & 0.0168 & 0.0454 & 0.4994 & 0.0291 & 0.0229 & 0.0846 & 99.8076 \\ 
  \\
 &\multicolumn{9}{c}{$n = 500$}\\
 \\
  \texttt{init} & 0.9994 & 0.0232 & 0.0185 & 0.0539 & 0.4985 & 0.0352 & 0.0280 & 0.1242 &  \\ 
  \texttt{ldaAR} & 0.9989 & 0.0209 & 0.0169 & 0.0437 & 0.4981 & 0.0327 & 0.0264 & 0.1073 &  \\ 
  \texttt{ldaCS} & 0.9991 & 0.0119 & 0.0095 & 0.0142 & 0.4992 & 0.0168 & 0.0134 & 0.0281 &  \\ 
  $\texttt{fda-1}$& 0.9995 & 0.0318 & 0.0254 & 0.1011 & 0.4981 & 0.0483 & 0.0382 & 0.2328 & 92.7586 \\ 
  $\texttt{fda-2}$ & 0.9988 & 0.0158 & 0.0127 & 0.0252 & 0.4988 & 0.0227 & 0.0183 & 0.0517 & 99.8272 \\ 
  $\texttt{fda-3}$ & 0.9989 & 0.0153 & 0.0123 & 0.0235 & 0.4989 & 0.0218 & 0.0176 & 0.0478 & 99.9574 \\ 
  $\texttt{fda-4}$ & 0.9997 & 0.0070 & 0.0050 & 0.0049 & 0.5002 & 0.0105 & 0.0072 & 0.0110 & 99.9811 \\ 
  $\texttt{fda-5}$ & 1.0000 & 0.0024 & 0.0007 & 0.0006 & 0.5001 & 0.0035 & 0.0011 & 0.0012 & 99.9923 \\ 
  $\texttt{fda-AIC}$ & 0.9988 & 0.0158 & 0.0127 & 0.0252 & 0.4988 & 0.0227 & 0.0183 & 0.0517 & 99.8272 \\
\end{tabular}
\label{Chapter2-qif-Table:pe5}
\end{table}

\begin{table}[t!]
\caption{Performance of the estimation procedure where the residuals are generated with rational quadratic covariance function 
where scale parameter $a_{0} = 1$ and shape parameter $b_{0} = 1$. Mean of the estimated coefficients, standard deviation, absolute bias, mean square error $(\times 100)$ and FVE in percentage are summarized.}
\centering
\begin{tabular}{rrrrrrrrrr}
  \\
 Method & \multicolumn{4}{c}{$\beta_{1}$} & \multicolumn{4}{c}{$\beta_{2}$} & FVE \%-age\\ 
 \\
 & Mean & SD & AB  & MSE & Mean & SD & AB & MSE & \\
 \\
 &\multicolumn{9}{c}{$n = 100$}\\
 \\
 \texttt{init} & 0.9967 & 0.0565 & 0.0455 & 0.3193 & 0.4957 & 0.0814 & 0.0651 & 0.6624 &  \\ 
  \texttt{ldaAR} & 0.9982 & 0.0507 & 0.0399 & 0.2567 & 0.4968 & 0.0725 & 0.0580 & 0.5257 &  \\ 
  \texttt{ldaCS} & 0.9983 & 0.0349 & 0.0275 & 0.1215 & 0.4986 & 0.0495 & 0.0396 & 0.2444 &  \\ 
  $\texttt{fda-1}$& 0.9957 & 0.0774 & 0.0617 & 0.5994 & 0.4946 & 0.1105 & 0.0890 & 1.2213 & 87.2296 \\ 
  $\texttt{fda-2}$ & 0.9980 & 0.0413 & 0.0326 & 0.1706 & 0.4979 & 0.0594 & 0.0471 & 0.3521 & 98.4366 \\ 
  $\texttt{fda-3}$ & 0.9989 & 0.0265 & 0.0210 & 0.0704 & 0.4988 & 0.0379 & 0.0302 & 0.1434 & 99.8285 \\ 
  $\texttt{fda-4}$ & 0.9985 & 0.0271 & 0.0212 & 0.0733 & 0.4991 & 0.0372 & 0.0296 & 0.1383 & 99.9800 \\ 
  $\texttt{fda-5}$ & 0.9987 & 0.0259 & 0.0203 & 0.0671 & 0.4987 & 0.0350 & 0.0276 & 0.1224 & 99.9977 \\ 
  $\texttt{fda-AIC}$ & 0.9990 & 0.0267 & 0.0212 & 0.0713 & 0.4990 & 0.0381 & 0.0303 & 0.1446 & 99.8211 \\ 
  \\
  &\multicolumn{9}{c}{$n = 300$}\\
  \\
  \texttt{init} & 0.9973 & 0.0319 & 0.0253 & 0.1021 & 0.4987 & 0.0463 & 0.0381 & 0.2144 &  \\ 
  \texttt{ldaAR} & 0.9971 & 0.0297 & 0.0239 & 0.0889 & 0.4992 & 0.0404 & 0.0327 & 0.1630 &  \\ 
  \texttt{ldaCS} & 0.9991 & 0.0203 & 0.0161 & 0.0412 & 0.4995 & 0.0271 & 0.0215 & 0.0736 &  \\ 
  $\texttt{fda-1}$& 0.9964 & 0.0429 & 0.0339 & 0.1846 & 0.4984 & 0.0643 & 0.0527 & 0.4132 & 87.6273 \\ 
  $\texttt{fda-2}$ & 0.9988 & 0.0238 & 0.0190 & 0.0568 & 0.4991 & 0.0317 & 0.0251 & 0.1002 & 98.4851 \\ 
  $\texttt{fda-3}$ & 0.9991 & 0.0159 & 0.0125 & 0.0255 & 0.5002 & 0.0215 & 0.0173 & 0.0462 & 99.8287 \\ 
  $\texttt{fda-4}$ & 0.9991 & 0.0157 & 0.0121 & 0.0248 & 0.5001 & 0.0214 & 0.0171 & 0.0457 & 99.9800 \\ 
  $\texttt{fda-5}$ & 0.9993 & 0.0144 & 0.0113 & 0.0209 & 0.5005 & 0.0202 & 0.0164 & 0.0408 & 99.9977 \\ 
  $\texttt{fda-AIC}$ & 0.9991 & 0.0159 & 0.0125 & 0.0255 & 0.5002 & 0.0215 & 0.0173 & 0.0462 & 99.8287 \\ 
  \\
  &\multicolumn{9}{c}{$n = 500$}\\
  \\
  \texttt{init} & 0.9995 & 0.0227 & 0.0181 & 0.0514 & 0.4987 & 0.0341 & 0.0271 & 0.1160 &  \\ 
  \texttt{ldaAR} & 0.9993 & 0.0200 & 0.0160 & 0.0399 & 0.4982 & 0.0310 & 0.0250 & 0.0962 &  \\ 
  \texttt{ldaCS} & 0.9994 & 0.0153 & 0.0122 & 0.0235 & 0.4996 & 0.0213 & 0.0174 & 0.0454 &  \\ 
  $\texttt{fda-1}$& 0.9995 & 0.0304 & 0.0243 & 0.0924 & 0.4982 & 0.0461 & 0.0365 & 0.2125 & 87.7225 \\ 
  $\texttt{fda-2}$ & 0.9991 & 0.0176 & 0.0142 & 0.0310 & 0.4993 & 0.0249 & 0.0204 & 0.0620 & 98.5022 \\ 
  $\texttt{fda-3}$ & 1.0003 & 0.0121 & 0.0096 & 0.0146 & 0.5002 & 0.0163 & 0.0133 & 0.0265 & 99.8293 \\ 
  $\texttt{fda-4}$ & 1.0005 & 0.0120 & 0.0096 & 0.0145 & 0.5004 & 0.0160 & 0.0131 & 0.0256 & 99.9801 \\ 
  $\texttt{fda-5}$ & 1.0002 & 0.0114 & 0.0090 & 0.0129 & 0.5000 & 0.0151 & 0.0121 & 0.0227 & 99.9977 \\ 
  $\texttt{fda-AIC}$ & 1.0003 & 0.0121 & 0.0096 & 0.0146 & 0.5002 & 0.0163 & 0.0133 & 0.0265 & 99.8293 \\ 
\end{tabular}
\label{Chapter2-qif-Table:quad1}
\end{table}

\begin{table}[t!]
\caption{Performance of the estimation procedure where the residuals are generated with rational quadratic covariance function 
where scale parameter $a_{0} = 1$ and shape parameter $b_{0} = 2$. Mean of the estimated coefficients, standard deviation, absolute bias, mean square error $(\times 100)$ and FVE in percentage are summarized.}
\centering
\begin{tabular}{rrrrrrrrrr}
  \\
 Method & \multicolumn{4}{c}{$\beta_{1}$} & \multicolumn{4}{c}{$\beta_{2}$} & FVE \%-age\\ 
 \\
 & Mean & SD & AB  & MSE & Mean & SD & AB & MSE & \\
 \\
 &\multicolumn{9}{c}{$n = 100$}\\
 \\
 \texttt{init} & 0.9967 & 0.0547 & 0.0441 & 0.2993 & 0.4960 & 0.0788 & 0.0628 & 0.6206 &  \\ 
  \texttt{ldaAR} & 0.9983 & 0.0498 & 0.0393 & 0.2473 & 0.4968 & 0.0714 & 0.0571 & 0.5091 &  \\ 
  \texttt{ldaCS} & 0.9977 & 0.0425 & 0.0337 & 0.1805 & 0.4981 & 0.0602 & 0.0481 & 0.3618 &  \\ 
  $\texttt{fda-1}$& 0.9957 & 0.0730 & 0.0583 & 0.5343 & 0.4951 & 0.1042 & 0.0839 & 1.0868 & 78.3365 \\ 
  $\texttt{fda-2}$ & 0.9972 & 0.0479 & 0.0381 & 0.2297 & 0.4975 & 0.0687 & 0.0544 & 0.4721 & 96.4030 \\ 
  $\texttt{fda-3}$ & 0.9981 & 0.0368 & 0.0293 & 0.1358 & 0.4982 & 0.0520 & 0.0418 & 0.2699 & 99.4758 \\ 
  $\texttt{fda-4}$ & 0.9982 & 0.0377 & 0.0299 & 0.1424 & 0.4979 & 0.0528 & 0.0421 & 0.2791 & 99.9264 \\ 
  $\texttt{fda-5}$ & 0.9986 & 0.0356 & 0.0280 & 0.1266 & 0.4975 & 0.0483 & 0.0384 & 0.2335 & 99.9902 \\ 
  $\texttt{fda-AIC}$ & 0.9981 & 0.0368 & 0.0293 & 0.1358 & 0.4982 & 0.0520 & 0.0418 & 0.2699 & 99.4758 \\ 
  \\
 &\multicolumn{9}{c}{$n = 300$}\\
 \\
  \texttt{init} & 0.9974 & 0.0309 & 0.0246 & 0.0959 & 0.4988 & 0.0445 & 0.0365 & 0.1977 &  \\ 
  \texttt{ldaAR} & 0.9972 & 0.0290 & 0.0233 & 0.0848 & 0.4991 & 0.0393 & 0.0317 & 0.1542 &  \\ 
  \texttt{ldaCS} & 0.9987 & 0.0246 & 0.0195 & 0.0606 & 0.4994 & 0.0327 & 0.0259 & 0.1066 &  \\ 
  $\texttt{fda-1}$& 0.9966 & 0.0403 & 0.0320 & 0.1633 & 0.4986 & 0.0607 & 0.0497 & 0.3682 & 78.7132 \\ 
  $\texttt{fda-2}$ & 0.9985 & 0.0276 & 0.0220 & 0.0764 & 0.4989 & 0.0364 & 0.0290 & 0.1327 & 96.4278 \\ 
  $\texttt{fda-3}$ & 0.9985 & 0.0218 & 0.0171 & 0.0477 & 0.4999 & 0.0290 & 0.0232 & 0.0842 & 99.4738 \\ 
  $\texttt{fda-4}$ & 0.9985 & 0.0218 & 0.0171 & 0.0478 & 0.4998 & 0.0293 & 0.0235 & 0.0858 & 99.9259 \\ 
  $\texttt{fda-5}$ & 0.9988 & 0.0198 & 0.0157 & 0.0394 & 0.5004 & 0.0271 & 0.0224 & 0.0733 & 99.9900 \\
  $\texttt{fda-AIC}$ & 0.9985 & 0.0218 & 0.0171 & 0.0477 & 0.4999 & 0.0290 & 0.0232 & 0.0842 & 99.4738 \\ 
  \\
 &\multicolumn{9}{c}{$n = 500$}\\
 \\
 \texttt{init} & 0.9995 & 0.0221 & 0.0176 & 0.0490 & 0.4989 & 0.0330 & 0.0263 & 0.1086 &  \\ 
  \texttt{ldaAR} & 0.9993 & 0.0197 & 0.0158 & 0.0389 & 0.4983 & 0.0302 & 0.0245 & 0.0910 &  \\ 
  \texttt{ldaCS} & 0.9994 & 0.0184 & 0.0146 & 0.0339 & 0.4996 & 0.0257 & 0.0211 & 0.0658 &  \\ 
  $\texttt{fda-1}$& 0.9996 & 0.0288 & 0.0229 & 0.0826 & 0.4984 & 0.0435 & 0.0344 & 0.1892 & 78.8101 \\ 
  $\texttt{fda-2}$ & 0.9991 & 0.0201 & 0.0162 & 0.0405 & 0.4993 & 0.0286 & 0.0235 & 0.0819 & 96.4486 \\ 
  $\texttt{fda-3}$ & 1.0003 & 0.0162 & 0.0128 & 0.0262 & 0.5000 & 0.0222 & 0.0181 & 0.0491 & 99.4752 \\ 
  $\texttt{fda-4}$ & 1.0003 & 0.0162 & 0.0129 & 0.0263 & 0.4999 & 0.0224 & 0.0182 & 0.0501 & 99.9261 \\ 
  $\texttt{fda-5}$ & 0.9999 & 0.0151 & 0.0120 & 0.0227 & 0.4994 & 0.0208 & 0.0170 & 0.0434 & 99.9900 \\ 
  $\texttt{fda-AIC}$ & 1.0003 & 0.0162 & 0.0128 & 0.0262 & 0.5000 & 0.0222 & 0.0181 & 0.0491 & 99.4752 \\
\end{tabular}
\label{Chapter2-qif-Table:quad2}
\end{table}

\begin{table}[t!]
\caption{Performance of the estimation procedure where the residuals are generated with rational quadratic covariance function 
where scale parameter $a_{0} = 1$ and shape parameter $b_{0} = 5$. Mean of the estimated coefficients, standard deviation, absolute bias, mean square error $(\times 100)$ and FVE in percentage are summarized.}
\centering
\begin{tabular}{rrrrrrrrrr}
  \\
 Method & \multicolumn{4}{c}{$\beta_{1}$} & \multicolumn{4}{c}{$\beta_{2}$} & FVE \%-age\\ 
 \\
 & Mean & SD & AB  & MSE & Mean & SD & AB & MSE & \\
 \\
 &\multicolumn{9}{c}{$n = 100$}\\
 \\
  \texttt{init} & 0.9967 & 0.0504 & 0.0407 & 0.2545 & 0.4965 & 0.0725 & 0.0577 & 0.5251 &  \\ 
  \texttt{ldaAR} & 0.9986 & 0.0466 & 0.0369 & 0.2173 & 0.4965 & 0.0671 & 0.0539 & 0.4508 &  \\ 
  \texttt{ldaCS} & 0.9970 & 0.0473 & 0.0377 & 0.2238 & 0.4978 & 0.0668 & 0.0531 & 0.4458 &  \\ 
  $\texttt{fda-1}$& 0.9959 & 0.0648 & 0.0517 & 0.4205 & 0.4960 & 0.0922 & 0.0742 & 0.8505 & 62.1253 \\ 
  $\texttt{fda-2}$ & 0.9964 & 0.0505 & 0.0405 & 0.2561 & 0.4976 & 0.0721 & 0.0572 & 0.5199 & 89.2976 \\ 
  $\texttt{fda-3}$ & 0.9970 & 0.0462 & 0.0368 & 0.2136 & 0.4974 & 0.0644 & 0.0513 & 0.4142 & 97.4903 \\ 
  $\texttt{fda-4}$ & 0.9981 & 0.0464 & 0.0363 & 0.2149 & 0.4957 & 0.0647 & 0.0520 & 0.4201 & 99.4795 \\ 
  $\texttt{fda-5}$ & 0.9986 & 0.0441 & 0.0345 & 0.1941 & 0.4952 & 0.0621 & 0.0497 & 0.3872 & 99.9012 \\  
  $\texttt{fda-AIC}$ & 0.9981 & 0.0464 & 0.0363 & 0.2149 & 0.4957 & 0.0647 & 0.0520 & 0.4201 & 99.4795 \\ 
  \\
 &\multicolumn{9}{c}{$n = 300$}\\
 \\
 \texttt{init} & 0.9977 & 0.0286 & 0.0229 & 0.0820 & 0.4991 & 0.0406 & 0.0332 & 0.1646 &  \\ 
  \texttt{ldaAR} & 0.9975 & 0.0271 & 0.0215 & 0.0737 & 0.4991 & 0.0367 & 0.0296 & 0.1346 &  \\ 
  \texttt{ldaCS} & 0.9983 & 0.0272 & 0.0216 & 0.0739 & 0.4994 & 0.0362 & 0.0288 & 0.1309 &  \\ 
  $\texttt{fda-1}$& 0.9972 & 0.0355 & 0.0283 & 0.1264 & 0.4991 & 0.0539 & 0.0441 & 0.2896 & 62.2661 \\ 
  $\texttt{fda-2}$ & 0.9983 & 0.0289 & 0.0229 & 0.0835 & 0.4990 & 0.0384 & 0.0308 & 0.1473 & 89.2161 \\ 
  $\texttt{fda-3}$ & 0.9981 & 0.0266 & 0.0209 & 0.0712 & 0.4992 & 0.0355 & 0.0282 & 0.1257 & 97.4664 \\ 
  $\texttt{fda-4}$ & 0.9981 & 0.0260 & 0.0205 & 0.0677 & 0.4992 & 0.0351 & 0.0278 & 0.1227 & 99.4727 \\ 
  $\texttt{fda-5}$ & 0.9984 & 0.0244 & 0.0192 & 0.0594 & 0.4996 & 0.0332 & 0.0267 & 0.1097 & 99.8990 \\ 
  $\texttt{fda-AIC}$ & 0.9981 & 0.0260 & 0.0205 & 0.0677 & 0.4992 & 0.0351 & 0.0278 & 0.1227 & 99.4727 \\ 
  \\
 &\multicolumn{9}{c}{$n = 500$}\\
 \\
 \texttt{init} & 0.9996 & 0.0207 & 0.0165 & 0.0430 & 0.4992 & 0.0304 & 0.0245 & 0.0921 &  \\ 
  \texttt{ldaAR} & 0.9993 & 0.0187 & 0.0151 & 0.0351 & 0.4985 & 0.0281 & 0.0229 & 0.0792 &  \\ 
  \texttt{ldaCS} & 0.9996 & 0.0201 & 0.0160 & 0.0404 & 0.4997 & 0.0282 & 0.0230 & 0.0792 &  \\ 
  $\texttt{fda-1}$& 0.9997 & 0.0256 & 0.0204 & 0.0654 & 0.4988 & 0.0385 & 0.0304 & 0.1484 & 62.3377 \\ 
  $\texttt{fda-2}$ & 0.9992 & 0.0210 & 0.0168 & 0.0440 & 0.4995 & 0.0299 & 0.0244 & 0.0892 & 89.2460 \\ 
  $\texttt{fda-3}$ & 1.0001 & 0.0193 & 0.0153 & 0.0372 & 0.4999 & 0.0270 & 0.0219 & 0.0728 & 97.4696 \\ 
  $\texttt{fda-4}$ & 0.9997 & 0.0188 & 0.0150 & 0.0355 & 0.4993 & 0.0269 & 0.0219 & 0.0724 & 99.4726 \\ 
  $\texttt{fda-5}$ & 0.9994 & 0.0180 & 0.0142 & 0.0322 & 0.4989 & 0.0260 & 0.0211 & 0.0674 & 99.8988 \\ 
  $\texttt{fda-AIC}$ & 0.9997 & 0.0188 & 0.0150 & 0.0355 & 0.4993 & 0.0269 & 0.0219 & 0.0724 & 99.4726 \\
\end{tabular}
\label{Chapter2-qif-Table:quad5}
\end{table}

\bibliographystyle{myjmva}
\bibliography{master-reference.bib}

\end{document}